\documentclass[11pt,reqno, final]{amsart}
\usepackage[foot]{amsaddr}
\usepackage{amsfonts} 
\usepackage{amsmath}
\usepackage{amsthm}
\usepackage{amscd}
\usepackage{amsfonts}
\usepackage{amssymb}
\usepackage{mathrsfs}
\usepackage{graphicx, wrapfig}
\usepackage[top=1.25in, bottom=1.25in, left=1.25in, right=1.25in]{geometry}
\usepackage{natbib}
\usepackage{caption}
\usepackage{subcaption}
\usepackage{longtable}
\usepackage{dsfont}
\usepackage{wrapfig}
\usepackage{booktabs}
\usepackage{comment}
\usepackage{showlabels}
\usepackage{enumerate}
\usepackage{hyperref}

\DeclareMathOperator*{\argmax}{arg\,max}

\newcommand{\R}{\mathbb{R}}
\newcommand{\N}{\mathbb{N}}
\newcommand{\bvarphi}{\boldsymbol{\varphi}}
\newcommand{\bu}{\mathbf{u}}
\newcommand{\bx}{\mathbf{x}}

\newtheoremstyle{mytheorem}
{6pt}
{6pt}
{\itshape}
{-0pt}
{\large \scshape}
{}
{1em}
{}

\newtheoremstyle{myremark}
{6pt}
{10pt}
{\rm}
{-0pt}
{\large \scshape}
{}
{1em}
{}

\theoremstyle{mytheorem}
\newtheorem{Theorem}{Theorem}[section]
\newtheorem{Definition}[Theorem]{Definition}
\newtheorem{Proposition}[Theorem]{Proposition}

\newtheorem{Corollary}[Theorem]{Corollary}
\newtheorem{Hypothesis}[Theorem]{Assumption}

\theoremstyle{myremark}
\newtheorem{Remark}[Theorem]{Remark}

\newcommand{\balpha}{\boldsymbol{\alpha}}

\bibpunct{(}{)}{; }{a}{,}{,}
\setcitestyle{authoryear}    
\begin{document}

\title[An integral transformation approach to differential games]{An integral transformation approach to differential games: a climate model application} 

\author[R. Boucekkine]{Raouf Boucekkine$^1$}
\address{$^1$Aix-Marseille School of Economics, Aix-Marseille University and CNRS, Marseille France.}
\email{raouf.boucekkine@univ-amu.fr}
\author[G. Fabbri]{Giorgio Fabbri$^2$}
\address{$^2$Univ. Grenoble Alpes, CNRS, INRA, Grenoble INP, GAEL, Grenoble, France.}
\email{giorgio.fabbri@univ-grenoble-alpes.fr}
\author[S. Federico]{Salvatore Federico$^3$}
\address{$^3$Dipartimento di Matematica, Università di Bologna, Bologna, Italy.}
\email{s.federico@unibo.it}
\author[F. Gozzi]{Fausto Gozzi$^4$}
\address{$^4$Dipartimento di Economia e Finanza, Libera Universit\`{a} Internazionale degli Studi Sociali ``Guido Carli'', Rome, Italy.}
\email{fgozzi@luiss.it}
\author[T. Loch-Temzelides]{Ted Loch-Temzelides$^5$}
\address{$^5$Rice University, Houston, USA.}
\email{tedt@rice.edu}
\author[C. Ricci]{Cristiano Ricci$^6$}
\address{$^6$Universita di Pisa, Italy}
\email{cristiano.ricci@unipi.it}

\date{\today}

\begin{abstract} 
We develop an Integral Transformation Method (ITM) for the study of suitable optimal control and differential game models. This allows for a solution to such dynamic problems to be found through solving a family of optimization problems parametrized by time. The method is quite flexible, and it can be used in several economic applications where the state equation and the objective functional are linear in a state variable. We illustrate the ITM in the context of a two-country integrated assessment climate model. We characterize emissions, consumption, transfers, and welfare by computing the Nash equilibria of the associated dynamic game. We then compare them to efficiency benchmarks. Further, we apply the ITM in a robust control setup, where we investigate how (deep) uncertainty affects climate outcomes. 
\end{abstract}

\hspace{-1cm}

\maketitle









\thanks{\tiny We thank participants at the 2024 Conference on Dynamic Games and Applications at Paris, the 2024 INFER Conference at La Sapienza University of Rome,  the 2022 Viennese Conference on Optimal Control and Dynamic Games at TU Wien, the 2023 Political Economics of Environmental Sustainability Conference at Stanford University, the AMAMEF conference 2023 at Bielefeld, the 2023 IFORS congress at Santiago del Chile, the 2023 Workshop DATA at Università di Torino, the 2023 Viennese Vintage Workshop at TU Wien, the FIRE conference 2023 in Lyon, the Time-Space Evolution of Economic Activities workshop in Rome 2023, the 2024 UMI Climate Group, at Luiss, Roma, the 2024 SPOC Workshop at Politecnico di Milano, the 2024 Ecolysm workshop at Università di Roma La Sapienza, the 2024 Workshop on Analysis, control and inverse problems for evolution equations arising in climate sciences at Università di Roma Tor Vergata, the PET 2024 conference in Lyon, the EURO 2024 conference in Copenhagen, the FAERE 2024 conference in Strasbourg, the ASSET 2024 conference in Venice, the seminars at Université Bretagne Occidentale (FR), UDEC (Chile), XJLTU (China), Fudan (China), GERAD-HEC (Canada), AMSE (FR) Naples (Italy) for comments and suggestions.  The work of Giorgio Fabbri is partially supported by the French National Research Agency in the framework of the ``Investissements d'avenir'' program (ANR-15-IDEX-02) and of the center of excellence LABEX MME-DII (ANR-11-LBX-0023-01).}



\newpage

\section{Introduction}

\label{sec:introduction}
Differential games are subject to restrictive linearity assumptions which are considered necessary for analytical tractability. However, such assumptions have far-reaching economic implications, as they a priori eliminate economically meaningful non-linear effects, and can lead to implausible sets of equilibria. In this paper, we develop a convenient method for solving a class of optimal control and differential game models that relaxes some of these assumptions. This method allows for a solution to such dynamic problems to be found through solving a family of “temporary” optimization problems parametrized by time. The method is quite flexible, and it can be used in several economic applications where the state equation and the objective functional are linear in the state. No additional linearity restrictions are needed. We illustrate the solution method in the context of a two-country integrated assessment climate model. This allows us to characterize emissions, consumption, transfers, and welfare by computing the Nash equilibria of the associated dynamic game and comparing them to efficient benchmarks. 

Anthropogenic climate change is unprecedented. Its effects will be experienced over a long horizon, and are expected to affect different regions in disparate ways. Consequently, incorporating time, heterogeneity, and (deep) uncertainty considerations is at least as important in climate economics as in any other economics subfield. Differential games have been the
framework of choice in many early modeling attempts in the field in the 1990s and
early 2000s. However, certain limitations have led to this approach being less used in the more recent literature, which often also abstracts from dynamic considerations altogether. \cite{Harstad2012}, for example, points to
the economic implications of restrictive linear-quadratic (LQ) specifications.\footnote{See, for example, \cite{Dutta2004, Dutta2006, Dutta2009}.} Such assumptions can rule out economically meaningful non-linear
effects, and often lead to implausible equilibria driven by
bang-bang solutions. They can also result in multiplicity of equilibria, often leading researchers to concentrate on a (generally unique) linear Markov Perfect Equilibrium (MPE) even though more efficient non-linear MPE might co-exist. 
Studies related to the analytical tractability of differential games in the optimization literature trace back to the 1980s; see, for example, \cite{dockner1985} for an early contribution and \cite{zaccour2005} for a more recent discussion. In the former paper, conditions on the associated Hamiltonian system are stated for the differential system to obtain explicit or implicit-form solutions. In this paper, we build on this literature to develop a novel and tractable method for solving a certain class of differential games.  

We build on the game-theoretic setup in \cite{dockner1985}. Two important properties of their setup are that the corresponding Hamiltonians are linear in the state variables, and that the second-order cross-derivatives with respect to the state and control variables are nil (additive separability). Importantly, they allow for non-linearity in the control variables, thus reducing the scope for bang-bang solutions. Our main contribution involves the development of a solution method for the associated class of dynamic decision problems and differential games. Unlike \cite{dockner1985} and \cite{zaccour2005}, our approach does not employ the Pontryagin principle. Rather, using the linearity in the state, our method transforms the original intertemporal optimization problem into an equivalent class of \emph{temporary} optimization problems in which, at each time, the decision-makers take into account the marginal future contributions of the evolution of the state variables to their intertemporal payoffs.\footnote{Here we use the term ``temporary" in parallel with the concept of \emph{temporary equilibrium} in \cite{grandmont1977}.} Thus, despite the fact that the optimization involved is ``temporary" at each time $t$, the objective function at that date incorporates the full expected future dynamics of the state variables. We shall refer to this approach as the \emph{``Integral Transformation Method" (ITM)}. Our method has several mathematical and economic advantages over standard Hamiltonian or dynamic programming approaches, including the ability to conveniently handle necessary and sufficient conditions, discontinuities, infinIte-dimensional variables, non-linear-quadratic specifications, heterogeneity, and non-autonomous features, such as exogenous technological progress. We will discuss these features in some detail after the formal treatment of the ITM in the next section.\footnote{Our methodology does not speak to the issue of multiplicity of equilibria. The analysis of repeated games displaying a unique symmetric MPE often uses specific linearity and separability assumptions; see, for example, \cite{Harstad2012}. Differential games that are linear in the state variables can offer additional insights regarding nonlinear MPEs; 
see, for example, \cite{zaccour2005}, who study the credibility of equilibrium strategies.}

We illustrate the ITM in the context of a dynamic analytical integrated assessment model. Following the general approach in \cite{nordhaus2003warming} and \cite{nordhaus2018evolution}, the model accounts for climate damages created by economic activity.\footnote{See, for example, \cite{weyant2017some} and the references therein for a review and \cite{traeger2023ace} for a more recent discussion.} More precisely, we build on the simplified version of Nordhaus's three-reservoir model developed in \cite{Golosov2014}. We will impose a utility damage specification that is linear in the stock of greenhouse gas emissions (GHG), which will be the key state variable in the model. In order to investigate strategic interactions, we introduce a two-country extension, which we interpret as a stylized model of interactions between a ``global north" and a ``global south." We incorporate a general nonlinear abatement technology that allows for a reduction of GHG in the atmosphere. This captures, for example, investing in reforestation, as well as in carbon-capture technologies that can directly affect the GHG stock. We then investigate the role of a variety of transfer schemes, including technological transfers. We assume that input use in production creates a flow of GHG emissions, and the accumulated GHG emissions damage each country's payoff function. A country-specific climate sensitivity parameter is used to capture factors that can make it more vulnerable to climate change due to, for example, geography, or the ability to engage in adaptation. Each country chooses an abatement effort towards reducing the stock of
GHG emissions. As with the climate sensitivity parameter, the abatement technology can capture reforestation efforts, carbon capture and storage systems, etc. We model the interactions between North and South through transfers and standard nonlinear catching-up equations. We consider different transfers between the two countries, including transfers that can improve the abatement technology. The model can accommodate several kinds of heterogeneity, including in preferences, time discount rates, and damages resulting from the stock of accumulated GHG. 

When applied to this framework, our solution method allows us to reduce the computation of the Nash equilibria of the dynamic game to the solution of temporary games indexed by time. Open-loop Nash equilibria computed by our method are also MPE. The uniqueness of MPE in the class of affine feedbacks is also discussed. To obtain comparisons between equilibrium outcomes and the efficient frontier, we distinguish between a social planner problem without country sovereignty constraints, where a ``global planner" can relocate production from one country to another, as well as the more realistic case of a ``restricted planner," who is subject to a resource constraint for each country. In the special case of logarithmic utility, linear production function, and a non-linear abatement function, we derive various comparisons between the equilibrium and the efficient values of variables of interest, such as consumption, abatement effort, and transfers between the two countries. We then use a numerical example to illustrate the role of heterogeneity in time-discounting and climate vulnerability between the two countries, as well as the role of the intertemporal elasticity of substitution. Finally, we demonstrate how the ITM can be applied in a robust control framework; see, for example, \cite{hansen2008robustness}, as well as in a game-theoretic framework in order to investigate the effects of uncertainty on various non-cooperative equilibrium outcomes. We find that when the marginal abatement efficiency gains are small relative to the marginal emissions created by production, it is not efficient to subsidize abatement in the global south. Under logarithmic payoffs we find that in the Nash equilibrium there is over-consumption both in the global north and (provided that technological differences between the two are not too large) in the global south. The global south receives lower abatement-technology transfers and under-invests in abatement relative to the social optimum. Both global emissions and welfare are lower as a result. Our numerical example points to some interesting implications of heterogeneous climate vulnerability. If the global south is more vulnerable to climate-related damages, then the
north emits more than the south in the Nash outcome. In contrast, in the symmetric case, the emissions in the south are higher than those in the north. Total emissions are higher in the asymmetric case, pointing to the need for improving the south’s ability to adapt. As in the case without model
uncertainty, non-cooperative equilibria fall short of both planners’ solutions when model-uncertainty is introduced. However, in
the presence of model uncertainty, the planners’ and the non-cooperative solutions are closer, as more cautious behavior provides a form of insurance towards adverse climate outcomes.

Two final points are worth mentioning. Recent climate modeling in economics assumes a simplified cumulative emissions equation; see, for example, \cite{dietz2019cumulative}, \cite{hansel2020climate}, and \cite{dietz2021economists}. Some authors consider simplified dynamics of the GHG stock with a temperature law of motion as in \cite{vosooghi2022self}. It is worth emphasizing here that our modeling contribution and qualitative findings do not depend on the details of the climate model employed, and we use the climate modeling in \cite{Golosov2014} as an illustration. Lastly, provided that the main linearity and separability in the state assumptions remain in place, the ITM can be applied to discrete-time settings.

The paper proceeds as follows. After a brief literature review, Section 2 contains a formal treatment of the ITM. Section 3 introduces the application of the ITM to an analytical integrated assessment model. Section 4 studies the non-cooperative outcomes and two normative benchmarks, while Section 5 investigates a numerical example. In Section 6 we introduce Knightian uncertainty and apply the ITM in the context of robust control. A brief conclusion follows. The Appendices contain the details of the proofs, as well as additional findings derived for special cases of interest.

\smallskip

\subsection{Relation to the Literature}

There is extensive literature on international climate and environmental
agreements. Early papers include \cite{ploeg1992}, \cite{long1992pollution}, \cite{Tahvonen1994}, \cite{Xepapadeas1995}, and \cite{Hoel1997}. More recent surveys include, for example, \cite{Kolstad2005}, and \cite{Aldy2009}. 
\cite{zaccour2010} provide a detailed survey focusing on
the differential game approach to pollution control. \footnote{Another body of environmental differential games is devoted to natural resources; see, for example, \cite{dockner1989} and, more recently, \cite{colombo2019}.}

\cite{long1992pollution} investigated transnational emissions using a differential game between two countries. As in our model, the open-loop Nash equilibrium is time-consistent although in general not subgame perfect. \cite{Hoel1997} considered a dynamic game with asymmetric countries and global
emissions, and studied the feasibility of an international (uniform) carbon tax.\footnote{See \cite{insley2019climate} for a more recent reference.} \cite{Tahvonen1994} concentrated on numerical simulations and emphasized the need to balance short-term versus long-term costs and benefits in the evaluation of
climate agreements. In contrast, several recent papers abstract from dynamics; see, for example, \cite{Weitzman2017}, and \cite{McEvoy2018}.

Following \cite{Tahvonen1994} and \cite{Xepapadeas1995}, our work studies dynamic international climate
policy in a model where countries are heterogenous. As in \cite{Xepapadeas1995}, we will emphasize
the technological divergence across countries. Our framework is closer to the one in \cite{Tahvonen1994}. He
builds a differential game based on the DICE model; see \cite{nordhaus2018evolution}, which includes different geopolitical regions. To obtain closed-form solutions, \cite{Tahvonen1994} assumes that each region's objective is
linear in both temperature and the global stock of
emissions. In contrast to \cite{Xepapadeas1995}, Tahvonen does not consider endogenous technical progress or
optimally derived regional abatement efforts. His main finding, obtained from
simulations of his model, is that the cooperative
solution\footnote{%
The cooperative solution corresponds to the efficient solution when
all regions are given the same weight. We will follow the same convention in
what follows.} is beneficial for developing countries, but it is more costly
for the developed ones compared with the non-cooperative Nash solution. The study of
political processes related to climate agreements typically concentrates on
the design of efficient international climate negotiation schemes; see, for
example, \cite{Dutta2009} and \cite{Harstad2012, Harstad2016}. Other studies
emphasize coalition-formation and the stability of coalitions participating
in international agreements; see, for example, \cite{deZeeuw2008}, or related voting schemes; see \cite{Weitzman2015, Weitzman2017}. However, since \cite{Xepapadeas1995}, few researchers have
focused on the role of technological asymmetries between countries on
climate agreements.\footnote{Another stream in the differential game literature concerns the study of deforestation when the North corresponds to a set of nations
who wish to have as much tropical forest as possible, while the South has to arbitrage between the exploitation of its forests (timber production) and agricultural
activities. For example, \cite{zaccour2004} consider subsidy schemes, which are similar in spirit to the transfer policies we study in our paper.}

Of special note are several influential papers by Harstad, emphasizing different strategic and dynamic aspects of climate negotiations. \cite{Harstad2007} studies conditions under which side payments across countries are efficient in internalizing externalities. \cite{Harstad2012a} considers the possibility of countries purchasing other
countries' ``dirty" assets (or the right to develop such assets). 
In a dynamic setup closer to ours, \cite{Harstad2012} studies a
discrete-time dynamic game where players contribute to the provision of
public goods when contracts are incomplete. The paper assumes
linear abatement investment costs and additive benefits of technology. Under full commitment, the first-best can be
implemented. If countries cannot contract on their investments, a ``hold-up" problem emerges because if one country develops a better technology for cutting emissions,
it will be expected to pay a higher share of the burden to reduce collective
emissions in the future. 
\cite{Harstad2016} develops an intertemporal framework in which countries
pollute and can invest in green technologies. The basic model involves a
dynamic version of the common-pool problem. Without a climate treaty, the
countries over-pollute and invest too little. Short-term emission-reduction agreements can reduce participant
countries' payoffs since countries will under-invest if they anticipate
future negotiations. \cite{Chen2022} study a dynamic game version of Nordhaus's Regional Integrated model of Climate and the Economy (RICE). Their analysis is based on simulations used to compare Nash versus efficient outcomes. \cite{vosooghi2022self} use an integrated assessment model with heterogeneous countries to study climate coalition formation. They find that, given sufficient patience, the equilibrium coalition formation has a specific mathematical form, which involves participation by several countries, and comes close to internalizing the social cost of carbon. \cite{jaakkola2019non} introduce breakthrough clean technologies in a multi-country world under different degrees of international cooperation. They find that  spillover effects would lead to double free-riding, over-pollution, and underinvestment in clean technologies.




\section{The Integral Transformation Method}
\label{sec:generaltheory}

Here, we elaborate on the key methodological contribution of the Integral Transformation Method (ITM) to optimal control and differential game models. We will demonstrate that, under certain conditions, a solution to such dynamic problems can be found through the solution of a family of ``temporary" optimization problems parametrized by time and including at any time, $t$, the expected future evolution of the state variables and their marginal contribution to the respective objective functions. After describing the underlying problem in abstract form, in what follows we present a set of assumptions under which our integral transformation method
works for any deterministic continuous-time, infinite-horizon, finite-dimensional problem. This set of assumptions is sufficient for our method to apply, but in general not necessary. Although we will not treat these cases in this section, this approach applies more generally and can easily be adapted to cases involving discrete-time, finite-horizon, uncertainty, and infinite-dimensional variables; see \cite{Boucekkine2022} for an earlier application to a specific infinite-dimensional problem. We first develop our method for the single-player optimal control and for the $N$-player Nash cases. An extension to the case of Knightian uncertainty is discussed in Section 6. We will end with some comments on the potential advantages of our method.

\smallskip

In what follows, all finite-dimensional subsets and the functions involved are implicitly assumed to be Borel-measurable. We will use $\langle\cdot,\cdot\rangle$ to denote the inner product in $\R^n$ and $|\cdot|$ to denote the norm of (finite-dimensional) vector spaces.
Bold symbols will denote vector-valued or, more generally, operator-valued objects. The key ITM properties are given in Proposition \ref{pr:rewriting-general} and Theorem \ref{th:staticreduction} with their respective proofs given in the main text. The proofs of the subsequent results are shown in Appendix \ref{app:reduction}.

\subsection{The single-player case: optimal control}
\label{app:reductioncontrol}

Consider the following optimal control problem. Let
$\mathbf{X}= \mathbb{R}^n$ denote the state space and let $\mathbf{U}= \mathbb{R}^k$ denote the control space. Time is continuous and the time variable is denoted by $t\in\R_+$. 
We denote the state/control trajectories by $\mathbf{x}(\cdot)$ and $\mathbf{u}(\cdot)$, respectively; thus, $\mathbf{x}(t)\in \mathbf{X}$ and $\mathbf{u}(t)\in\mathbf{U}$. We assume that the  state equation is linear in the state variable, taking the form:
\begin{equation}\label{eq:stategeneral}
\mathbf{x}'(t)=\mathbf{A}(t)\mathbf{x}(t) +\mathbf{f}(t,\mathbf{u}(t)), \qquad \mathbf{x}(0)=\mathbf{x}_0\in \mathbf{X},
\end{equation}
where\footnote{We are using the following standard notation: given two finite dimensional  vector spaces $\mathbf{Y}$ and $\mathbf{Z}$, we denote by $\mathcal{L}(\mathbf{Y},\mathbf{Z})$ the space of all linear operators from $\mathbf{Y}$ to $\mathbf{Z}$. These can be identified with a spaces of matrices with suitable dimensions. When $\mathbf{Y}=\mathbf{X}$, then one simply writes $\mathcal{L}(\mathbf{X})$ fo $\mathcal{L}(\mathbf{X};\mathbf{X})$.} $\mathbf{A}:\R_+\to\mathcal{L} (\mathbf{X})$
and $\mathbf{f}:\R_+\times\mathbf{U}\to \mathbf{X}$.
When, for a given $\mathbf{u}(\cdot):\R_+\to \mathbf{U}$, the state equation is well posed, in the sense that it admits a unique global solution over $\R_+$, we denote the latter by $\mathbf{x}^{\mathbf{x}_0, \mathbf{u}(\cdot)}(\cdot)$, or simply by $\mathbf{x}(\cdot)$ when no confusion is possible.

We impose possibly time-inhomogenoeus constraints on the state and on the control variables as follows. Given $d,p\in \N\setminus\{0\}$, $\mathbf{g}:\R_+\times\mathbf{X}\to \mathbb{R}^d$, and $\mathbf{l}:\R_+\times\mathbf{U}\to \mathbb{R}^p$,
the state trajectory must satisfy the constraint
$$
\mathbf{g}(t,\mathbf{x}(t) )\le \mathbf{0}, \qquad \forall t \ge 0,
$$
whereas the control trajectory must satisfy the constraint:
$$
\mathbf{l}(t,\mathbf{u}(t) )\le \mathbf{0}, \qquad \forall t \ge 0.
$$

The objective functional is given in the following integral form \begin{equation}
\label{eq:OCPgeneral}
\mathcal{J}(\mathbf{x}_0;\mathbf{u}(\cdot)):=
\int_{0}^{\infty }e^{-\rho t} [
\langle \mathbf{a}(t),\mathbf{x}(t)\rangle + h(t,\mathbf{u}(t))]dt,
\end{equation}
where $\mathbf{a}:\R_+\to \mathbf{X}$ and 
${h}:\R_+\times\mathbf{U} \to\mathbb{R}$.

In order to guarantee well-posedness of the state equation, well-posedness and finiteness of objective functional, and the verification of the constraints, the set of admissible control strategies has to be chosen as a subset of\footnote{By $L_{\rho}^1(\mathbb{R}_+;\R)$ (or, more simply, by $L_{\rho}^1(\mathbb{R}_+)$), we hereafter denote the space of functions $f:\mathbb{R}_+\to \mathbb{R}$ such that $\int_0^\infty e^{-\rho t} |f(t)|dt<\infty$. By $L^1_{loc}(\mathbb{R}_+;\R)$ we denote the space of locally integrable functions; i.e., functions $f:\mathbb{R}_+\to \mathbb{R}$ such that $\int_0^M e^{-\rho t} |f(t)|dt<\infty$, for each $M>0$. Similar notation will be used for vector-valued functions.}
\begin{multline}
\label{eq:ADMCONTRgeneral}
\Bigg \{
\mathbf{u}(\cdot):\mathbb{R}_+ \to \mathbf{U} \; : \ \ 
t \mapsto  \mathbf{f}(t,\mathbf{u}(t))
\in L_{loc}^1(\mathbb{R}_+;\mathbf{X});
\\
t \mapsto \langle \mathbf{a}(t),\mathbf{x}(t)\rangle \in L_{\rho}^1(\mathbb{R}_+), \quad t \mapsto  h(t,\mathbf{u}(t)) \in L_{\rho}^1(\mathbb{R}_+);
\\
\mathbf{l}(t,\mathbf{u}(t)) \le \mathbf{0} \quad
\text{ and } \mathbf{g}(t,\mathbf{x}(t)) \le \mathbf{0}, \qquad \forall t \ge 0
\Bigg \}.
\end{multline}

We introduce the following.

\begin{Hypothesis}
\label{hp:reduction}
\begin{enumerate}[(i)]
\item[]
\item
The operator-valued map $\mathbf{A}:\R_+\to \mathcal{L}(\mathbf{X})$ is locally integrable. We denote the family of evolution operators\footnote{This family of operators is defined, for $t\geq s\geq 0$,  as  the unique solution to the operator-valued ODE 
$$
\begin{cases}
\frac{d}{dt} \Phi_{\mathbf{A}}(t, s) = \mathbf{A}(t) \Phi_{\mathbf{A}}(t, s), \\
\Phi_{\mathbf{A}}(s, s) = I.
\end{cases}$$.
see, e.g. \cite[Section 3.5]{bensoussan2007representation}.} generated by $\mathbf{A}$ by $\{\Phi_{\mathbf{A}}(t,s)\}_{t\geq s\geq 0}$.
\item
There exists $C>0$ such that
$$|\mathbf{f}(t,\mathbf{u})|\le C(1+|\mathbf{u}|),\quad \forall t\ge 0,\; \mathbf{u}\in \mathbf{U};$$
\item
The map
\begin{equation}
\label{eq:defbfb}
\R_+ \to \mathbf{X}, \qquad t \mapsto \mathbf{b}(t):= \int_{0}^{\infty }e^{-\rho \tau}\Phi^*_{\mathbf{A}}(t+\tau,t)
\mathbf{a}(t+\tau)d\tau,
\end{equation}
is well defined and bounded.
\end{enumerate}
\end{Hypothesis}

The above regularity assumptions are quite general. They allow us to rewrite the objective functional into a convenient form as the next proposition shows.

\begin{Proposition}
\label{pr:rewriting-general}
Suppose Assumption \ref{hp:reduction} holds and let 
$\mathbf{u}(\cdot) \in
L^1_{\rho}(\R_+,\mathbf{U})$. Then
\begin{equation}
\label{eq:reductionFubini}
\int_{0}^{\infty }e^{-\rho t} \langle\mathbf{a}(t), \mathbf{x}(t)\rangle dt
=
\left\langle
\mathbf{b}(0),\mathbf{x}_0
\right\rangle
+
\int_{0}^{\infty }e^{-\rho t}\left\langle
\mathbf{b}(t),\mathbf{f}(t,\mathbf{u}(t))\right\rangle dt,
\end{equation}
the right-hand side being well-defined and finite.
If, in addition, the map $t \mapsto  h(t,\mathbf{u}(t))$ belongs to $L_{\rho}^1(\mathbb{R}_+)$,
then the objective functional \eqref{eq:OCPgeneral} is well defined and finite, and it
can be written as
\begin{equation}
\label{eq:rewritingJ}
\mathcal{J}(\mathbf{x}_0;\mathbf{u}(\cdot))=
\left\langle
\mathbf{b}(0),\mathbf{x}_0
\right\rangle
+
\int_{0}^{\infty }e^{-\rho t}\left[
\left\langle
\mathbf{b}(t),\mathbf{f}(t,\mathbf{u}(t))\right\rangle +  h(t,\mathbf{u}(t))\right]dt.
\end{equation}
\end{Proposition}
\begin{proof}
We first demonstrate that \eqref{eq:reductionFubini} holds. For simplicity, here we deal with the case when $\mathbf{A}$ is constant; the general case is analogous.
Let $\mathbf{u}(\cdot)$ be as in \eqref{eq:ADMCONTRgeneral}.
By Assumption \ref{hp:reduction}(i), we have
$$
\mathbf{x}(t)=e^{\mathbf{A} t}\mathbf{x}_0+ \int_0^t e^{\mathbf{A}(t-s) }\mathbf{f}(s,\mathbf{u}(s))ds.
$$
Hence,
\[
\int_{0}^{\infty }e^{-\rho t} \langle\mathbf{a}(t), \mathbf{x}(t)\rangle dt
=
\int_{0}^{\infty }e^{-\rho t} \langle\mathbf{a}(t), e^{\mathbf{A} t}\mathbf{x}_0\rangle dt
+
\int_{0}^{\infty }e^{-\rho t} \bigg\langle\mathbf{a}(t),  \int_0^t e^{\mathbf{A}(t-s) }\mathbf{f}(s,\mathbf{u}(s))ds\bigg\rangle dt.
\]
Using the definition of $\mathbf{b}$ in \eqref{eq:defbfb}, the first term simply rewrites as 
$$
\left\langle \int_{0}^{\infty }e^{-\rho t} e^{\mathbf{A}^* t}\mathbf{a}(t)dt , \mathbf{x}_0\right\rangle
=
\left\langle
\mathbf{b}(0),\mathbf{x}_0
\right\rangle.
$$
Using Assumption \ref{hp:reduction}(iii), and applying the Fubini-Tonelli Theorem to the second term, we obtain that the map
$$\R_+ \to \R, \qquad t\mapsto \langle\mathbf{a}(t), \mathbf{x}(t)\rangle $$
belongs to $L^1_\rho (\mathbb{R}_+;  \mathbb{R})$ and that the second term may be rewritten as
\begin{multline*}
\int_{0}^{\infty }\int_{0}^{t}e^{-\rho t}\left\langle
\mathbf{a}(t),  e^{\mathbf{A}(t-s) }\mathbf{f}(s,\mathbf{u}(s))\right\rangle ds\,\, dt\\
=
\int_{0}^{\infty }\int_{0}^{t}e^{-\rho (t-s)}e^{-\rho s}\left\langle
e^{\mathbf{A}^*(t-s) }\mathbf{a}(t),  \mathbf{f}(s,\mathbf{u}(s))\right\rangle ds\,\, dt\\
=
\int_{0}^{\infty }e^{-\rho s}\left\langle
\int_{s}^{\infty}e^{-\rho (t-s)}   e^{\mathbf{A}^*(t-s) }\mathbf{a}(t) dt ,\mathbf{f}(s,\mathbf{u}(s))\right\rangle ds
=
\int_{0}^{\infty }e^{-\rho s}\left\langle \mathbf{b}(s),\mathbf{f}(s,\mathbf{u}(s))\right\rangle ds,
\end{multline*}
where, in the last step, we shifted the variable in the second integral.
The claim on the objective functional follows from the above arguments and from the additional integrability condition.
\end{proof}

\medskip

As the transformation of the objective functional \eqref{eq:OCPgeneral} leading to the representation in \eqref{eq:rewritingJ} is at the core of our methodology, here we offer some comments on Proposition \ref{pr:rewriting-general}. 
The power of this transformation comes from the fact that the state variable no longer appears in  \eqref{eq:rewritingJ}. All the relevant information about the evolution of the state is now encoded in the coefficient $\mathbf{b}(t)$, which captures the forward-looking component of the transformed problem. By construction, $\mathbf{b}(t)$ encompasses the relevant information on the future evolution of the state variables (through the current and future matrices $A$), as well as their marginal future impact on the intertemporal payoffs (through the current and future vectors $a$). As we shall see in Section 4.1., depending on the application, $\mathbf{b}(t)$ can admit a meaningful economic interpretation. The technique thus allows us to solve the dynamic optimization problem by rewriting the functional in a way that allows us to perform a pointwise optimization of the integrand (part iv in Theorem \ref{th:staticreduction} below). This approach is quite flexible, and the same simple transformation can be used in cases involving time-dependent systems, control and state constraints, as well as generalizations to differential games, which we will develop in the next subsections. On the other hand, we should emphasize that the applicability of our approach is limited to systems with a special structure: both the state equation and the objective functional must be linear in the state variable. 

\medskip

Proposition \ref{pr:rewriting-general} implies that the second condition in \eqref{eq:ADMCONTRgeneral}
is satisfied when Assumption \ref{hp:reduction} holds and $\mathbf{u}(\cdot)\in L^1_\rho(\R_+,\mathbf{U})$. Hence, under Assumption \ref{hp:reduction} it is reasonable to choose the following set of admissible control strategies:\footnote{An equivalent way to proceed would be to take the set of admissible strategies to be the set of all measurable
maps $\mathbf{u}(\cdot):\R_+\to \mathbf{U}$,
and to assign value $-\infty$ to all strategies
which are not in $\mathcal{U}(\mathbf{x}_0)$. } 
\begin{multline}
\label{eq:ADMCONTRgeneralreduced}
\mathcal{U}(\mathbf{x}_0):=
\Bigg \{
\mathbf{u}(\cdot) \in L_{\rho}^1(\mathbb{R}_+;  \mathbf{U}) \; : \\ 
t \mapsto  h(t,\mathbf{u}(t)) \in L_{\rho}^1(\mathbb{R}_+),
\; \mathbf{g}(t,\mathbf{x}(t))\le 0,\  \mathbf{l}(t,\mathbf{u}(t)) \le \mathbf{0}, \; \forall t \ge 0
\Bigg \}.
\end{multline}
In the remainder of this subsection  we assume that Assumption \ref{hp:reduction} holds and restrict attention to the set of admissible control strategies in \eqref{eq:ADMCONTRgeneralreduced}.
The standard notion of optimality is given in the following.
\begin{Definition}
An \emph{optimal control} $\mathbf{u}^*(\cdot)$ starting at $\mathbf{x}_0$ is a control strategy in the set $\mathcal{U}(\mathbf{x}_0)$ such that   
$$\mathcal{J}(\mathbf{x}_0;\mathbf{u}^*(\cdot))\geq \mathcal{J}(\mathbf{x}_0;\mathbf{u}(\cdot)), \ \ \ \forall \mathbf{u}(\cdot)\in\mathcal{U}(\mathbf{x}_0).$$ Given an optimal control $\mathbf{u}^*(\cdot)$, the associated  state trajectory $\mathbf{x}^*(\cdot):=\mathbf{x}^{\mathbf{x}_0,\mathbf{u}^*(\cdot)}(\cdot)$ is called an \emph{optimal state trajectory} starting at $\mathbf{x}_0$. We refer to the pair $\left(\mathbf{x}^*(\cdot),\mathbf{u}^*(\cdot)\right)$ as an \emph{optimal state/control trajectory} or an \emph{optimal couple} starting at $\mathbf{x}_0$.
\end{Definition}

Next, we will exploit the above result to demonstrate that the original dynamic problem is equivalent to solving a $t$-by-$t$ family of temporary optimization problems, where each time $t$ problem involves the expected future evolution of the state variables and their marginal contributions to the objective function. To this end, we use the result in Proposition \ref{pr:rewriting-general} and, for each fixed $t\ge 0$, we consider the following temporary optimization problem:
\begin{equation}
\label{eq:staticOCgeneral}
\max_{\mathbf{u\in \mathbf{U}}: \; \mathbf{l}(t, \mathbf{u})\le \mathbf{0}} \; \Big  [ \left\langle
\mathbf{b}(t),\mathbf{f}(t,\mathbf{u})\right\rangle +  h(t,\mathbf{u})
\Big ].
\end{equation}

We can then establish the following equivalence result for the ITM.

\begin{Theorem}[Optimal control]
\label{th:staticreduction}
\begin{itemize}
\item[]
\item[]
\item[(i)] \emph{(Sufficiency)}
Let  
$\mathbf{u}^*(\cdot)\in \mathcal{U}(\mathbf{x}_0)$ be
such that $\mathbf{u}^*(t)$ is a solution to the temporary optimization problem \eqref{eq:staticOCgeneral} for a.e. $t\in\R^{+}$. 
Then, $\mathbf{u}^*(\cdot)$ is optimal for the dynamic optimal control problem starting at $\mathbf{x}_0$. Moreover, in the absence of state constraints, $\mathbf{u}^*(\cdot)$ is optimal for every initial $\mathbf{x}_0$.
\item[(ii)] \emph{(Necessity)} Let $\mathbf{u}^*(\cdot)\in\mathcal{U}(\mathbf{x}_0)$ be as in \emph{(i)} above, and let $\overline{\mathbf{u}}(\cdot)\in\mathcal{U}(\mathbf{x}_0)$ be another optimal control starting at $\mathbf{x}_0$. Then, $\overline{\mathbf{u}}(\cdot)$ is a solution to the temporary optimization problem \eqref{eq:staticOCgeneral} for a.e. $t\in\R_{+}$. 

\item[(iii)] \emph{(Uniqueness)}
Suppose there exists  $\mathbf{u}^*(\cdot)\in\mathcal{U}(\mathbf{x}_0)$  as in item \emph{(i)} and the solution of \eqref{eq:staticOCgeneral} is unique for a.e. $t\ge 0$. Then  $\mathbf{u}^*(\cdot)\in\mathcal{U}(\mathbf{x}_0)$ is the a.e. unique optimal control problem starting at $\mathbf{x}_{0}$.

\item[(iv)] \emph{(Existence)}
Assume there are no state constraints and that, for a.e. $t\ge 0$, there exists a solution
to the temporary optimization problem \eqref{eq:staticOCgeneral}; that is, 
\begin{equation}
\label{eq:staticOCgeneralbis}
\arg \!\!\!\!\!\!\!\!\max_{\mathbf{u\in \mathbf{U}}: \; \mathbf{l}(t, \mathbf{u})\le \mathbf{0}} \; \Big  [ \left\langle
\mathbf{b}(t),\mathbf{f}(t,\mathbf{u})\right\rangle +  h(t,\mathbf{u})
\Big ]\neq \emptyset, \ \ \ \mbox{for a.e.} \ t\ge 0.
\end{equation} 
Then, all the single-valued (measurable) selections $\hat{\mathbf{u}}(\cdot)$ of the multivalued map 
\begin{equation}\label{selec}
\R_+\to\mathcal{P}(\mathbf{U}), \ \ \ \ \ t\mapsto \arg \!\!\!\!\!\!\!\!\max_{\mathbf{u\in \mathbf{U}}: \; \mathbf{l}(t, \mathbf{u})\le \mathbf{0}} \; \Big  [ \left\langle
\mathbf{b}(t),\mathbf{f}(t,\mathbf{u})\right\rangle +  h(t,\mathbf{u})
\Big ],
\end{equation} 
satisfying the integrability conditions
\begin{equation}\label{into}
\hat{\mathbf{u}}(\cdot)\in L^1_\rho(\R_+;\mathbf{U}), \ \ \ \ h(\cdot,\hat{\mathbf{u}}(\cdot)) \in L_{\rho}^1(\mathbb{R}_+)
\end{equation}
fulfill by construction the requirements of item \emph{(i)}, therefore,
are optimal for every  initial $\mathbf{x}_0$.\footnote{One can readily prove that if the map
$$
\R_+\times \mathbf{U} \to \R, \ \ \ \ \ \ (t,\mathbf{u})\mapsto \left\langle
\mathbf{b}(t),\mathbf{f}(t,\mathbf{u})\right\rangle +  h(t,\mathbf{u})
$$
is upper semicontinuous and the map $\mathbf{l}$ is lower semicontinuous, then there exists a single-valued measurable selection of the multivalued map \eqref{selec} verifying the integrability conditions \eqref{into}. This follows from Proposition 7.33, p. 153 of \cite{BertsekasShreveBOOK}.}
\end{itemize}

\end{Theorem}
\begin{proof}

\begin{enumerate}[(i)]
\item
By Proposition
\ref{pr:rewriting-general}, for any
$\mathbf{u}(\cdot) \in \mathcal{U}(\mathbf{x}_0)$,
\begin{align}
&\mathcal{J}(\mathbf{x}_0;\mathbf{u}(\cdot))=
\left\langle
\mathbf{b}(0),\mathbf{x}_0
\right\rangle
+
\int_{0}^{\infty }e^{-\rho t}\left[
\left\langle
\mathbf{b}(t),\mathbf{f}(t,\mathbf{u}(t))\right\rangle +  h(t,\mathbf{u}(t))\right]dt\\
&\nonumber
\le
\left\langle
\mathbf{b}(0),\mathbf{x}_0
\right\rangle
+
\int_{0}^{\infty }e^{-\rho t}
\sup_{\mathbf{u\in \mathbf{U}}: \; \mathbf{l}(t, \mathbf{u})\le \mathbf{0}}
\left[
\left\langle
\mathbf{b}(t),\mathbf{f}(t,\mathbf{u})\right\rangle +  h(t,\mathbf{u})\right]dt
=
\mathcal{J}(\mathbf{x}_0;\mathbf{u}^*(\cdot)).
\end{align}

Since
$\mathbf{u}^*(\cdot) \in \mathcal{U}(\mathbf{x}_0)$,
the above implies the optimality of
$\mathbf{u}^*(\cdot)$ starting at $\mathbf{x}_0$.

In the absence of state constraints, the set $\mathcal{U}(\bx_{0})$ does not depend on $\bx_{0}$. Therefore, $\mathbf{u}^*(\cdot)$  is optimal for every initial $\mathbf{x}_{0}$.

\medskip

\item 
Let $\mathbf{u}^*(\cdot)$ and $\overline{\mathbf{u}}(\cdot)$ be as in the statement.
From the above discussion,
it follows that
\begin{align*}
&\mathcal{J}(\mathbf{x}_0;
\overline{\mathbf{u}}(\cdot))
=
\mathcal{J}(\mathbf{x}_0;\mathbf{u}^*(\cdot))\\&
=
\left\langle
\mathbf{b}(0),\mathbf{x}_0
\right\rangle
+
\int_{0}^{\infty }e^{-\rho t}
\sup_{\mathbf{u\in \mathbf{U}}: \; \mathbf{l}(t, \mathbf{u})\le \mathbf{0}}\left[
\left\langle
\mathbf{b}(t),\mathbf{f}(t,\mathbf{u})\right\rangle +  h(t,\mathbf{u})\right]dt.
\end{align*}
This implies:
$$
\int_{0}^{\infty }e^{-\rho t}
\left[
\left\langle
\mathbf{b}(t),\mathbf{f}(t,\overline{\mathbf{u}}(t))\right\rangle +  h(t,\overline{\mathbf{u}}(t))\right]dt= \int_{0}^{\infty }e^{-\rho t}
\sup_{\mathbf{u\in \mathbf{U}}: \; \mathbf{l}(t, \mathbf{u})\le \mathbf{0}}\left[
\left\langle
\mathbf{b}(t),\mathbf{f}(t,\mathbf{u})\right\rangle +  h(t,\mathbf{u})\right]dt,
$$
hence, 
$$
\left\langle
\mathbf{b}(t),\mathbf{f}(t,\overline{\mathbf{u}}(t))
\right\rangle +  h(t,\overline{\mathbf{u}}(t))
=
\sup_{\mathbf{u\in \mathbf{U}}: \; \mathbf{l}(t, \mathbf{u})\le \mathbf{0}}\left[
\left\langle
\mathbf{b}(t),\mathbf{f}(t,\mathbf{u})\right\rangle +  h(t,\mathbf{u})\right], \ \ \ \mbox{for a.e.} \ t\geq 0,
$$
as claimed.

\medskip

\item Follows from (ii).

\medskip

\item Existence follows since, in the absence of state constraints, the control $\hat{\mathbf{u}}(\cdot)$ belongs to  $\mathcal{U}(\mathbf{x}_0)$. Thus, by construction, it verifies the requirements in (i).  

\end{enumerate}
\end{proof}

\subsubsection{ITM versus standard dynamic optimization methods} 
Here we build on Proposition \ref{pr:rewriting-general} to elaborate on our Integral Transformation Method. Transforming the objective functional \eqref{eq:OCPgeneral} into the representation in \eqref{eq:rewritingJ} allows us to rewrite the optimal control problem in an advantageous way, as the state variable no longer appears in  \eqref{eq:rewritingJ}, and all relevant information about the evolution of the state is encoded in the coefficient $\mathbf{b}(t)$. This makes solving the dynamic optimization problem straightforward, as it only requires performing a uniform optimization on the integrand (part (iv) in Theorem \ref{th:staticreduction}). 

Traditional approaches, such as Dynamic Programming and the use of the Maximum Principle, are applicable to this problem but the following points are worth mentioning.
\begin{enumerate}[(i)]
\item 
Following the Dynamic Programming approach, one could write the HJB equation for $(t,\mathbf{x})\in\mathbb{R}_+\times \mathbf{X}$ as follows: 
\begin{align}\label{HJB} \rho v(t,\mathbf{x}) -v_{t}(t,\mathbf{x})= &\  \langle \mathbf{a}(t)\mathbf{x},\,v_{\mathbf{x}}(t,\mathbf{x})\rangle\\&+ \sup_{\mathbf{u \in \mathbf{U}}: \; \mathbf{l}(t, \mathbf{u})\le \mathbf{0}} \{\langle \mathbf{f}(t,\mathbf{u}), v_{\mathbf{x}}(t,\mathbf{x})\rangle+h(t,\mathbf{x},\mathbf{u})\}.
\nonumber \end{align} 
Assuming a solution of the form $v(t,\mathbf{x})=\langle \balpha(t),\mathbf{x}\rangle +\beta(t)$, the following steps would be followed to solve the problem with the Dynamic Programming-HJB method: \smallskip

\begin{enumerate}[--] 

\item Determine $\balpha:\mathbb{R}_{+}\to \mathbb{R}^{n}$ and $\beta: \mathbb{R}_{+}\to \mathbb{R}$ that solve 
\eqref{HJB}. This step involves solving a system of  ODEs. 

\item Confirm, through a verification theorem, that $v$ is indeed the value function, and that the corresponding candidate optimal control is indeed optimal. This procedure passes through a transversality condition on $v$ for $t\to\infty$ and can be fairly complex when state constraints are present. Once completed, it provides \emph{sufficient conditions} for optimality. Under suitable conditions, e.g. uniqueness of solutions to the HJB equation, these conditions are also necessary.
\end{enumerate}

\item
Similarly, applying the classical Maximum Principle would lead to a forward-backward system of $2n$ ODEs with time-dependent coefficients. In addition, to identify the solution to the optimal control problem, appropriate transversality conditions for the adjoint variables are necessary. Ultimately, in the absence of state constraints, this approach yields \emph{necessary conditions} for optimality, which are also sufficient under concavity assumptions on the data. When state constraints are present, additional challenges arise, as one must carefully consider the case where the constraints are binding, which can be non-trivial. 
\end{enumerate}

Although the ITM applies only when the underlying problem has a specific structure (additive separability in state and control variables and linearity in the state variable in both the state equation and the objective functional), when these conditions are met, it has certain notable advantages:

\begin{enumerate}[(i)]

\item The ITM provides a direct approach that allows us to obtain in a single step both necessary and sufficient conditions for optimality in the optimal control problem, even in the absence of concavity, based on optimizing a family of temporary problems.

\item The ITM handles problems involving discontinuities more easily than traditional dynamic optimization methods. This includes not just discontinuous functions (e.g., piecewise constant) but also generalized distributions (e.g., Dirac delta functions) which arise in impulse control problems.
\item The ITM is applicable when $\mathbf{X}$ and $\mathbf{U}$ are infinite-dimensional. This is the case, for example, in applications of optimal control in a PDE setup; see \cite{boucekkine2022managing, Boucekkine2022}, or in dealing with delay equations. The method is particularly advantageous in these cases, as it does not involve the use of unbounded differential operators in the solution of the HJB equation and in the verification theorem; see \cite{boucekkine2019geographic}, and  \cite{boucekkine2022managing}. Moreover, the ITM allows for a useful economic interpretation of certain quantities; see, for example, the interpretation of the function $\alpha(\cdot)$ in \cite{boucekkine2022managing}, which corresponds to our $\mathbf{b}(0)$ in the autonomous case.

\item As we will demonstrate in the next section, the ITM is readily applicable in the $N$-player differential game context, where it provides necessary and sufficient conditions for Nash equilibria from corresponding necessary and sufficient equilibrium conditions for a family of temporary non-cooperative games. In this context, the Dynamic Programming approach involves solving a system of $N$ HJB equations, while the Maximum Principle approach requires solving a coupled forward-backward system of $2n\times N$ ODEs. While both approaches are valid, they are significantly more cumbersome.

\end{enumerate}

In the next subsection we demonstrate in particular how the ITM approach described above for the single-agent dynamic decision problem can be extended to the case of N-player dynamic games. We will begin with an exploration of open-loop Nash equilibria. Later, we will investigate existence and uniqueness of Markov Perfect equilibria in the special class of affine feedbacks.

\subsection{The $N$-player case: Nash equilibria}
\label{app:reductionNash}

Consider the following non-cooperative differential game. There are $N\geq 2$ players, labeled by $i=1,..,N$.
We denote by  $\mathbf{X}= \mathbb{R}^n$ the state space, which is assumed to be common to all players,
and by $\mathbf{U}^i=\mathbb{R}^{k_i}$ the control space of player $i$, for $i=1,\dots, N$. We set
$$
\mathbf{U}:=\mathbf{U}^1\times \dots \times \mathbf{U}^N
$$
and let $k:=k_1+\dots + k_N$ be the dimension of $\mathbf{U}$. The state path is denoted by $\mathbf{x}(\cdot)$, whereas the control path of player $i$ is denoted by $\mathbf{u}^i(\cdot)$.
Finally, we set
$$
\mathbf{u}(\cdot):=\left(\mathbf{u}^1(\cdot),\dots,
\mathbf{u}^N(\cdot) \right).
$$
The state equation is the same as in \eqref{eq:stategeneral}. However, now 
the constraints to be satisfied by the state may be different for each player. For player $i$, these are given by
$$
\mathbf{g}^i(t,\mathbf{x}(t)) \le \mathbf{0}, \qquad \forall t \ge 0,
$$
where $\mathbf{g}^i:\R_+\times\mathbf{X}\to \mathbb{R}^{d_i}$ for some $d_i \in \mathbb{N}$.
The constraints satisfied by the control strategies can also be different for each player. For player $i$, these are given by
$$
\mathbf{l}^i(t,\mathbf{u}(t)) \le \mathbf{0}, \qquad \forall t \ge 0,
$$
where $\mathbf{l}^i:\R_+\times\mathbf{U}\to \mathbb{R}^{p_i}$ for some $p_i \in \mathbb{N}$.


The objective functional of player $i$ is given by\footnote{Hereafter, we use $\{-i\}$ to denote the profile of every player except player $i$; i.e., $\{1,...,N\}\setminus \{i\}$ for $i\in\{1,...,N\}.$}
\begin{equation}
\label{eq:OCPgeneralNash}
\mathcal{J}^i(\mathbf{x}_0, \mathbf{u}^{-i}(\cdot);\mathbf{u}^i(\cdot)):=
\int_{0}^{\infty }e^{-\rho_i t} [\langle\mathbf{a}^i(t), \mathbf{x}(t)\rangle + h^i(t,\mathbf{u}^i(t),\mathbf{u}^{-i}(t))]dt,
\end{equation}
where $\mathbf{a}^i:\R_+\to\mathbf{X}$ and ${h^i}:\R_+\times\mathbf{U} \to\mathbb{R}$.

As in the previous subsection, we introduce the maps: 
$$\mathbf{b}^{i}(t):\R_+ \to\mathbf{X}, \qquad t \mapsto \mathbf{b}^{i}(t):= \int_{0}^{\infty }e^{-\rho \tau}\Phi^*_{\mathbf{A}}(t+\tau,t)
\mathbf{a}^{i}(t+\tau)d\tau, \ \ \ i=1,...,N.$$
For the remainder of the subsection, we impose the following.

\begin{Hypothesis}
\label{hp:reductionNash}
Hypotheses \ref{hp:reduction}\emph{(i)-(ii)} and \ref{hp:reduction}\emph{(iii)} hold, with $\mathbf{a},\mathbf{b}$ replaced by $\mathbf{a}^{i}$, $\mathbf{b}^{i}$; 
i.e., for all $i=1,\ldots,N$,
\begin{equation}
\label{eq:defbfbnash}
\mathbf{b}^{i}(t):\R_+ \to \mathbf{X}, \qquad t \mapsto \mathbf{b}^{i}(t):= \int_{0}^{\infty }e^{-\rho \tau}\Phi^*_{\mathbf{A}}(t+\tau,t)
\mathbf{a}^{i}(t+\tau)d\tau, 
\end{equation}
is well defined and bounded.

\end{Hypothesis}

Similarly to  the previous subsection, we have the following.
\begin{Proposition}
\label{pr:rewriting-generalNash}
Suppose Assumption \ref{hp:reductionNash} holds and, for every $i=1,\dots,N$, consider a control strategy of the $i$-th player:
$\mathbf{u}^i(\cdot) \in L^1_{\rho_i}(\R_+,\mathbf{U}^i)$. Let
$\mathbf{u}(\cdot):=(\mathbf{u}^1(\cdot),\dots, \mathbf{u}^N(\cdot))$,
and $\mathbf{x}(\cdot):=\mathbf{x}^\mathbf{x_0,\mathbf{u}(\cdot)}(\cdot)$ be the associated state trajectory and fix $i\in \{1,\dots N\}$. Then:
\begin{equation}
\label{eq:reductionFubiniNash}
\int_{0}^{\infty }e^{-\rho_i t} \langle\mathbf{a}^i(t), \mathbf{x}(t)\rangle dt
=
\left\langle
\mathbf{b}^i(0),\mathbf{x}_0
\right\rangle
+
\int_{0}^{\infty }e^{-\rho_i t}\left\langle
\mathbf{b}^i(t),\mathbf{f}(t,\mathbf{u}(t))\right\rangle dt,
\end{equation}
where the right-hand side is well-defined and finite. If, in addition the map $t \mapsto  h^i(t,\mathbf{u}^i(t))$ belongs  to $L_{\rho_i}^1(\mathbb{R}_+)$,
then the objective functional \eqref{eq:OCPgeneralNash} is well-defined and finite and 
can be written as:
\begin{equation}
\label{eq:rewritingJNash}
\mathcal{J}^i(\mathbf{x}_0,\mathbf{u}^{-i}(\cdot);\mathbf{u}^i(\cdot))=
\left\langle
\mathbf{b}^i(0),\mathbf{x}_0
\right\rangle
+
\int_{0}^{\infty }e^{-\rho_i t}\left[
\left\langle
\mathbf{b}^i(t),\mathbf{f}(t,\mathbf{u}(t))\right\rangle +  h^i(t,\mathbf{u}^i(t))\right]dt.
\end{equation}

\end{Proposition}

The proof is omitted as the argument is similar to the one used in the proof of Proposition \ref{pr:rewriting-general}.

\subsubsection{Open-loop Nash equilibria}
Prior to defining equilibrium, we first define the set of admissible strategies. We begin by considering the following set:
$$\mathcal{U}^0_G:=L^1_{\rho_1}(\R_+,\mathbf{U}^1)\times \ldots \times L^1_{\rho_N}(\R_+,\mathbf{U}^N).$$
The set of \emph{open-loop} admissible strategies is then given by 
\begin{multline}\label{ug}
\mathcal{U}_G(\mathbf{x}_{0})=\Big\{\mathbf{u}(\cdot)= (\mathbf{u}^1(\cdot),...,\mathbf{u}^N(\cdot))\in \mathcal{U}^0_G:  \ \ \ t\mapsto h^{i}(t,\mathbf{u}(t))\in L^{1}_{\rho_{i}}(\R_{+}) \ \forall i=1,...,N,\\
\mathbf{g}^{i}(t,\mathbf{x}(t))\leq \mathbf{0}, \ \ \ \ \mathbf{l}^{i}(t,\mathbf{u}(t))\leq \mathbf{0}, \ \ \ \  \ \forall t\geq 0, \ \forall i=1,...,N\ \Big\}.
\end{multline}
Notice that, in the absence of state constraints, the above set does not depend on $\mathbf{x}_{0}$. We shall thus simply denote it by $\mathcal{U}_{G}$.

Given $i\in\{1,...,N\}$ and\footnote{With some abuse of notation, here we denote
$$
\mathcal{U}_{G}^0\setminus L^{1}_{\rho_{i}}(\R_{+};\mathbf{U}^{i}):= \prod_{j\neq i} L^1_{\rho_j}(\R_+,\mathbf{U}^j).
$$} 
$\mathbf{u}^{-i}(\cdot)\in \mathcal{U}_{G}^0\setminus L^{1}_{\rho_{i}}(\R_{+};\mathbf{U}^{i})$, we define the set of admissible control strategies of player $i$ as the $i$-th section of the above set; i.e.,
\begin{equation}\label{ug2}
\mathcal{U}^{i}_G(\mathbf{x}_{0},\mathbf{u}^{-i}(\cdot)):=\Big\{ \mathbf{u}^{i}(\cdot)\in L^{1}_{\rho_{i}}(\R^{+};\mathbf{U}^{i}): \ (\mathbf{u}^{i}(\cdot),\mathbf{u}^{-i}(\cdot))\in \mathcal{U}_{G}(\mathbf{x}_{0})\Big\}.
\end{equation}

As, in the absence of state constraints, the above sets do not depend on $\mathbf{x}_{0}$ we again denote them by $\mathcal{U}^{i}_G(\mathbf{u}^{-i}(\cdot))$. As is standard in game theory, the set of admissible strategies $\mathcal{U}^{i}_G(\mathbf{x}_{0},\mathbf{u}^{-i}(\cdot))$ of player $i$ depends not only on the initial state $\mathbf{x}_{0}$, but also on the strategies $\mathbf{u}^{-i}(\cdot)$ of the other players.\footnote{The above choice is equivalent to taking the set of admissible strategies for each player $i$ to be $L^{1}_{\rho_{i}}(\R_{+};\mathbf{U}_{i})$, and assigning value $-\infty$ to control strategies which do not lie in $\mathcal{U}^{i}_G(\mathbf{x}_{0},\mathbf{u}^{-i}(\cdot))$.} Next, we introduce the notion of \emph{open-loop Nash equilibrium} for this setup.

\begin{Definition}[Open-loop Nash equilibrium]
An admissible open-loop strategy $\hat{\mathbf{u}}(\cdot)\in \mathcal{U}_{G}(\mathbf{x}_{0})$ starting at $\mathbf{x}_0$ is called an \emph{open-loop Nash equilibrium} for the dynamic game starting at $\mathbf{x}_{0}$ if, for all $i\in\{1,...,N\}$,
$$\mathcal{J}^{i}(\mathbf{x}_0, \hat{\mathbf{u}}^{-i}(\cdot);\hat{\mathbf{u}}^i(\cdot))\ \geq\ \mathcal{J}^{i}(\mathbf{x}_0, \hat{\mathbf{u}}^{-i}(\cdot);\mathbf{u}^i(\cdot)), \ \ \ \forall \mathbf{u}(\cdot)\in\mathcal{U}_{G}^{i}(\mathbf{x}_0,\hat{\mathbf{u}}^{-i}(\cdot)).$$
\end{Definition}

Similarly to the previous subsection, the ITM allows us to find Nash equilibria of the dynamic game by investigating a family of associated temporary games parametrized by time. More precisely, for each $t\in\R_{+}$, consider the following temporary game. There are $N$ players and, for $i=1,...,N$, player $i$ takes as given the choices $\mathbf{u}^{-i}_{t}\in\mathbf{U}^{-i}$ of the others players at time $t$ and seeks to maximize,
over the set\footnote{This is again equivalent to taking the set of admissible strategies for player $i$ to be
$\mathbf{U}_i$,
and assigning value $-\infty$ to strategies
which are not in
$\mathbf{U}_{i,\mathbf{u}^{-i}_{t}}$.}

$$
\mathbf{U}_{i,\mathbf{u}^{-i}_{t}}:=
\left\{
\mathbf{u}^{i}\in \mathbf{U}^i: \;
\mathbf{l}^i(t,\mathbf{u}^{i},\mathbf{u}^{-i}_{t})\le 0
\right\},
$$
the $i$-th objective function:
\begin{equation}
\label{eq:staticNash}
\mathbf{U}_{i,\mathbf{u}^{-i}_{t}}
\rightarrow \R,
\qquad \qquad
\mathbf{u}^i \mapsto
\left\langle
\mathbf{b}^i(t),\mathbf{f}(t,\mathbf{u}^i,
\mathbf{u}^{-i}_{t}) \right\rangle +h(t,\mathbf{u}^i,\mathbf{u}^{-i}_{t}).
\end{equation}
The notion of Nash equilibrium in this temporary context is given below.

\begin{Definition}[Temporary Nash equilibrium]\label{def:Nashstatic}
A Nash equilibrium for the game at time $t\in\R_{+}$ is an $N$-tuple of temporary strategies $\hat{\mathbf{u}}_{t}=(\hat{\mathbf{u}}^{1}_{t}, ..., \hat{\mathbf{u}}_{t}^{N})\in\mathbf{U}$ such that for each $i=1,...,N$,
\begin{equation}\label{Nashstatic}
\begin{cases}
\hat{\mathbf{u}}_{t}^{i}\in \mathbf{U}_{i,\hat{\mathbf{u}}^{-i}_{t}},\\\\
\displaystyle{\hat{\mathbf{u}}_{t}^{i}\in \arg\!\!\!\max_{\mathbf{u^{i}\in \mathbf{U}_{i,\hat{\mathbf{u}}^{-i}_{t}}}}  \  \big\{ \left\langle
\mathbf{b}^i(t),\mathbf{f}(t,\mathbf{u}^i,
\hat{\mathbf{u}}^{-i}_{t}) \right\rangle +h(t,\mathbf{u}^i,\hat{\mathbf{u}}^{-i}_{t})\big\}}.
\end{cases}
\end{equation}
\end{Definition}

We will denote by $\mathbf{NE}_{t}$ the (possibly empty) set of Nash equilibria for the temporary game in $t\in\R_{+}$. We then have the following equivalence result, which is analogous to the one obtained in the optimal control case above.

\begin{Theorem}[Open-loop Nash equilibria]
\label{th:staticreductionNash}
\begin{itemize}
\item[]
\item[]
\item[(i)] \emph{(Sufficient conditions for open-loop Nash equilibria)}
Let  
$\widehat{\bu}(\cdot)\in \mathcal{U}_{G}(\mathbf{x}_0)$ be
such that $\widehat{\mathbf{u}}_{t}:=\widehat{\mathbf{u}}(t)$ belongs to $\mathbf{NE}_{t}$ for a.e. $t\in\R_{+}$.
Then $\widehat{\mathbf{u}}(\cdot)$ is an open-loop Nash equilibrium for the dynamic game starting at $\mathbf{x}_0$. Moreover, in the absence of state constraints, $\widehat{\mathbf{u}}(\cdot)$ is a Nash equilibrium  for every  initial condition $\mathbf{x}_0$.
\end{itemize}
Assume that there are no state constraints.
\begin{itemize}
\item[(ii)] \emph{(Necessary conditions for open-loop Nash equilibria)} Let $\overline{\mathbf{u}}(\cdot)\in\mathcal{U}_{G}$ be an open-loop Nash equilibrium for the dynamic game starting  at $\mathbf{x}_0$ and assume that, for $i=1,\dots,N$, the map
$$
\R_+\times \mathbf{U}_{i}\to \R^{p_i}; \ \ \ (t,\mathbf{u}_i)\mapsto
\mathbf{l}_{i}(t,
\overline{\mathbf{u}}_{-i}(t),\mathbf{u}_i)
$$
is lower semicontinuous. Let
$$
D_{i}:=\Big\{(t,\mathbf{u}_{i})\in \R_+\times \mathbf{U}_{i}: \ 
\mathbf{u}_i\in \mathbf{U}_{i,\overline{\mathbf{u}}_{-i}(t)},\;
\forall t \ge 0\Big\}
$$
and assume that the map 
$$
D_{i}\to \R; \ \  \ (t,\mathbf{u}_i)\mapsto \left\langle
\mathbf{b}_{i}(t),\mathbf{f}(t,
\overline{\mathbf{u}}_{-i}(t),\mathbf{u}_i)
\right\rangle +  h(t,\overline{\mathbf{u}}_{-i}(t),\mathbf{u}_i)
$$
is upper semicontinuous and coercive in $\bu_i$, uniformly in $t\in \R_+$.
Then,  $\overline{\mathbf{u}}(t)\in  \mathbf{NE}_{t}$, for a.e. $t\in \R_{+}$, and $\overline{\mathbf{u}}(t)$ is an open-loop Nash equilibrium for any initial condition.
\medskip
\item[(iii)] \emph{(Uniqueness of open-loop Nash equilibria)} Assume that $\mathbf{NE}_{t}$ is at most a singleton for a.e. $t\in\R_{+}$. Then, there is at most one open-loop Nash equilibrium for the dynamic game.
\medskip

\item[(iv)] \emph{(Existence of open-loop Nash equilibria)}
Assume that  $\mathbf{NE}_{t}$  is non-empty for a.e. $t\in\R_{+}$. 
Then, any single-valued measurable selection $\widehat{\mathbf{u}}(\cdot)$ of the multi-valued map 
\begin{equation}\label{selecNash}
\R_+\to\mathcal{P}(\mathbf{U}); \ \ \ \ \ t\mapsto \mathbf{NE}_{t}
\end{equation} 
satisfying the integrability conditions
\begin{equation}\label{inttt}
\widehat{\mathbf{u}}_{i}(\cdot)\in L^1_{\rho_{i}}(\R_+;\mathbf{U}), \ \ \ \ h_{i}(\cdot,\widehat{\mathbf{u}}(\cdot)) \in L_{\rho_{i}}^1(\mathbb{R}_+), \ \ \ \forall i=1,...,N
\end{equation} satisfies the requirements of item \emph{(i)} and thus
constitutes an open-loop Nash equilibrium for the dynamic game, for every initial condition $\mathbf{x}_0$.\footnote{Similarly to the optimal control case, if we also assume that 
the maps
$$
\R_{+}\times \mathbf{U} \to\R, \ \ \ (t,\mathbf{u})\mapsto \left\langle
\mathbf{b}^{i}(t),\mathbf{f}(t,\mathbf{u})\right\rangle +  h^{i}(t,\mathbf{u})
$$
are upper semicontinuous and that the maps $\mathbf{l}^{i}$ are lower semicontinuous, for each $i=1,...,N$, then a single measurable selection of the multivalued map \eqref{selecNash} satisfying the integrability conditions \eqref{inttt} exists. }

\end{itemize}
\end{Theorem}

A detailed proof is given in Appendix \ref{app:reduction}. 

Theorem \ref{th:staticreductionNash} demonstrates that the reformulation of the objective functional allows us to transform the problem of solving the differential game to that of solving a family of temporary games. The conditions allowing the application of the ITM are the same as in the single-agent optimal control case. Namely, they require linearity of the state equation and of the objective functional with respect to the state variable. In the context of differential games, related conditions have been discussed in \cite{dockner1985} and \cite{zaccour2005}, who studied the solvability of a system of ODEs representing the (sufficient) optimality conditions via the Pontryagin Maximum Principle. Our method is flexible, allowing for the treatment of constraints, as well as for generalizations to the analysis of Markovian equilibria). We next turn our attention to the study of Markovian Nash equilibria.

\subsubsection{Markov-Nash equilibria} It is well-known that open-loop Nash equilibrium can be restrictive, as it relies on commitment and it does not incorporate the reactions of the players' strategies to the other players' choices. An alternative concept is that of (closed-loop) Markovian equilibrium \cite[Ch.\,4, Sec.\,1]{dockner2000differential}. In what follows, we extend the notion of admissible strategies to accommodate feedback maps.

Given a measurable map $\bvarphi=(\bvarphi_{1},...,\bvarphi_{N}):\R_+\times \mathbf{X}\to \mathbf{U}$, we say that $\bvarphi$ is an \emph{admissible Markovian strategy for the dynamic game starting at} $\mathbf{x}_0$ if it satisfies the following. 
\begin{itemize}
\item[(i)] The closed-loop equation
\begin{equation}\label{eq:cle}
\mathbf{x}'(t)=\mathbf{A}(t)\mathbf{x}(t) +\mathbf{f}(t,\bvarphi(t,\mathbf{x}(t))), \qquad \mathbf{x}(0)=\mathbf{x}_0\in \mathbf{X},
\end{equation}
admits a unique solution, denoted by $\mathbf{x}^{\mathbf{x}_{0},\bvarphi}(\cdot)$. \smallskip
\item[(ii)] The feedback strategy $\mathbf{u}^{\mathbf{x}_{0},\bvarphi}(\cdot):= \bvarphi(\cdot, \mathbf{x}^{\mathbf{x}_{0},\bvarphi}(\cdot))$ lies in $\mathcal{U}_{N}(\mathbf{x}_{0})$; i.e.,
\smallskip
\begin{itemize}  
\item[--] the solution to the closed-loop equation satisfies the state constraints: 
$$\mathbf{g}_{i}(t,\mathbf{x}^{\mathbf{x}_{0},\bvarphi}(t))\leq \mathbf{0}, \ \ \ \ \forall t\in\R_{+}, \ \forall i=1,...,N;$$
\item[--] the feedback strategy satisfies the control constraints: 
$$\mathbf{l}_{i}(t, {\boldsymbol{\bvarphi}}(t,\mathbf{x}^{\mathbf{x}_{0},\bvarphi}(t))\leq \mathbf{0}, \ \ \ \ \forall t\in\R_{+}, \ \forall i=1,...,N;$$
\item[--] the map $\mathbf{u}^{\mathbf{x}_{0},\bvarphi}(\cdot) \in  L^1_{\rho_1}(\R_+,\mathbf{U}_1)\times \ldots \times L^1_{\rho_N}(\R_+,\mathbf{U}_N)$, and the map 
$ t\mapsto h_{i}(t,\mathbf{u}(t))\in L^{1}_{\rho_{i}}(\R_{+})$, for every  $i=1,...,N$.
\end{itemize}
\end{itemize}

We will denote the set of admissible Markovian strategies starting at $\mathbf{x}_{0}$ by 
$\mathcal{M}_G(\mathbf{x}_{0})$. In the context of no time-dependent coefficients, it might be interesting to consider admissible Markovian strategies that also do not depend explicitly on time. We term these strategies \emph{stationary} Markovian, and denote the set of such strategies by $\mathcal{M}^{o}_G(\mathbf{x}_{0})$.

Given $\bvarphi\in\mathcal{M}_{G}(\mathbf{x}_{0})$, abusing notation, we define
$$
\mathcal{J}_{i}(\mathbf{x}_{0},\bvarphi_{-i};\bvarphi_{i}):=\mathcal{J}_{i} (\mathbf{x}_{0},\mathbf{u}^{\mathbf{x}_{0},\bvarphi;-i}(\cdot);\bu^{\bx_{0},\bvarphi;i}(\cdot)).
$$
Finally, given $\bvarphi_{-i}:\R_{+}\times \mathbf{X}\to \mathbf{u}_{-i}$, we set
\begin{equation}\label{ugr}
\mathcal{M}_{i}^G(\mathbf{x}_{0},\bvarphi_{-i}):=\Big\{ \bvarphi_{i}: \R^+\times \mathbf{X}\to \mathbf{U}_{i} \ : \ \  (\bvarphi_{i},\bvarphi_{-i})\in \mathcal{M}_{G}(\mathbf{x}_{0})\Big\}.
\end{equation}
We are now ready to define Markovian Nash equilibria.

\begin{Definition}[Markov-Nash equilibrium]
Given $\mathbf{x}_{0}\in \R^{n}$, a Markovian admissible strategy profile $\widehat{\bvarphi}\in \mathcal{M}_{G}(\mathbf{x}_{0})$  is called a \emph{Markov-Nash equilibrium} for the dynamic game starting at $\mathbf{x}_{0}$ if, for all $i\in\{1,...,N\}$, 
$$\mathcal{J}_{i}(\mathbf{x}_0, \widehat{\bvarphi}_{-i};\widehat{\bvarphi}_i)\geq\mathcal{J}_{i}(\mathbf{x}_0, \widehat{\bvarphi}_{-i};\bvarphi_i), \ \ \ \forall \bvarphi_{i}\in\mathcal{M}_{G}^{i}(\mathbf{x}_0,\bvarphi_{-i}).$$
\end{Definition}
\begin{Remark}\label{rem:on}
The above definition of Markovian Nash equilibrium is equivalent to the one provided in  \cite[Def.\,4.1]{dockner2000differential}. Indeed, 
for (Markovian) optimal control problems, there is no difference between optimizing over open-loop controls or over closed-loop (Markovian) controls. Hence, $\widehat{\bvarphi}\in \mathcal{M}_{G}(\mathbf{x}_{0})$ being a \emph{Markovian Nash equilibrium} for the dynamic game starting at $\mathbf{x}_{0}$ is equivalent to, for each $i=1,...,N$, the feedback control $\widehat{\bu}_{i}(\cdot):=\widehat{\bvarphi}_{i}(\cdot;\mathbf{x}^{\mathbf{x}_{0},\widehat{\bvarphi}}(\cdot))$ being an optimal control of the control problem: 
$$
\sup_{\bu_{i}(\cdot)\in \mathcal{U}_{i}^{G}(\mathbf{x}_0,\widehat{\bvarphi}_{-i})} \int_{0}^{\infty} e^{-\rho_{i}t} [\langle\mathbf{a}_i(t), \mathbf{x}(t)\rangle + h_i(t, \widehat\bvarphi_{-i}(t,\bx(t)), \mathbf{u}_{i}(t))]dt,
$$
under the state equation 
$$
\mathbf{x}'(t)=\mathbf{A}(t)\mathbf{x}(t) +\mathbf{f}(t,\widehat\bvarphi_{-i}(t,\mathbf{x}(t)),\bu_{i}(t)), \qquad \mathbf{x}(0)=\mathbf{x}_0\in \mathbf{X};
$$
where, once again abusing notation, 
\begin{multline*}
\mathcal{U}_{i}^{G}(\mathbf{x}_0,\widehat{\bvarphi}_{-i})=\Big\{\mathbf{u}_{i}(\cdot)\in L^{1}_{\rho_{i}}(\R_+;\mathbf{U}_{i}):\\ \ \ \ \ \ \ \ \ \ \ \ \ \ \ \   \mathbf{g}_{j}(t,\mathbf{x}(t))\leq \mathbf{0}, \ \ \mathbf{l}_{j}(t, \widehat\bvarphi_{-i}(t,\bx(t)), \mathbf{u}_{i}(t))\leq \mathbf{0},  \ \ \ \forall j=1,...,N, \ \ \forall t\in\R_{+}, \\ \ \ \mbox{and} \ \    t\mapsto h_{j}(t, \widehat\bvarphi_{-i}(t,\bx(t)), \mathbf{u}_{i}(t))\in L^{1}_{\rho_{j}}(\R_{+}) \ \forall j=1,...,N\Big\}.
\end{multline*}
The latter is exactly  the condition  of equilibrium required in  \citealp[Def.\,4.1]{dockner2000differential}.\hfill$\square$
\end{Remark}
\begin{Remark}[Open-loop and Markovian Nash equilibria]
By definition, an open-loop Nash equilibrium is also a (degenerate) Markovian Nash equilibrium. If the problem is stationary; i.e., the data $\mathbf{A},\mathbf{f}, \mathbf{g}_{i},\mathbf{l}_{i},\mathbf{a}_{i}, h_{i}$ do not depend on $t$, then by Theorem \ref{th:staticreductionNash}, an open-loop Nash equilibrium, if it exists, will not depend on $t$, as the temporary problem is time independent. Thus, it will be a stationary Markovian equilibrium.
\end{Remark} 
\begin{Remark}[Markov Perfect Equilibrium]
In the context of Markovian Nash equilibria, one may inquire whether a stronger property holds. Starting the game from any initial time $t_{0}$ and any initial condition $\mathbf{x}_{0}$, does the equilibrium map $\widehat\bvarphi$, when restricted to $[t_{0},\infty)$, continue to be a Markovian equilibrium for the game starting at $(t_{0},\bx_{0})$? If the answer to this question is affirmative, we say that $\widehat\bvarphi$ is a \emph{Markov Perfect Equilibrium (MPE)}. In our setup, in the absence of state constraints, the \emph{degenerate Markov equilibrium} (if it exists) only depends on time, as it is the solution to a family of temporary problems, hence it is \emph{automatically} an MPE.\hfill$\square$
\end{Remark}

Next, we investigate whether a uniqueness result holds for Markovian Nash equilibria. The next Proposition provides a qualified confirmation. The proof is given in Appendix \ref{app:reduction}.

\begin{Proposition}[Uniqueness of affine Markovian Nash equilibria] \label{uniqueaffine}
Suppose that:
\begin{itemize}
\item[(i)]  The map $\mathbf{f}$ is affine in $\bu$; i.e., 
$$
\mathbf{f}(t,\bu)=  \mathbf{P}(t)\mathbf{u}+ \mathbf{j}(t), \ \ \ \ \ \ \ \mathbf{P}(t)\in \mathcal{L}(\mathbf{U},\mathbf{X}), \ \ \ \mathbf{j}(t)\in\mathbf{X};$$ 
\item[(ii)] There are no state or control constraints; i.e.,
$$\mathbf{g}_i,\mathbf{l}_{i}\equiv \mathbf{0};$$ 
\item[(iii)] For each $i=1,...,N$, the function $h_{i}$ only depends on $(t,\bu^{i})$ and is strictly concave and coercive, for each $t\in\R_{+}$.
\end{itemize}
Then, the differential game admits at most one MPE in the class of time-dependent affine Markovian feedbacks:\smallskip
\begin{align*}
\mathcal{M}_{G}^{L}&=\bigg\{\bvarphi=(\bvarphi_1,...,\bvarphi_{N}):\R_{+}\times \mathbf{X}\to \mathbf{U} \, :\  \ \bvarphi_{i}(t,\bx)=\mathbf{L}_{i}(t)\bx+\mathbf{w}_{i}(t),\\& \ \ \ \ \ \ \emph{where} \  \mathbf{L}_{i}: \R_{+}\to \mathcal{L}(\mathbf{X},\mathbf{U}_{i})\ \emph{and} \  \mathbf{w}_{i}: \R_{+}\to \mathbf{U}_{i}\  \emph{bounded} \ \ \forall  i=1,...,N\bigg\}. 
\end{align*}
\end{Proposition}

Next, we will showcase the power of the ITM in the context of a specific economic application.

\section{An analytical integrated assessment model}
\label{sec:model}

The analytical advantages of the ITM over standard optimization methods can be exploited in a variety of dynamic models in economics and beyond. In this section, we will illustrate the method in the context of an application to climate economics. Provided that the conditions for the applicability of the ITM hold, a variety of climate models could be used for the illustration. Here, we will employ a version of the integrated assessment model in \cite{Golosov2014}. Our analysis will extend their basic model in several directions that might be of independent interest, including introducing multiple heterogenous regions, technological progress, strategic considerations, and deep (Knightian) uncertainty.

\subsection{Preferences and technology}

We consider two infinitely lived countries, $1$ and $2$. Country $i$, $i\in \{1,2\}$, chooses a consumption flow $C_{i}(t)$ earning a payoff $u_{i}(C_{i}(t))$ for $t\in[0,+\infty)$. The instantaneous utility functions $u_i$ are assumed to be strictly increasing, strictly concave and satisfy the usual Inada conditions. The common discount factor is $\rho>0$. To fix ideas, we consider country 1 to be representative of the ``global north," while country 2 represents the ``global south." The lifetime payoff of each country $i$ is given by:\footnote{We thus use a utility linear damage specification; see, for example, \cite{withagen1994pollution} and \cite{tahvonen1997fossil}.}
\begin{equation}  \label{eq:defUi}
U_{i}=\int_{0}^{\infty }e^{-\rho t}
\left[u_{i}\left( C_{i}(t)\right)
- \gamma_{i}\left( S(t)-\overline{S}\right)\right] dt
\end{equation}
where the variable $S(t)$ stands for the total stock of greenhouse gas emissions (GHG) relative to the pre-industrial level, $\overline{S}$. As different countries have different degrees of vulnerability to climate change, the parameters $\gamma_{i}>0$ capture the relative sensitivity of each country's payoff to GHG concentrations. The climate sensitivity parameter can be the result of a country's geography, but can also capture the ability to engage in adaptation. Of course, there are several different ways to model climate damage. Our formulation can be thought of as a reduced form of the one used, for example, in \cite{van2012there}. The additive linear structure can be justified as an approximation of the composition between the mapping from GHG emersions to temperatures, which is concave, and a convex mapping from temperatures to actual damages. Flow output can be produced by using an input, $K_i$, according to
\begin{equation}
Y_{i}(t)=A_{i}(t)f_i(K_i(t)).
\end{equation}
The functions $f_i$ are assumed to be differentiable, increasing, and concave. 
We assume that the input is ``dirty;" i.e., its use creates a flow of GHG emissions. For simplicity, folllowing \cite{Golosov2014}, we will assume that the input depreciates completely after it is used in production. Each country chooses an abatement effort $B_i(t)$ towards reducing the stock of GHG emissions. Like with the climate-sensitivity parameter, the abatement technology can capture several factors. For example, it might include reforestation efforts, carbon capture and storage systems, etc.

Next, we discuss the production technology and introduce policy interventions through a variety of transfers between the two countries. The total factor productivity (TFP) parameters, $A^{i}(t)$, in the two countries are as follows. For country 1 (the global north) we have
\begin{equation}
\label{eq:defA1}
A_{1}(t) = \overline{A}(t),
\end{equation}
where $\overline{A}(t)$ stands for the (exogenous) technological frontier, which is assumed to be continuous. For country 2 (the global south) we have  
\begin{equation}
\label{eq:defA2}
A_{2}(t) = h(R_a(t)) \overline{A}(t).
\end{equation}
In the above expression, $R_a$ stands for a production technology-specific transfer from country 1. The function $h$ is assumed to be positive, strictly increasing, and concave. Furthermore, we assume that $\lim_{R\to\infty}h(R) \in (0,1)$ (so that the TFP in country 2 is always lower than in country 1). Along the same lines, we assume that the efficiency of the abatement effort in country 2 may depend on a technology transfer from country 1. Henceforth, $R_b$ represents an abatement-specific technology transfer from country 1, to be used in improving the effectiveness of the abatement technology in country $2$.\footnote{We thus assume that the know-how needed to improve TFP and the abatement technology in country 2 must originate in country 1. Allowing for technology improvements from related investments in country 2 would add additional decision variables without significantly contributing to the issues we investigate here.} Flow output in the two countries is respectively given by
\[Y_1(t) = \overline{A}(t)f_1(K_1(t)), \quad  {\rm and }\; \quad  Y_2(t) = h(R_a(t)) \overline{A}(t)f_2(K_2(t)).\]
Aggregate feasibility requires that output in each country equals the respective total amount of resources used; i.e.,
\[
Y_1=K_1 + C_1 + B_1 + R_a + R_b; \qquad  Y_2=K_2 + C_2 + B_2.
\]

\subsection{The climate model}\label{subsect:theclimatemodel}

We follow the approach in \cite{Golosov2014} in modeling the two-way interplay between climate and economic activity. Building on the models described in \cite{nordhaus2003warming}, this approach incorporates explicitly the increase in GHG, and implicitly the effects of carbon sinks like the terrestrial biosphere and shallow and deep oceans. The modeling allows for nonlinear absorption of atmospheric carbon, but it abstracts from the delays of the economic impact of this carbon content and it does not separately keep track of the dynamics of different GHG. Importantly, this approach concentrates on temperatures and abstracts from the effects of precipitation.\footnote{As mentioned earlier, our modeling contribution and qualitative findings do not depend on the details of the climate model employed and we use the \cite{Golosov2014} model as an illustration.} 

The evolution of $S(t)$ depends on the aggregate use of the ``dirty" production input as well as on the aggregate investment in abatement. We assume that a fraction $\phi_{L}$ of emitted carbon stays permanently in the atmosphere, while a fraction $(1-\phi _{0})$ of the remaining emissions exits into the biosphere and the remaining part decays at geometric rate $\phi$. We use $P(t)$ and $T(t)$, respectively, to indicate the permanent and the temporary components of the total emissions, $S(t)$. Given the pre-industrial level of GHG\
concentration in the atmosphere $S(0)=\overline{S}=P(0)+T(0)$, we then have:

\begin{eqnarray}
P'(t) &=&\phi _{L}G\left(t, K_{1}(t),K_{2}(t),B_{1}(t),B_{2}(t),R_b(t)\right)
\label{eq:P} \\
T'(t) &=&-\phi T(t)+(1-\phi _{L})\phi _{0}
G\left(t,K_{1}(t),K_{2}(t),B_{1}(t),B_{2}(t),R_b(t)\right)  \label{eq:T} \\
S(t) &=&P(t)+T(t)  \label{eq:S}
\end{eqnarray}

The function $G:\mathbb{R}_{+}^{6}\rightarrow \mathbb{R}$ specifies how the use of the dirty input and the abatement technologies in the two countries affect the flow of emissions. We assume that $G$ is continuous, strictly increasing in the use of dirty capital, $K_i$, strictly decreasing in the abatement-related variables, $B_i$ and $R_b$, and that it has at most linear growth in $K_1,K_2,B_1,B_2$, and $R_b$.

\subsection{A special case}
\label{sub:plannerlinear}

For illustration purposes, we shall consider a special case of models where explicit results can be readily derived. This is specified in the following. 

\medskip

\medskip

\begin{Hypothesis} [A special case]
\label{hp:easycase}
\begin{itemize}
Suppose that:
\item[(i)] The map $G$ is independent of $t$ and given by:
$$
G(K,B,R_b)=
\eta_K (K_1+K_2){-} \eta_B (B_1(t)^{\theta_1} 
+ g(R_b) B_2(t)^{\theta_2}),
$$
where $\theta_1, \theta_2 \in (0,1)$ 
and $\eta_K,\eta_B>0$.

\medskip

\item[(ii)] The production function is linear: $f_i(K_i) = K_i$, so that
\begin{equation}
Y_{i}(t)=A_{i}(t)K_{i}(t),
\end{equation}
where $A_i(t)>1$ are as in (\ref{eq:defA1})-(\ref{eq:defA2}).

\medskip

\item[(iii)] The production technological frontier is constant:
\[
\overline{A}(t) \equiv \overline{A}>1.
\] 

\medskip

\item[(iv)] For $\sigma_1, \sigma_2 >0$, the instantaneous payoff functions $u_i$ are given by 
\[
u_i(C_i(t)) = \frac{C_i^{1-\sigma_i}}{1-\sigma_i}, \qquad \text{or} \qquad
u_i(C_i(t)) = \ln (C_i) \quad \text{(logarithmic case)}.
\]

\medskip

\item[(v)] The function $h$ is in $C^2([0,+\infty), (0,1))$, with $h'>0$, $h''<0$, and such that  $\lim_{R\to\infty}h(R) \in (0,1)$ and the same holds for the function $g$. Moreover, $\overline{A} h(0)>1$, and the map
$R_b\mapsto g(R_b)^{\frac{1}{1-\theta_2}}$ is strictly concave.
\end{itemize}
\end{Hypothesis}
We will restrict attention on cases where $G>0$. The concavity of $g(R_b)^{\frac{1}{1-\theta_2}}$ guarantees the uniqueness of the solution in what follows. We emphasize again that, as $K$ is a control variable in our model, our method allows for the functions $f_i$ to be strictly concave. We next specify and characterize the normative and positive arrangements that will be considered henceforth.\footnote{The above conditions will imply the existence and uniqueness of a solution to the auxiliary problem and the uniqueness of solutions to the social planners' problem as well as the uniqueness of open-loop Nash equilibria for the cases we will investigate next. Although we assume differentiability throughout for expository purposes, our results do not require differentiability and can be demonstrated using convex optimization techniques.}

\section{Normative and positive investigations}
\label{sec:institutionalarrangements}

In what follows, we will illustrate the ITM method  in the context of the analytical integrated assessment model in section 3. We first consider two normative benchmarks by characterizing the solutions to two social planner problems. We then study the Nash equilibria of a suitable non-cooperative dynamic game. We will discuss the reformulation given in Proposition \ref{pr:rewriting-general}, which lies at the heart of the ITM, in some detail in the context of the ``global planner" problem. As the same steps apply, we will skip the details for the other cases.

\subsection{The Global planner's (GP) problem}  
\label{sub:planner}

We will first consider the problem of a benevolent ``global planner (GP)" who has control over resources and production in both countries and who can freely transfer resources from one country to the other. Clearly, this defines an extreme normative benchmark, as it abstracts from any strategic considerations between the two countries, as well as mobility constraints, transportation costs, etc. 

The GP chooses: $C_1, C_2, B_1, B_2, R_a, R_b, K_1, K_2$, under the following single resource constraint:
\begin{multline}
\label{eq:plannerresourceconstraint1}
C_1(t)+ C_2(t)+ B_1(t)+ B_2(t)+ R_a(t)+ R_b(t)+ K_1(t)+ K_2(t) \\
\le Y_1(t) + Y_2(t) =  \overline{A}(t)f_1(K_1(t)) + h(R_a(t)) \overline{A}(t)f_2(K_2(t)).
\end{multline}
We assume that the GP maximizes the sum of the two countries' objectives: 
\begin{multline}
\label{eq:defUplannernew}
U^{P}=U_1+U_2=
\int_{0}^{\infty }e^{-\rho t}
\left[u_{1}\left( C_{1}(t)\right)+ u_{2}\left( C_{2}(t)\right)
- \gamma_{1}\left( S(t)-\overline{S}\right)
-  \gamma_{2}\left( S(t)-\overline{S}\right)
\right] dt \\
= \frac1\rho(\gamma_{1}+\gamma_2)\overline{S} +
\int_{0}^{\infty }e^{-\rho t}
\left[u_{1}\left( C_{1}(t)\right)+ u_{2}\left( C_{2}(t)\right)
- (\gamma_{1}+\gamma_2) S(t)\right] dt.
\end{multline}

To unburden the presentation, we will occasionally abuse notation and use $C$ to indicate the vector $(C_1, C_2)$, $R$ to indicate $(R_a, R_b)$, etc. Formally, the above optimal control problem is characterized by: (i) the state equation given by the system (\ref{eq:P})-(\ref{eq:S}), and (ii) the set of admissible policies\footnote{We denote by $L_{\rho}^1(\mathbb{R}_+)$ the set $\left \{ 
f\colon \mathbb{R}_+ \to \mathbb{R} \; : \; \int_0^\infty e^{-\rho t} \mid f(t) \mid  dt <\infty\right \}$, while $L_{loc}^1$ stands for the set of functions that are locally Lebesgue-integrable. The set of admissible strategies is chosen in a way that guarantees that the state equation and the objective functional are well-defined.}
\begin{multline}
\label{eq:admstratplanner1}
\mathcal{U}^{p_1} :=
\Bigg \{
C, B, K, R \in L_{loc}^1(\mathbb{R}_+;  \mathbb{R}^8_+) \; :
t \mapsto u_1(C_1(t)) \in L_{\rho}^1(\mathbb{R}_+), \\ t \mapsto u_2(C_2(t)) \in L_{\rho}^1(\mathbb{R}_+), t\mapsto S(t) \in L_{\rho}^1(\mathbb{R}_+)\\
\text{ and } (\ref{eq:plannerresourceconstraint1}) \text{ holds for all $t\geq 0$}
\Bigg \}
\end{multline}
and (iii) the objective functional given by (\ref{eq:defUplannernew}).

We will apply the ITM to this problem by applying Proposition \ref{pr:rewriting-general}. We begin by defining:  
\begin{equation}
\label{eq:defplannerxu}
\mathbf{x}(t) := \begin{pmatrix} P(t) \\ T(t) \\ \end{pmatrix},\quad \mathbf{u}(t) := \left(C(t),B(t),K(t),R_a(t),R_b(t)\right)^{T},
\end{equation}
\begin{equation}
\label{eq:defplannerAf}
\mathbf{A}(t) := \begin{pmatrix} 0 & 0 \\ 0 & -\phi \\ \end{pmatrix}, \quad \mathbf{f}(t,\mathbf{u}(t)) := \begin{pmatrix} \phi_L G\left(t,K(t),B(t),R_b(t)\right) \\ (1-\phi_L)\phi_0 G\left(t,K(t),B(t),R_b(s)\right)  \end{pmatrix},
\end{equation}
\begin{equation}
\label{eq:defplannerah}
\mathbf{a}(t) := -(\gamma_1+\gamma_2)\begin{pmatrix} 1 \\ 1 \\ \end{pmatrix},\quad  h(t,\mathbf{u}(t)) := u_1(C_1(t))+u_2(C_2(t)) + (\gamma_1+\gamma_2)\overline{S},
\end{equation}
\begin{equation}
\label{eq:defplannerg}
\mathbf{g}(t,\mathbf{x}(t)) := 0,
\end{equation}
and
\begin{multline}
\label{eq:defplannerl}
\mathbf{l}(t,\mathbf{u}(t)) := (-C(t),-B(t),-K(t),-R_a(t),-R_b(t),\\ C(t)+B(t)+K(t)+R_a(t)+R_b(t) - \overline{A}(t)f_1(K_1(t))+h(R_a(t))\overline{A}(t)f_2(K_2(t)))^{T}.
\end{multline}
With these choices, the abstract Problem defined in \eqref{eq:stategeneral}-\eqref{eq:ADMCONTRgeneral} reduces to Problem  \eqref{eq:plannerresourceconstraint1}-\eqref{eq:admstratplanner1}. Moreover:
\begin{equation*}
\Phi_A^{\ast}(t+\tau,t) = \exp\left(\tau A\right) = \begin{pmatrix} 1 & 0 \\ 0 & e^{-\phi \tau} \\ \end{pmatrix},
\end{equation*}
and, therefore,
\begin{equation*}
\mathbf{b}(t) = -(\gamma_1+\gamma_2) \begin{pmatrix} {1}/{\rho} \\ {1}/{(\rho+\phi)} \\ \end{pmatrix}.
\end{equation*}
As discussed earlier, the map $t \to \mathbf{b(t)}$ is a crucial component of the ITM. It captures the impact on the objective function of the future evolution of the state variables, as well as their marginal (exogenous) impact on the intertemporal payoffs. The problem studied in this section provides an instance of this principle. In this example, the first element of $\mathbf{b(t)}$, $\frac{(\gamma_1+\gamma_2)}{\rho}$, captures the discounted intertemporal disutility stream generated by an extra unit of GHG added to the permanent stock, $P(t)$. The second element, $ \frac{(\gamma_1+\gamma_2)}{\rho+\phi}$, gives the analogous expression for the addition of an extra unit of GHG to the transitory stock, $T(t)$, where $\phi$ is the rate at which GHG emissions decay over time.

Returning to the global planner's problem, note that Assumption \ref{hp:reduction} is satisfied, due to the assumptions in Section \ref{subsect:theclimatemodel}. We can thus apply directly Proposition \ref{pr:rewriting-general} yielding the following reformulation of the GP's objective functional:

\begin{multline}
\label{eq:defUPbisnew}
U^P=
(\gamma_{1}+\gamma_2)\left[
\frac{\overline{S}}{\rho} -
\frac{P(0)}{\rho}-\frac{T(0)}{\rho+\phi}
\right]+
\\[2mm]
\notag
\int_{0}^{\infty }e^{-\rho t}
\left[u_{1}\left(C_1(t)\right)+ u_{2}\left(C_2(t)\right)
-(\gamma_{1}+\gamma_2) \Phi G\left(t,K(t),B(t),R_b(t)\right)\right]dt
\end{multline}
where
\begin{equation}
\label{eq:defPhi}
\Phi:=
\left[
\frac{\phi_{L}}{\rho}+
\frac{(1-\phi _{L})\phi _{0}}{\rho+\phi}
\right].
\end{equation}

Applying Theorem 2.5, we can now characterize the outcomes of the GP's dynamic optimization problem.
\begin{Corollary}
\label{cor:rewriting-planner-max} 
Assume that, for every $t\ge 0$, $(C^*(t), B^*(t), K^*(t), R^*(t))$ solves the following temporary optimization problem:
\begin{equation}
\label{eq:staticglobalplannerSpecific}
\max_{C, B, K, R} \; \Big  [ u_1(C_1) +  u_2(C_2) - (\gamma_{1}+\gamma_2) \Phi G(t,K,B,R_b) 
\Big ]
\end{equation}

over the feasible set $\mathcal{E}^{GP}(t)\subseteq \R^8$ described by the inequalities 
\begin{equation}
\label{eq:StaticConstraintsPlannerSpecific}
\left \{
\begin{array}{l}
C_1, C_2, B_1, B_2, K_1, K_2, R_a, R_b \geq 0,\\[8pt]
C_1 + C_2 + B_1 + B_2 + K_1 + K_2 + R_a + R_b \leq
\overline{A}(t) \left [f_1(K_1) + h(R_a) f_2(K_2) \right ],
\end{array}
\right .
\end{equation}
where 
\begin{equation*}
\Phi:=
\left[
\frac{\phi_{L}}{\rho}+
\frac{(1-\phi _{L})\phi _{0}}{\rho+\phi}
\right].
\end{equation*}
In addition, assume $(C^*(t), B^*(t), K^*(t), R^*(t))$ is admissible for the original dynamic optimization problem. Then the map $t \mapsto (C^*(t), B^*(t), K^*(t), R^*(t))$ solves the GP's problem.
Conversely, a solution to the GP's problem constitutes, for a.e. $t\ge 0$, a solution to the temporary problem above.
Consequently, if the solution to the temporary problem is unique for a.e. $t\ge 0$, then the solution to the GP's problem is a.e. unique.
\end{Corollary}
\begin{proof}
See Appendix \ref{subapp:generalplannerspecificcas}.
\end{proof}


Clearly, a global planner who can costlessly allocate resources across the two countries would not choose to produce in country 2, as country 1 is more efficient.\footnote{Note that $R_b = 0$ in \eqref{eq:RbB2plannernew-bis-ultra} and \eqref{eq:RbB2plannernew-bis-ultra-log}
if and only if $\frac{\eta_K}{\eta_B} \geq g(0)^{\theta_2} g'(0)^{1-\theta_2} \theta_2^{{\theta_2}} (\overline{A} -1)$. To see this, observe that, in that case, the derivative of (\ref{eq:chesevealremark}) evaluated $R_b=0$ is positive. Conversely, when $\frac{\eta_K}{\eta_B} < g(0)^{\theta_2} g'(0)^{1-\theta_2} \theta_2^{{\theta_2}} (\overline{A} -1)$, the interior maximum point $\bar R_b$ satisfies: $\frac{\eta_K}{\eta_B} =  \left [ g(\bar R_b)^{\theta_2} g'(\bar R_b)^{1-\theta_2} \theta_2^{{\theta_2}} (\overline{A} -1) \right ]$.} The interpretation is straightforward. When the marginal abatement efficiency ($\eta_B$) is small relative to the marginal addition to GHG emissions generated by production ($\eta_K$), the GP will refrain from subsidizing abatement in country 2. 

\subsection{The restricted planner's (RP) problem}
\label{subsec:plannerNOtransfer}

As a second normative benchmark we consider the more relevant case of a world planner who cannot directly move resources across the two countries. The planner can still invest in producing and in improving the abatement technology through choosing positive $R_a$ and $R_b$. In this case we find that, as is more natural, production takes place in both countries. As before, the RP planner maximizes the functional in \eqref{eq:defUplannernew}. The difference is reflected in the constraints, as we now replace the single resource constraint (\ref{eq:plannerresourceconstraint1}) with the following  two constraints, one for each country:
\begin{equation}
\label{eq:planner-notrasnfer-resourceconstraintmain}
\left \{
\begin{array}{l}
C_1(t)+ B_1(t)+ K_1(t)+ R_a(t)+ R_b(t)\le Y_1(t)= \overline{A}(t)f_1(K_1(t))\\ 
C_2(t)+ B_2(t)+ K_2(t) \le Y_2(t) =
\overline{A}(t)h(R_a(t))f_2(K_2(t)).
\end{array}
\right .
\end{equation}
This implies that the set of admissible policies for the RP's problem is now the set 
$\mathcal{U}^{RP}$, which is defined exactly as $\mathcal{U}^{GP}$ in \eqref{eq:admstratplanner1},
but with the constraints \eqref{eq:planner-notrasnfer-resourceconstraintmain} in the place of \eqref{eq:plannerresourceconstraint1}.

Similarly to Corollary  \ref{cor:rewriting-planner-max}, we can specify Theorem \ref{th:staticreduction} for this case as follows.
\begin{Corollary}
\label{cor:rewriting-planner-notransfer} 
Assume that, for every $t\ge 0$, $(C^*(t), B^*(t), K^*(t), R^*(t))$ solves the following temporary optimization problem:
\[
\max_{C, B, K, R} \; \Big  [ u_1(C_1) + u_2(C_2) - (\gamma_1 + \gamma_2)\Phi G(t,K,B,R_b) 
\Big ]
\]
over the feasible set $\mathcal{E}^{RP}(t)\subseteq \R^8$ defined by the inequalities: 
\begin{equation}
\label{eq:StaticConstraintsPlannerNoTransfer}
\left \{
\begin{array}{l}
C, B, K, R \geq 0\\[8pt]
C_1 +  B_1  + K_1  + R_a + R_b \leq
\overline{A}(t) f_1(K_1)
\\[8pt]
C_2 + B_2 + K_2  \leq
\overline{A}(t) h(R_a) f_2(K_2)
\end{array}
\right .
\end{equation}
In addition, assume $(C^*(t), B^*(t), K^*(t), R^*(t))$ is admissible for the original dynamic optimization problem. Then the map $t \mapsto (C^*(t), B^*(t), K^*(t), R^*(t))$ solves the RP's problem.
Conversely, a solution to the RP's problem constitutes, for a.e. $t\ge 0$, a solution to the above temporary problem.
Consequently, if the solution of the temporary problem is unique for a.e. $t\ge 0$, then the solution to the RP's problem is a.e. unique.
\end{Corollary}

Like before, one can return to the illustrative special case described in Assumption \ref{hp:easycase}; see Appendix \ref{subapp:restrictedplannerspecificcas}, in particular Proposition \ref{pr:planner-notrasnfer-easy}.

\subsection{The non-cooperative dynamic game}
\label{subsec:Nash}

Normative issues aside, an important question concerning implementation is whether outcomes of interactions between self-interested countries are likely to be efficient. 
To investigate these questions, we now turn to the study of Nash equilibria of an underlying non-cooperative game between the two countries. 
As before, we first investigate the game keeping $G$, $z_i$ $f_i$, $u_i$, $g$, and $h$ in a general form. Player 2 (the global South) takes $R_a$ and $R_b$ as given, as we assume that they are chosen by country 1 (the global North). Given the strategy of the other player, each country $i$ maximizes its own payoff, which can be written as
\begin{equation}
\label{eq5:defUPbisgameNashtransfer}
U_i=
\gamma_i\frac{\overline S}{\rho}+ \int_{0}^{\infty }e^{-\rho t}
u^{i}\left(C_i(t)\right)dt
-\gamma_{i} \int_{0}^{\infty }e^{-\rho t} (P(t)+T(t)) dt.
\end{equation}

We first discuss how the general result in Theorem \ref{th:staticreductionNash} applies in this specific setup. The proof can be found in Appendix \ref{subapp:Nashspecificcas}.

\begin{Corollary}
\label{cor:rewriting+conditions-Nash}
Assume that, for every $t\ge 0$, $(C^*(t), B^*(t), K^*(t), R^*(t))$ is a Nash equilibrium of the temporary game where:\footnote{In order to rule out outcomes where (\ref{eq:vincolo1gameplayer1}) or (\ref{eq:vincolo2gameplayer2}) do not hold, we assume that the payoff to either player is $-\infty$ if their budget constraint is violated.}
\begin{itemize}
\item[(i)] Given $B_2,K_2\geq 0$, Country 1 chooses $(C_1,B_1, K_1, R_a, R_b) \geq 0$ to maximize 
\begin{equation}
\label{eq:H1-per-Nash}
H^1(t,C_1,B_1,K_1,R_a,R_b):= u_{1}\left(C_1\right) -\gamma_{1}\Phi 
G\left(t,K(t),B(t),R_b(t)\right),
\end{equation}
under the constraint:
\begin{equation}
\label{eq:vincolo1gameplayer1}
C_1+B_{1}+R_a+R_b+I_1 \le Y_1(t)=\overline{A}(t) f_1(K_1).
\end{equation}
\item[(ii)] Given $B_1,K_1, R_a, R_b\geq 0$, country 2 chooses $(C_2,B_2,K_2) \geq 0$ to maximize
\begin{equation}
\label{eq:H2-per-Nash}
H^2(t,C_1,B_1,K_1,R_a,R_b):= u_{2}\left(C_2\right) -\gamma_{2}\Phi 
G\left(t,K(t),B(t),R_b(t)\right)
\end{equation}

under the constraint:
\begin{equation}
\label{eq:vincolo2gameplayer2}
C_2+B_{2}+K_2 \le Y_2(t)=\overline{A}(t)h(R_a)f_2(K_2).
\end{equation}
\end{itemize}

In addition, assume that $(C^*(t), B^*(t), K^*(t), R^*(t))$ is an admissible strategy for the original differential game.
Then the map $t \to (C^*(t), B^*(t), K^*(t), R^*(t))$ is an open-loop Nash equilibrium to the original differential game. Conversely, every Nash equilibrium of the original differential game is, for a.e. $t\ge 0$, a
Nash equilibrium for the above temporary game.
Consequently, if the Nash equilibrium of the temporary game is unique for a.e. $t\ge 0$, then the Nash equilibrium to the original differential game is a.e. unique.
\end{Corollary}
\begin{proof}
See Appendix \ref{subapp:Nashspecificcas}
\end{proof}

Like before, additional insights, as well as necessary and sufficient conditions for $R_{b}>0$, can be obtained if we impose further structure on the problem.\footnote{See Proposition \ref{pr:Nash-easy} in Appendix \ref{subapp:Nashspecificcas}.}




\section{Numerical explorations}  

An advantage of the ITM is that it allows for analytical comparisons. Indeed, explicit comparison results are provided for the specific cases described in Hypothesis \ref{hp:easycase} under logarithmic utility. To unburden the exposition, we report these comparisons in Appendix \ref{app:proofs}. Our tractable dynamic framework can be used for quantitative analysis. Although we will not pursue a calibration here, in this section we will use a numerical example to illustrate some quantitative features of our modeling approach. The relevant parameters we use are listed in Table 1 in the Appendix.\footnote{The codes used in this section are freely available at: \url{https://github.com/crricci/climate_change_optimal/}.} We will focus attention on various comparisons between the solutions to the global planner (GP), restricted planner (RP) and Nash (N) solutions for several variables of interest. We will pay particular attention to the GHG emissions and temperatures paths under different scenarios, including under heterogeneous damages. 

We map carbon concentrations into global temperatures, $T$, using the following expression; see, for example, Golosov et al. (2014):
\begin{equation}
Temp(S_{t})=3\ln \left( \frac{S_{t}}{\overline{S}}\right) /\ln (2) \nonumber,
\end{equation}%
where $\overline{S}$ is the pre-industrial level of the GHG concentration. In what follows, we will focus on the sensitivity of various outcomes to (i) the intertemporal elasticity of substitution, and (ii) the North-South heterogeneity.

\subsection{Intertemporal elasticity of substitution}
Here we investigate the sensitivity of model outcomes to the value of the intertemporal elasticity of substitution, $\sigma$. In the next section we will investigate the effects of heterogeneity, including in the value of discounting, which is an important parameter in climate economics.\footnote{The sensitivity with respect to $\sigma$ is also discussed in \cite{Golosov2014}, who concentrate on the logarithmic case for their theoretical investigations.} Using values from our benchmark parametrization  (see Appendix \ref{app:Tables}), here we vary $\sigma$ in $[0.1,2]$. Figure \ref{fig:differentSigma} illustrates the differences in net GHG emissions and in temperatures between the Nash case and the GP solution (first-best).  

\begin{figure}[!htpb]
\centering

\begin{subfigure}{0.35\textwidth}  
\includegraphics[width=\textwidth]{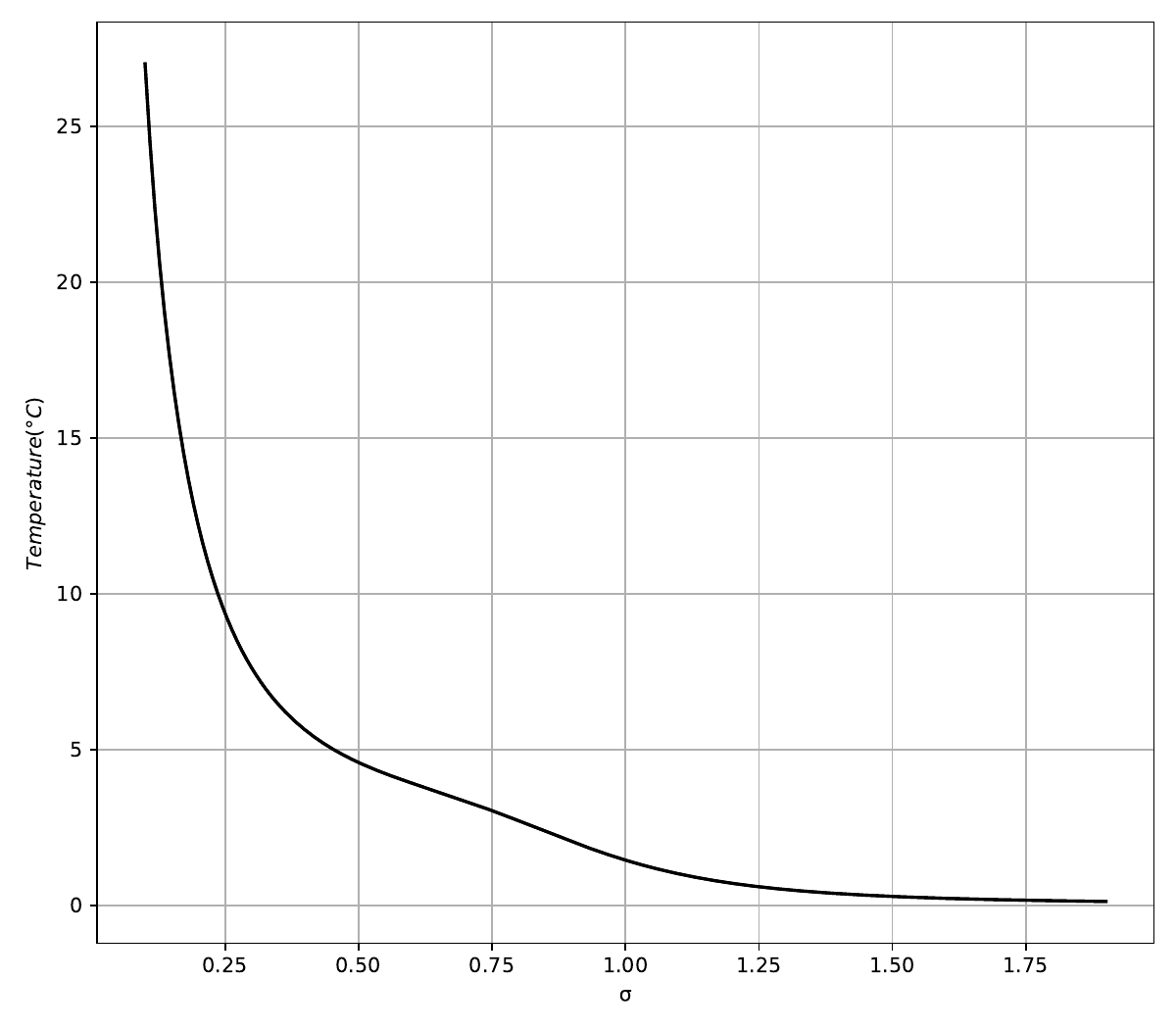}
\caption{Absolute change in the final temperature}
\label{fig:abschangeTsigma}
\end{subfigure}
\begin{subfigure}{0.35\textwidth}  
\includegraphics[width=\textwidth]{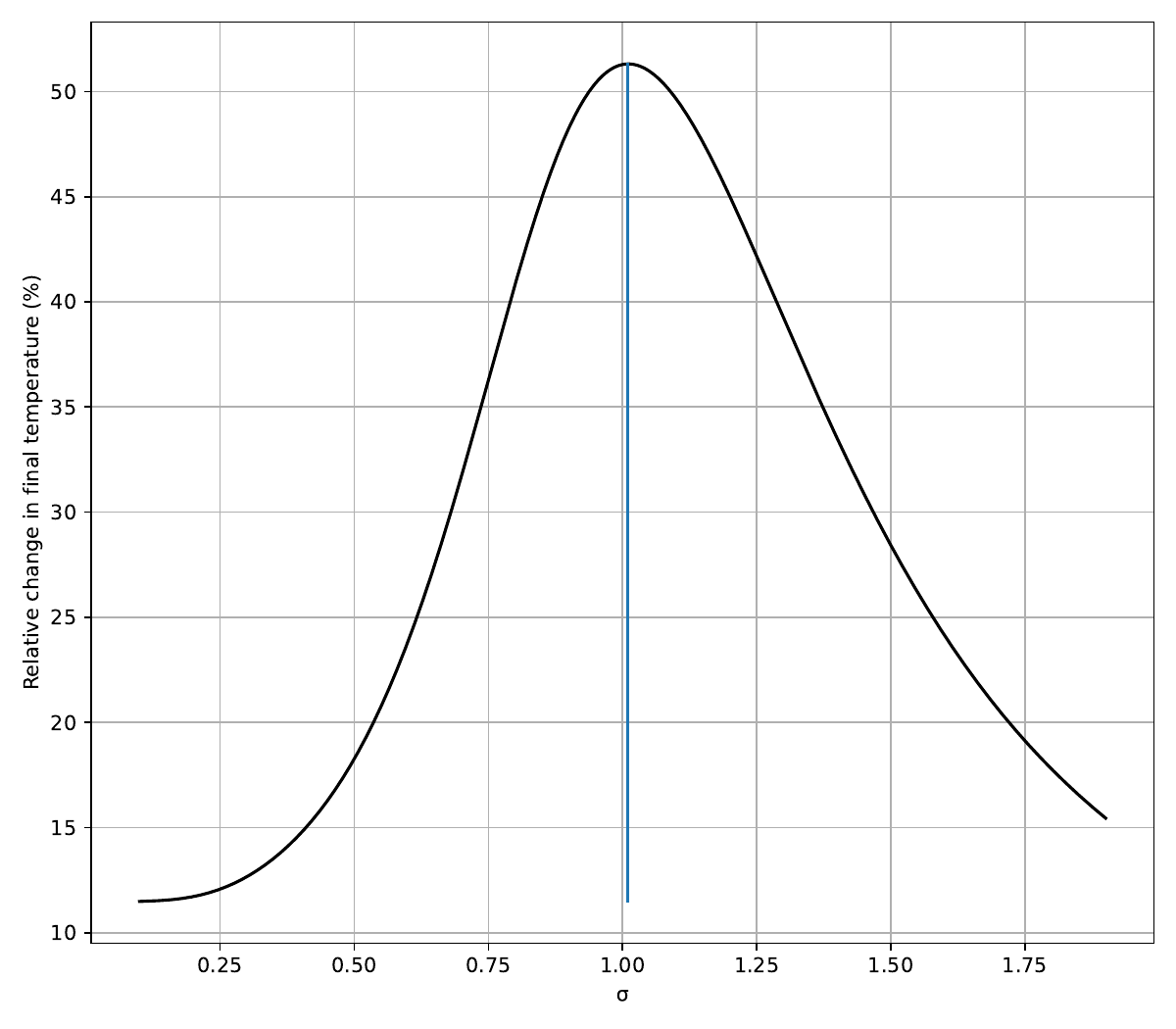}
\caption{Relative change in the final temperature (\%)}
\label{fig:relchangeTsigma}
\end{subfigure}
\caption{Absolute vs relative Nash inefficiency gap in terms of the final temperature when $\sigma$ varies.}
\label{fig:differentSigma}
\end{figure}

As $\sigma$ increases, consumption is smoother and both net emissions and temperatures decrease with $\sigma$. When $\sigma$ is small, say below $0.5$, the absolute deviation of the temperature in the Nash case with respect to the first best is larger than $5$ degrees (of course, it becomes very large when $\sigma$ goes to zero). This inefficiency gap drops to less than $2$ degrees close to the logarithmic case, and then it decreases quickly to zero for $\sigma \leq 1.5$. It turns out that the difference between the first-best and the Nash outcome are higher in the neighborhood of the logarithmic case (about $50\%$ around $\sigma=1$). This suggests that large differences in relative efficiency between the Nash and the GP outcome can result when $\sigma$ is in the neighborhood of the logarithmic case.

\subsection{Heterogeneity}

Here we briefly explore the sensitivity to heterogeneity with respect to the time discount factor ($\rho$) and with respect to the parameter measuring relative vulnerability to climate change ($\gamma$).

\subsubsection{Discounting}
Time discounting is an important parameter in climate economics, as many of the GHG-related damages occur in the future; see, for example, \cite{stern2007economics}. Few game-theoretic models in this area have investigated the role of heterogeneity in time discounting.\footnote{One recent exception is \cite{vosooghi2022self}, however, the authors do not investigate equilibrium dynamics.} Here, we will consider the autonomous benchmark case and illustrate that our methodology can accommodate this type of heterogeneity at low mathematical and computational costs.  The results also illustrate the intrinsically dynamic nature of our methodology. As we will see, the planners' problems internalize the heterogeneity in discounting rates, leading to endogenous dynamics as the optimal controls will be time-dependent. 

To illustrate this point, define
\begin{equation*}
\Phi(\rho) = \frac{\phi_{L}}{\rho} + \frac{(1-\phi_{L}\phi_{0})}{\rho+\phi},
\end{equation*}
where $\phi,\phi_{L},\phi_{0}$ are the climate parameters defined earlier. \\
Then the planners' problem (the GP and the RP only differ in their constraints) objective function under equal discount factors, $\rho$, is given by
\begin{equation}
\max_{C, B, K, R} \; \Big  [ u_1(C_1) +  u_2(C_2) - (\gamma_{1}+\gamma_2) \Phi(\rho) G(t,K,B,R_b) \Big ].
\end{equation}

In the case of heterogeneous discount factors, $\rho_1$, $\rho_2$, this objective becomes
\begin{equation}
\max_{C, B, K, R} \; \Big  \{ e^{-\rho_{1}t}\left[u_1(C_1) -\gamma_{1}\Phi(\rho_{1})G(t,K,B,R_b) \right] + e^{-\rho_{2}t}\left[u_2(C_2) -\gamma_{2}\Phi(\rho_{2})G(t,K,B,R_b) \right]\Big \}.
\end{equation}
Note that, apart from the difference in the coefficients, $\gamma_{i}\Phi(\rho_{i})$, the expression  assigns a different (time-dependent) weight to country 1 versus country 2, as the exponential decay at  different rates. As a result, the controls in this problem will be time-dependent. Figure \ref{fig:differentRho} shows the optimal paths for both the GP and the RP problems when $\rho_1$ is our benchmark while $\rho_2= 1.2 \; \rho_1$ (thus, the global south is ``more impatient").

\begin{figure}[!htpb]
\centering
\begin{subfigure}{0.35\textwidth}  
\includegraphics[width=\textwidth]{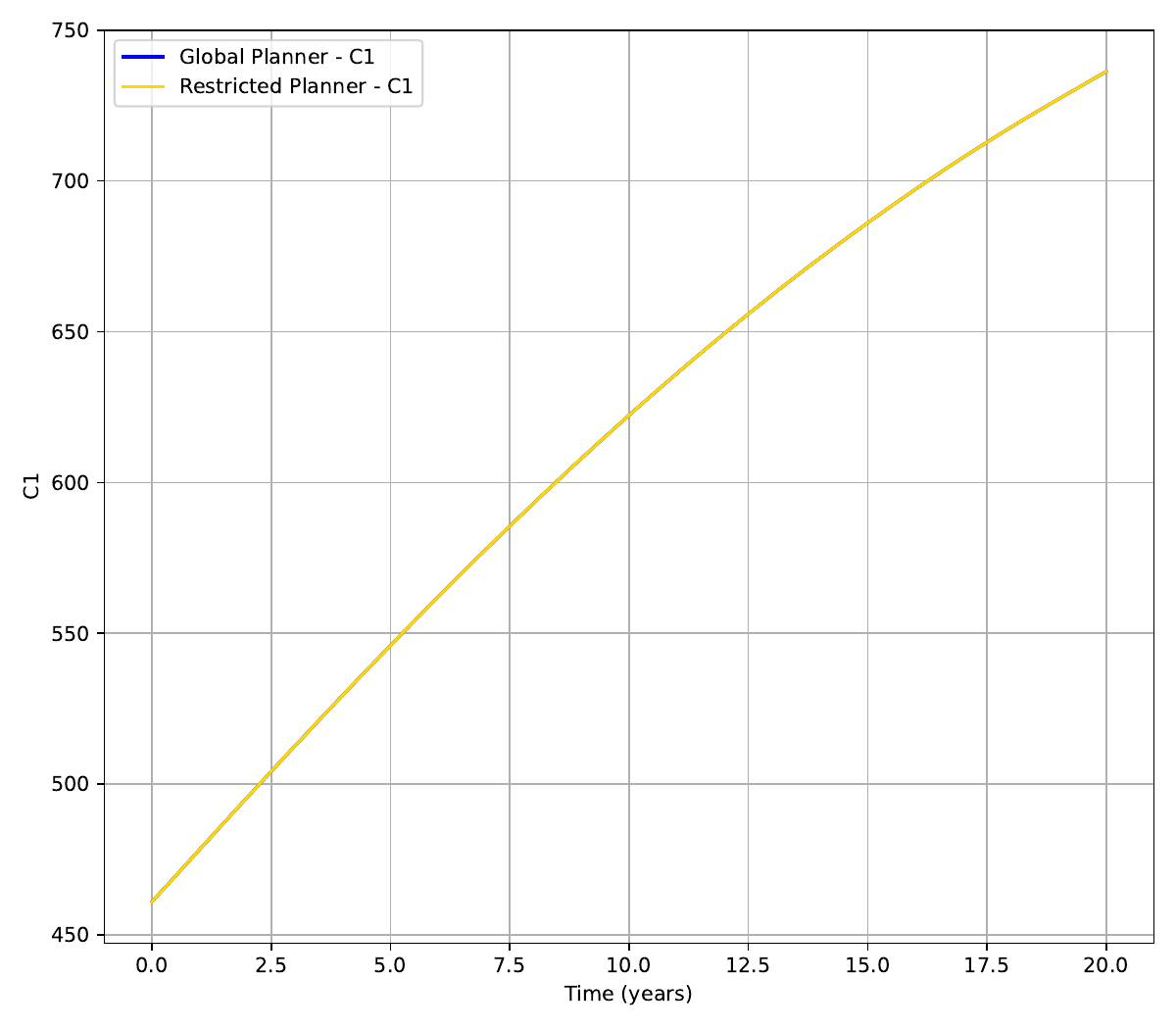}
\caption{$C_1$}
\label{fig:differentRhoC1}
\end{subfigure}
\begin{subfigure}{0.35\textwidth}  
\includegraphics[width=\textwidth]{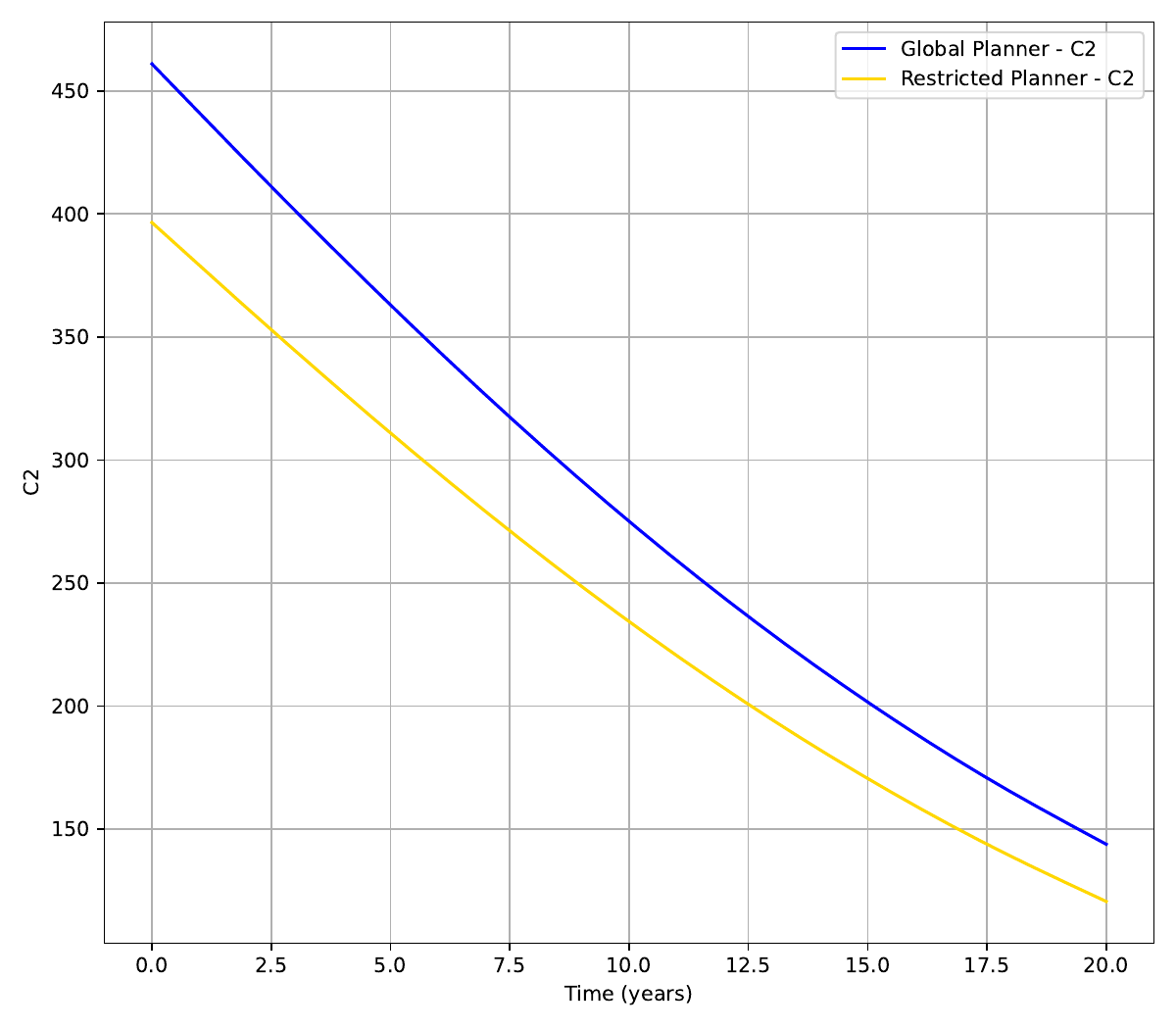}
\caption{$C_2$}
\label{fig:differentRhoC2}
\end{subfigure}
\begin{subfigure}{0.35\textwidth}  
\includegraphics[width=\textwidth]{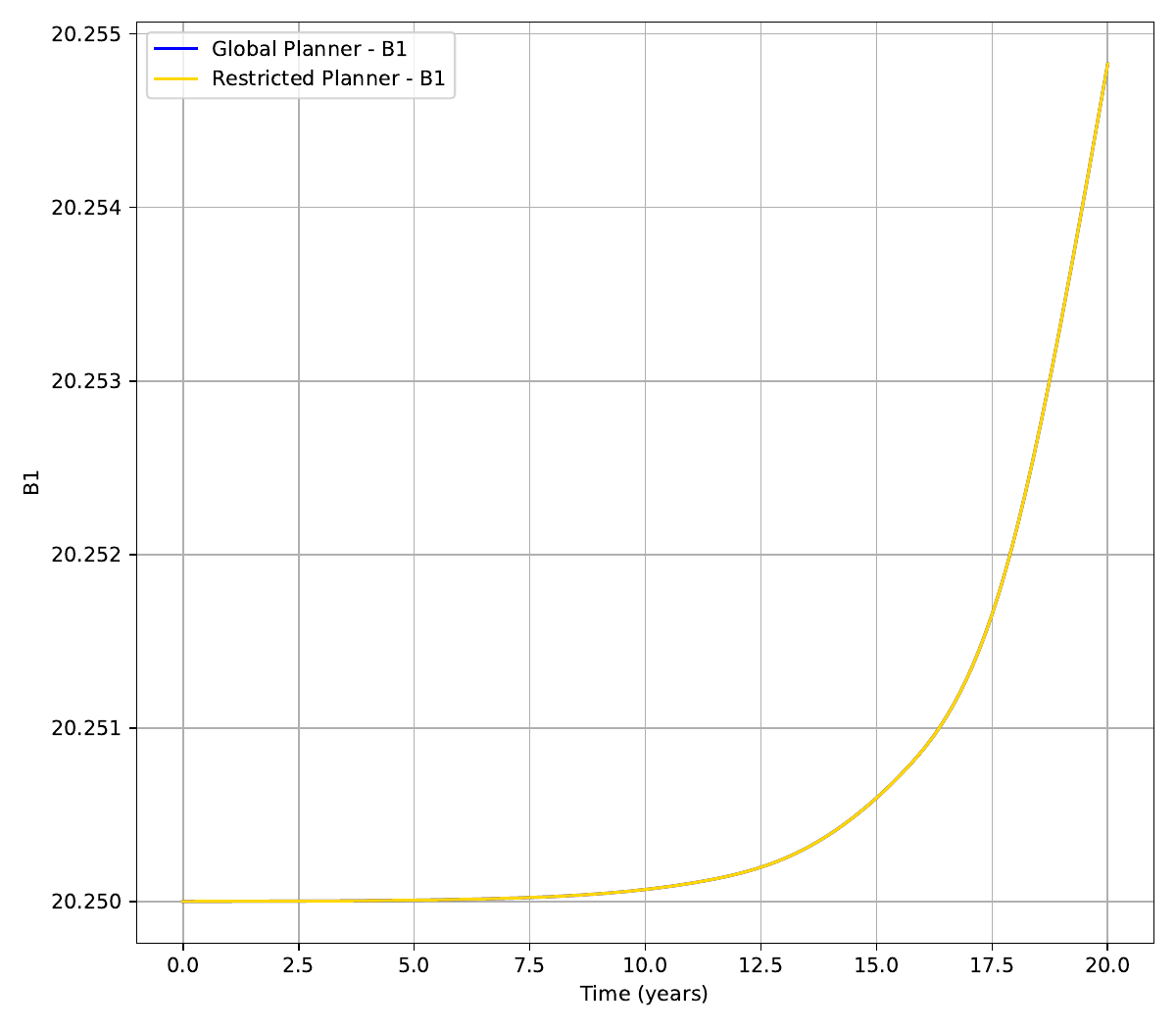}
\caption{$B_1$}
\label{fig:differentRhoB1}
\end{subfigure}
\begin{subfigure}{0.35\textwidth}  
\includegraphics[width=\textwidth]{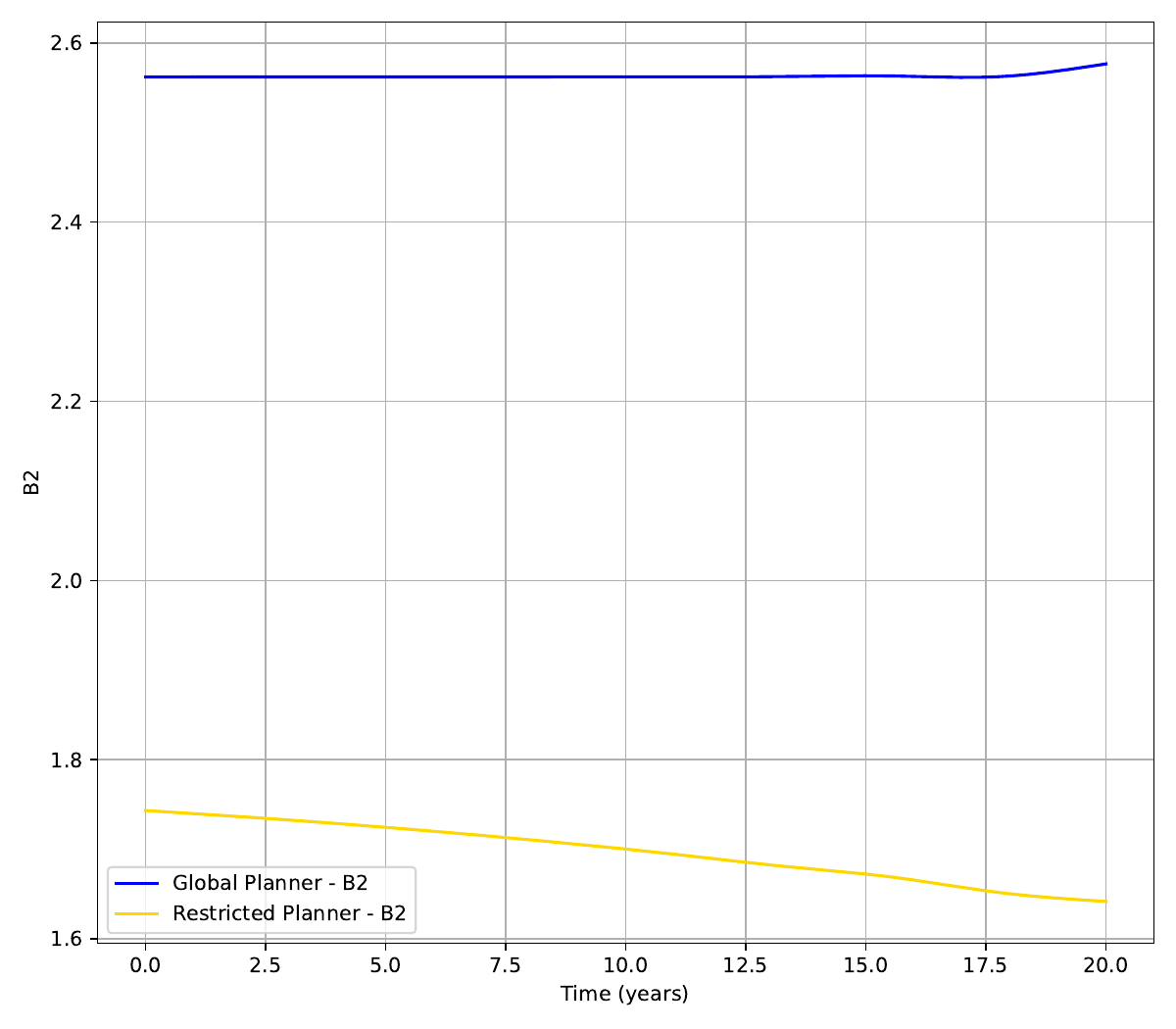}
\caption{$B_1$}
\label{fig:differentRhoB2}
\end{subfigure}
\begin{subfigure}{0.35\textwidth}  
\includegraphics[width=\textwidth]{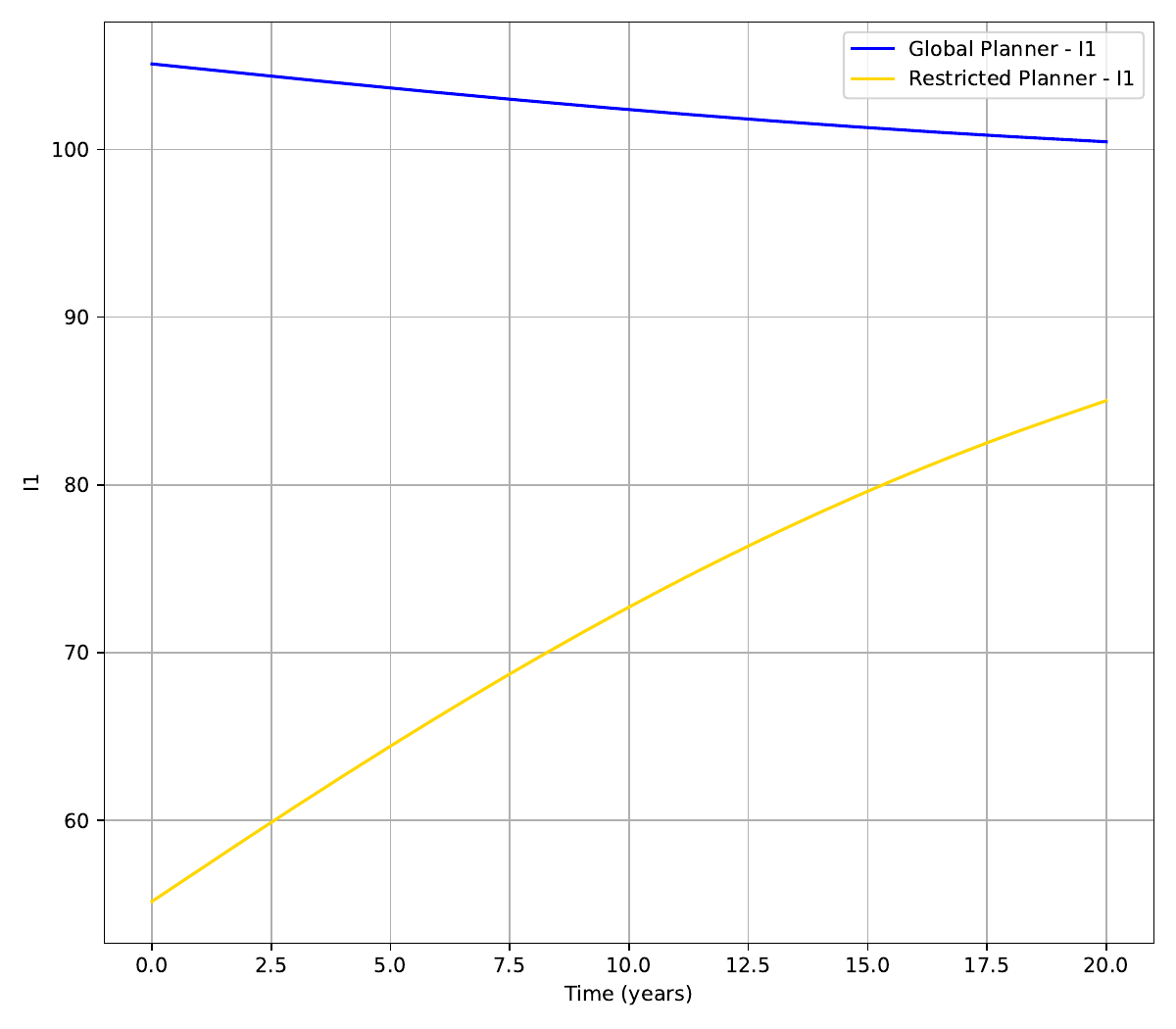}
\caption{$I_1$}
\label{fig:differentRhoI1}
\end{subfigure}
\begin{subfigure}{0.35\textwidth}  
\includegraphics[width=\textwidth]{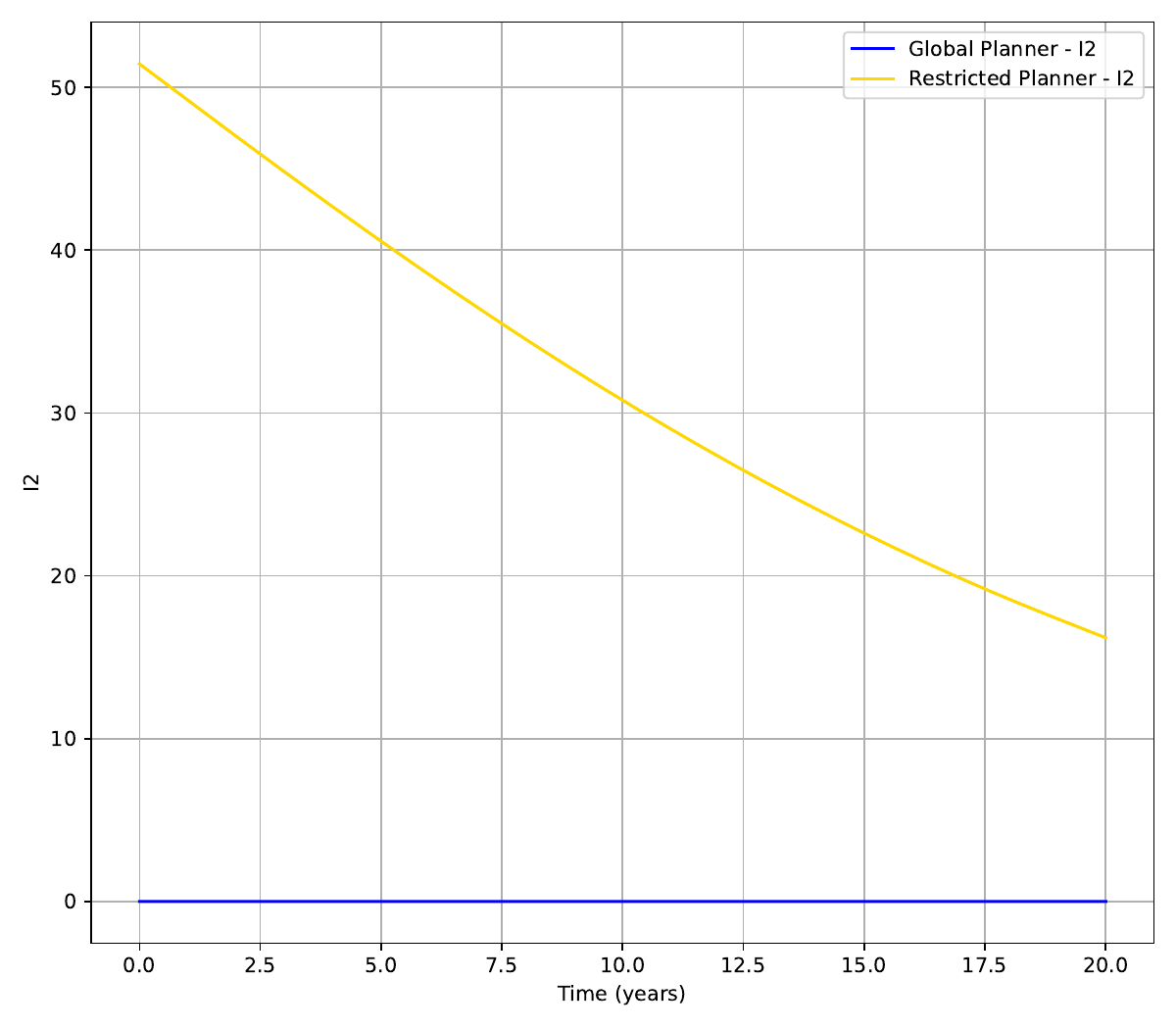}
\caption{$I_2$}
\label{fig:differentRhoI2}
\end{subfigure}
\begin{subfigure}{0.35\textwidth}  
\includegraphics[width=\textwidth]{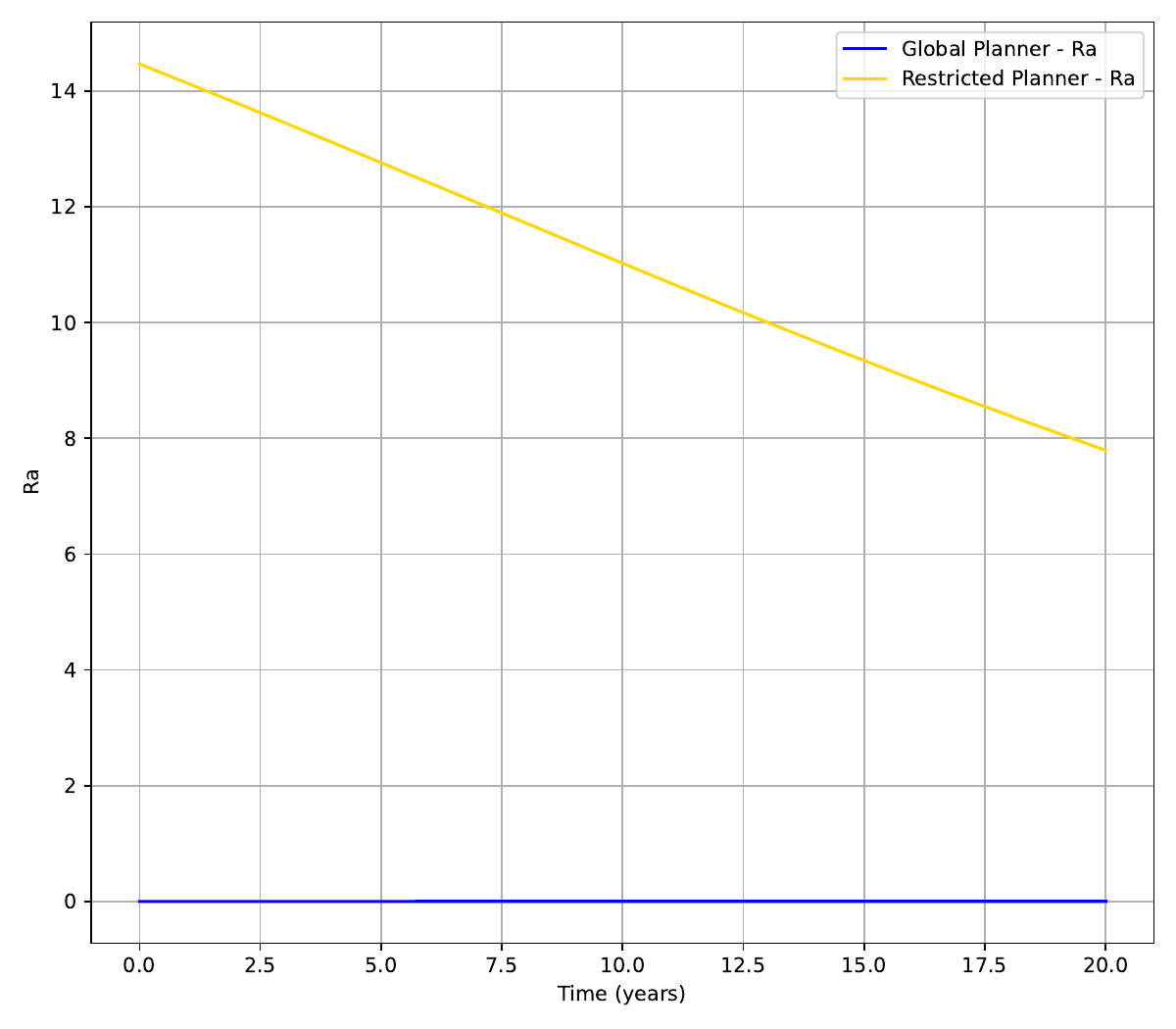}
\caption{$R_a$}
\label{fig:differentRhoRa}
\end{subfigure}
\begin{subfigure}{0.35\textwidth}  
\includegraphics[width=\textwidth]{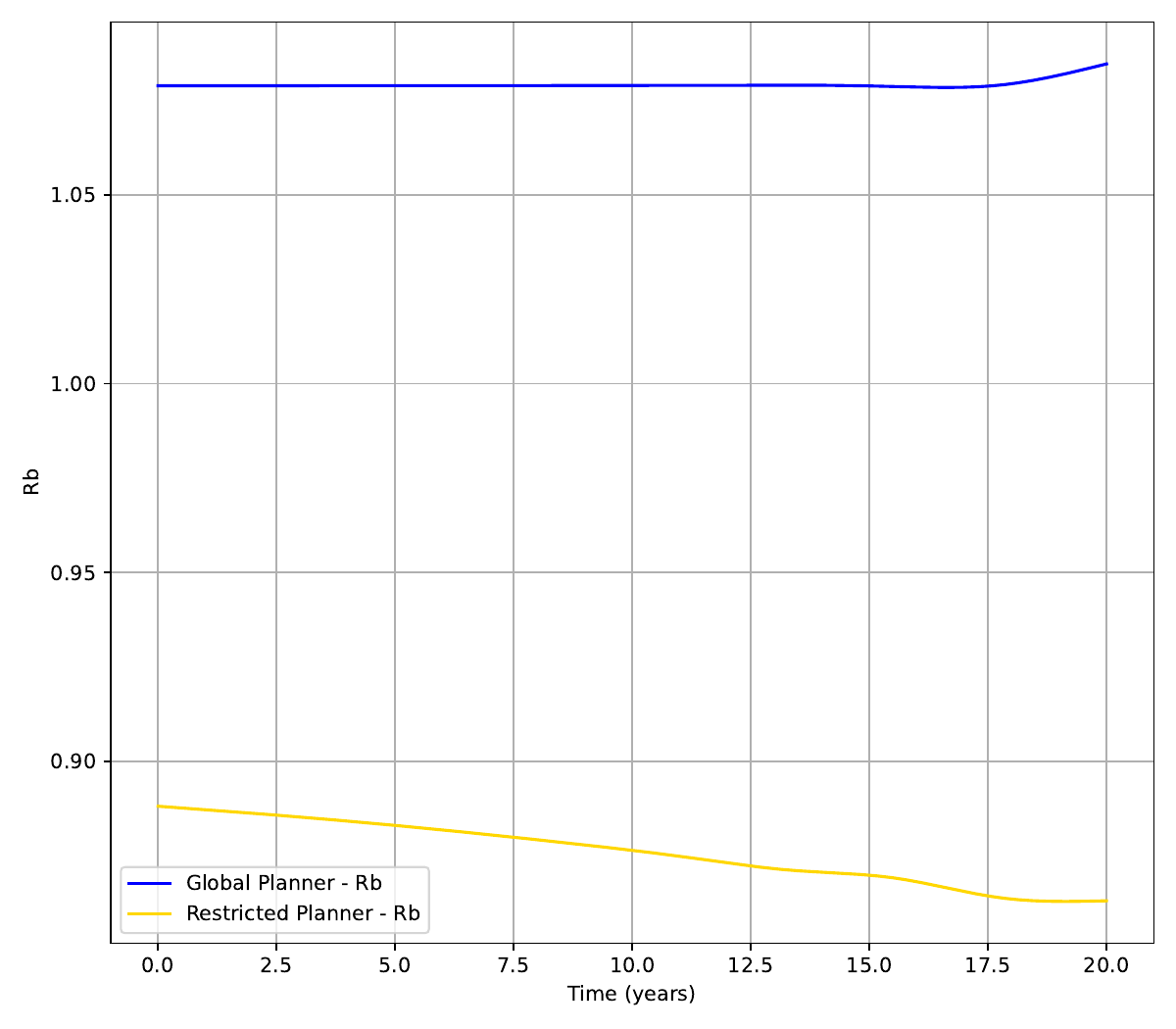}
\caption{$R_b$}
\label{fig:differentRhoRb}
\end{subfigure}
\caption{Control variables $C_1,C_2,B_1,B_2,I_1,I_2,R_a,R_b$ in the Global and Restricted planner cases when $\rho_2 = 1.2 \rho_1$.}
\label{fig:differentRho}
\end{figure}

Optimal consumption in the global north is rising, with the GP and the RP solution largely overlapping. In contrast, optimal consumption in the global south is declining, as the higher preference for present consumption in the south in internalized. Note that in this case, the GP solution implies higher consumption to the south relative to the RP solution. The optimal abatement effort in both countries is increasing over time (and these largely overlap) in both the GP and the RP solution in order to balance for the emissions resulting from increasing consumption. As the GP assigns higher consumption to the south, they also choose higher abatement. Abatement dynamics in the north remains limited, except in the case of the RP, where abatement is more markedly decreasing over time, matching the a decreasing pattern of optimal consumption in the South. 

\subsubsection{Vulnerability to climate change}
Figure \ref{fig:emissionSigma1} below shows the over-time stocks of emissions in the homogeneous damage ($\gamma_1 = \gamma_2 = 0.0125$) case, versus the case where the global south experiences higher damages, $\gamma_1 = 0.0075$, and $\gamma_2 = 0.0175$ (we keep the total damages constant for this comparison).\footnote{Recall that a higher $\gamma$ can also be interpreted as a result of a lesser degree or effectiveness of the respective adaptation measures.}
\begin{figure}[!htpb]
\begin{subfigure}{0.49\textwidth}  
\includegraphics[width=\textwidth]{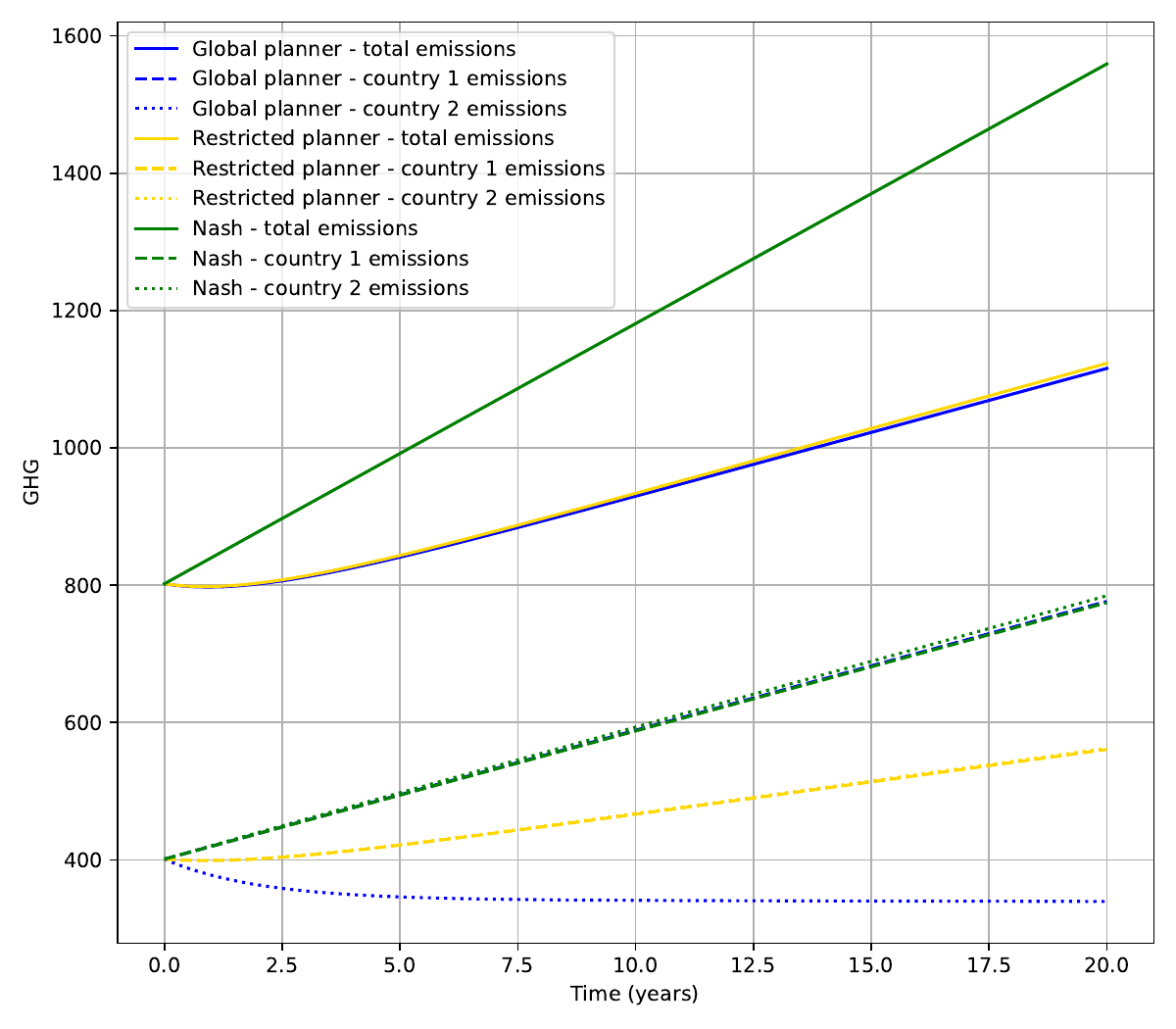}
\caption{Symmetric case}
\label{fig:emissionSigma1Symmetric}
\end{subfigure}
\begin{subfigure}{0.49\textwidth}  
\includegraphics[width=\textwidth]{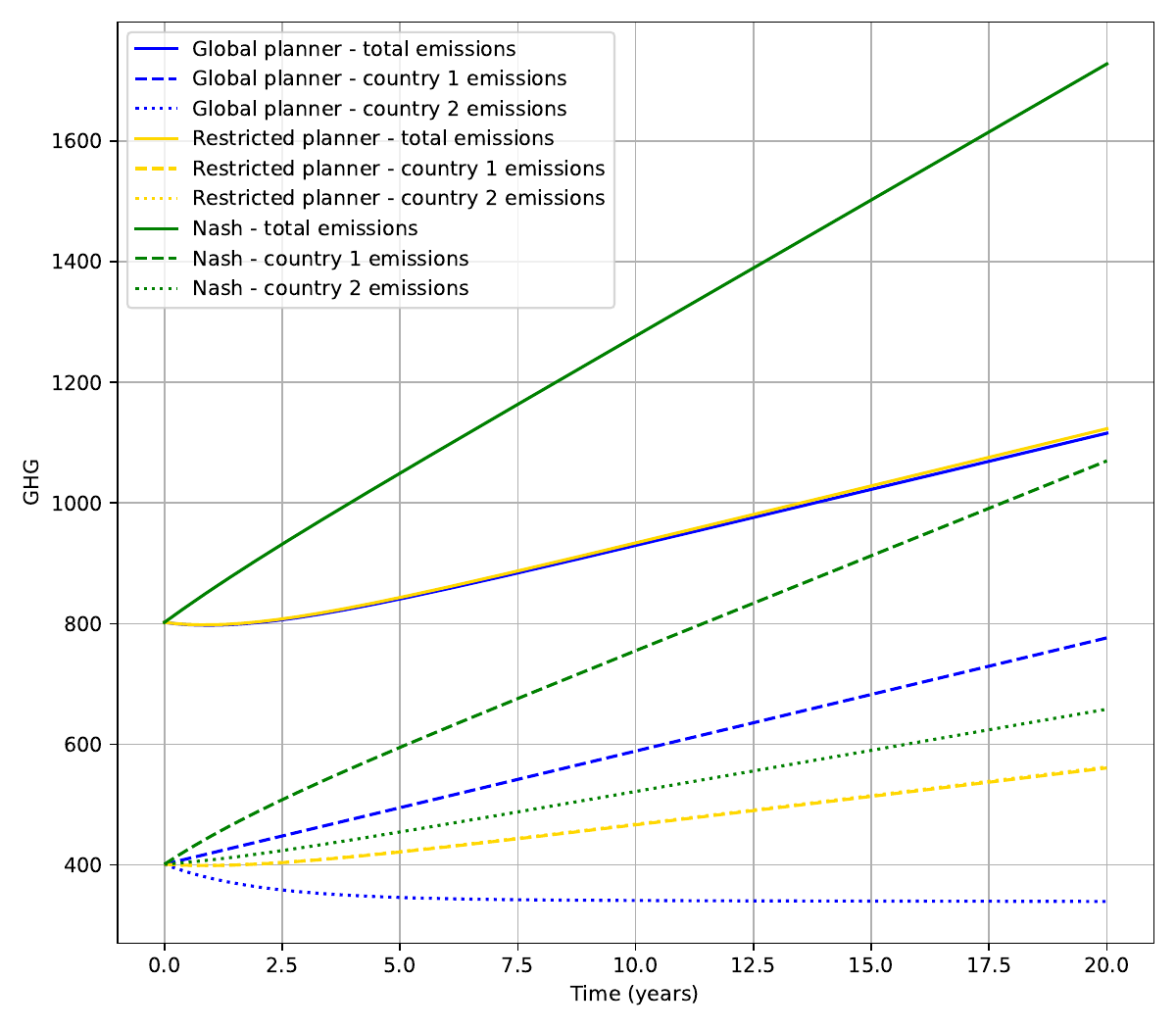}
\caption{Asymmetric case}
\label{fig:emissionSigma1Asymmetric}
\end{subfigure}
\caption{$\sigma = 1$ ($\log$), Total emissions.}
\label{fig:emissionSigma1}
\end{figure}

A few remarks are in order. First, note that, as it is more efficient, the GP would only engage in production in the north. Global emissions are higher in the non-cooperative outcome than in the GP and the RP  solutions, which are close. The associated global temperature paths are given in Figure \ref{fig:temperatureSigma1}. In the Nash outcome under asymmetric $\gamma$'s, the north emits more than the south, reflecting the higher resilience to damages in the north. In contrast, in the symmetric case, the emissions in the south are higher than those in the north. Total emissions in the asymmetric case are higher than in the symmetric one, pointing to the need for improving the south's ability to adapt.

\begin{figure}[!htpb]
\begin{subfigure}{0.49\textwidth}  
\includegraphics[width=\textwidth]{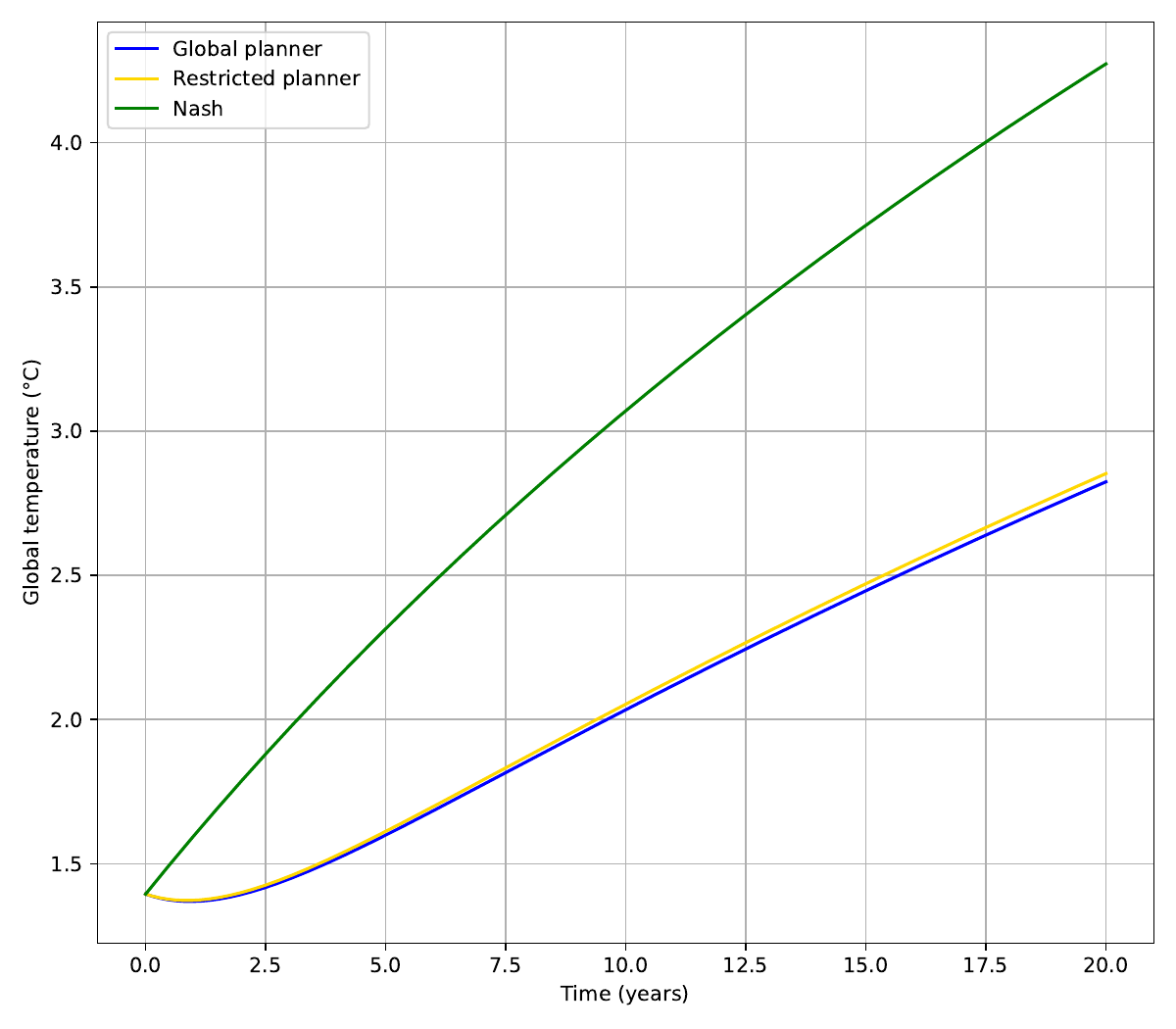}
\caption{Symmetric case}
\label{fig:temperatureSigma1Symmetric}
\end{subfigure}
\begin{subfigure}{0.49\textwidth}  
\includegraphics[width=\textwidth]{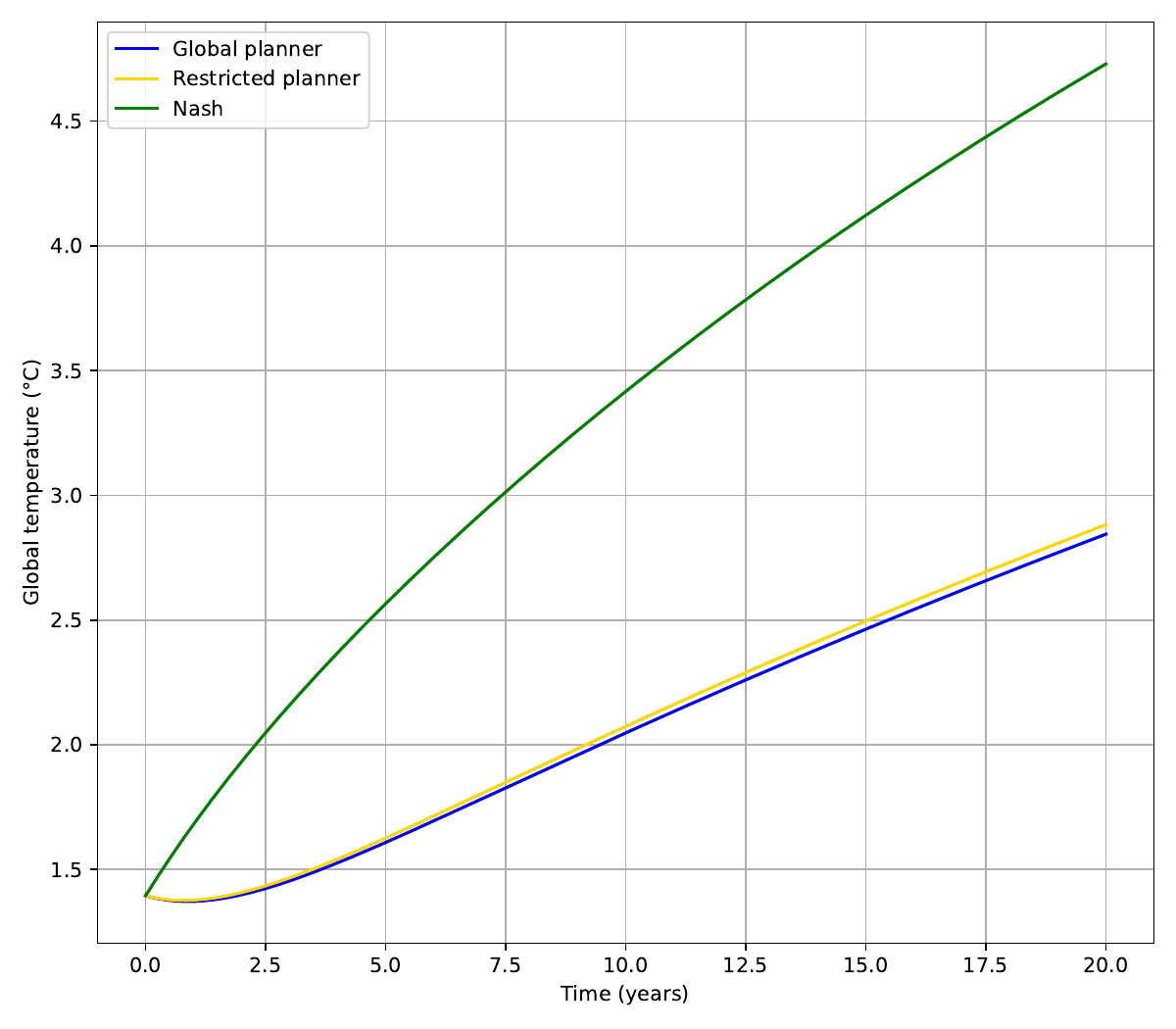}
\caption{Asymmetric case}
\label{fig:temperatureSigma1Asymmetric}
\end{subfigure}
\caption{$\sigma = 1$ ($\log$), Global temperature.}
\label{fig:temperatureSigma1}
\end{figure}

A lower intertemporal substitution parameter ($\sigma=1.2$). In this case, emissions in both countries are lower compared to the logarithmic utility case.  In addition, the relative differences between the efficiency benchmarks and the Nash outcome ar smaller. The qualitative properties concerning players' relative positions in the symmetric versus the asymmetric cases remain largely unchanged from the logarithmic case, although the differences between the north and the south are somewhat more pronounced (see Figure \ref{fig:emissionSigma1.2}).
\begin{figure}[!htpb]
\begin{subfigure}{0.49\textwidth}  
\includegraphics[width=\textwidth]{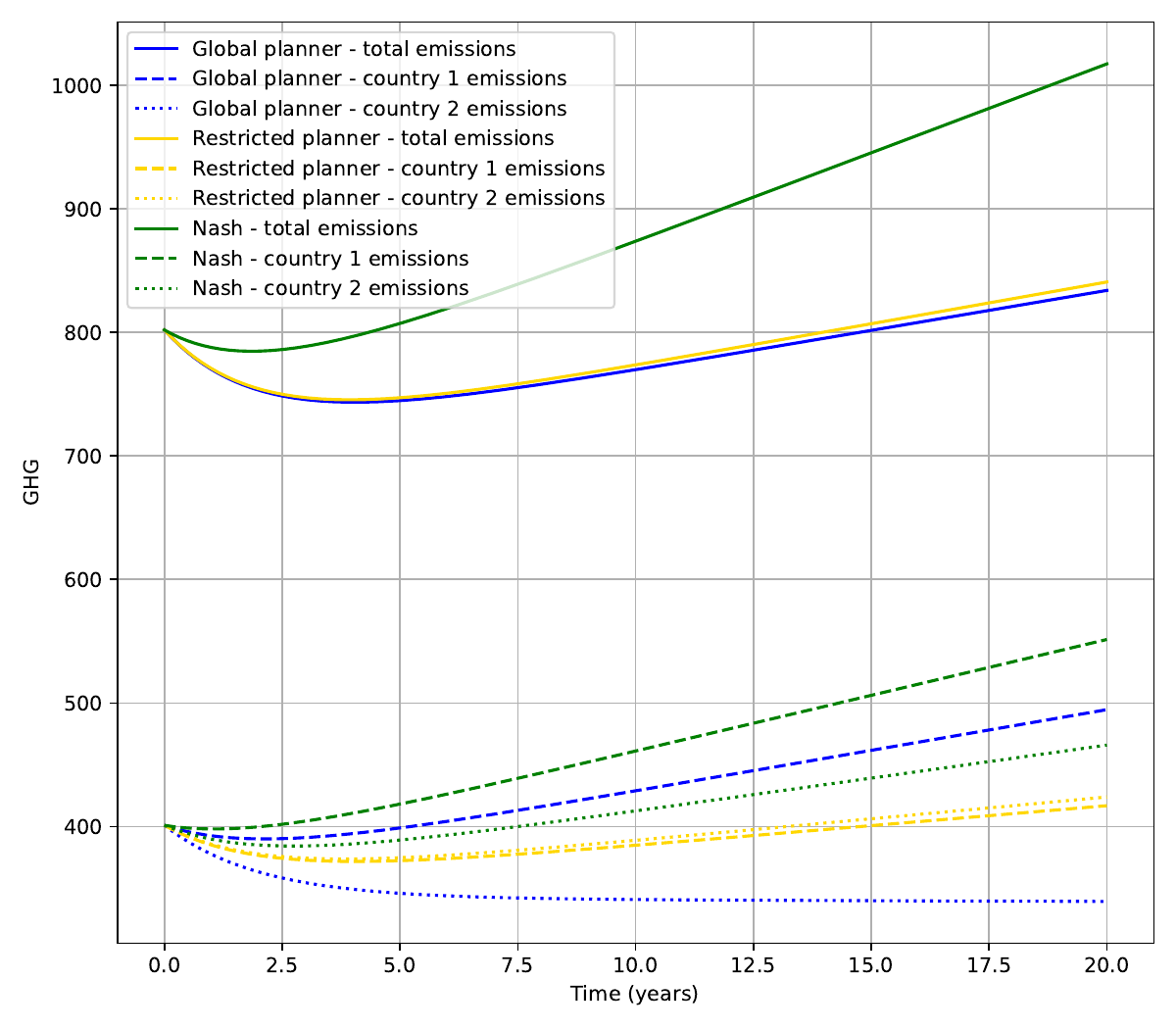}
\caption{Symmetric case}
\label{fig:emissionSigma1.2Symmetric}
\end{subfigure}
\begin{subfigure}{0.49\textwidth}  
\includegraphics[width=\textwidth]{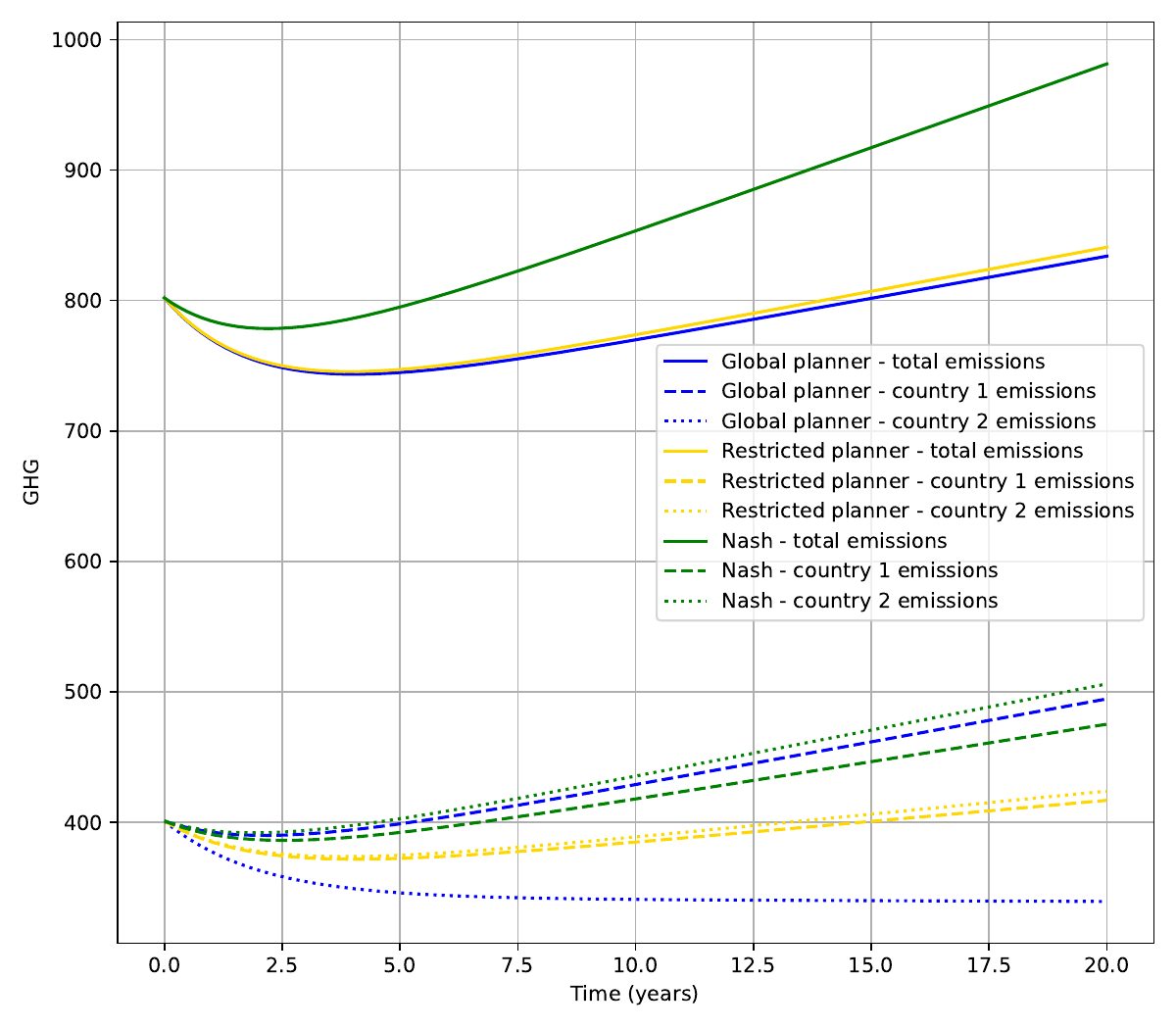}
\caption{Asymmetric case}
\label{fig:emissionSigma1.2Asymmetric}
\end{subfigure}
\caption{$\sigma = 1.2$, Total emissions}
\label{fig:emissionSigma1.2}
\end{figure}

Finally, we turn to the effect of an increase in relative climate vulnerability in the global south (modeled as an increase in the ratio $\frac{\gamma_2}{\gamma_1}$) on the north-south transfers. We choose the RP case for the illustration. Figure \ref{fig:RaRbGamma2RestrictedPlannerLog} suggests that north-south abatement technology transfers become more important relatively to production transfers as climate vulnerability in the south increases. 

In the following section we extend the climate model to account for Knightian uncertainty and demonstrate that the ITM can be applied to that case.

\begin{figure}[!htpb]
\centering
\includegraphics[width=0.7\textwidth]{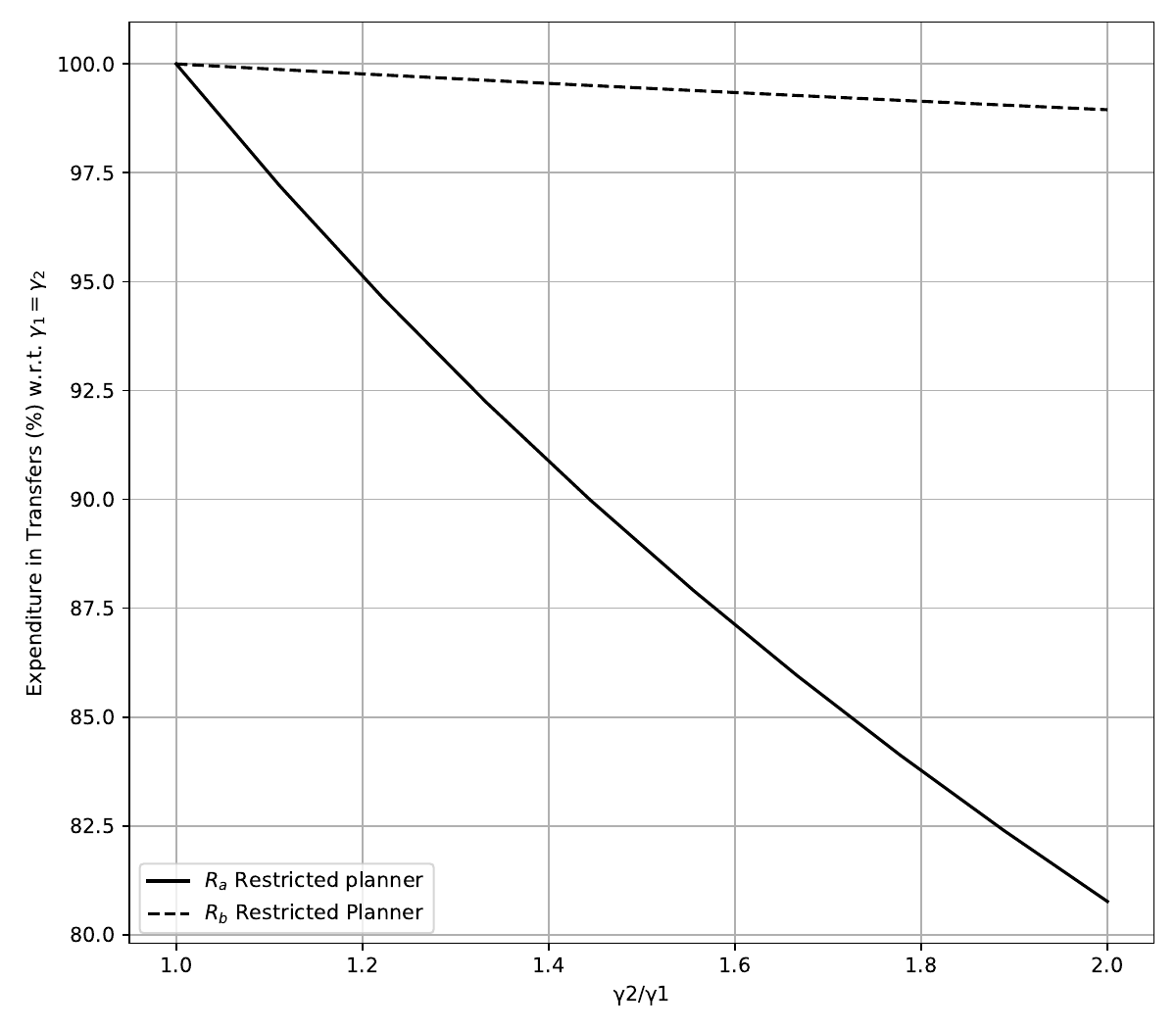}
\caption{Impact of $\gamma_2$ on the expenditure of transfers $R_a,R_b$ in the Restricted planner case ($\sigma = 1$ (log)).}
\label{fig:RaRbGamma2RestrictedPlannerLog}
\end{figure}

\section{Knightian Uncertainty and Robust control}  
\label{sec:robusteness}

While it is well recognized that GHG concentrations in excess of the preindustrial level are detrimental to economic well-being, the
details of the mapping from excess concentrations to damages and to the resulting
utility reductions are subject to considerable uncertainty. Indeed, as anthropogenic climate change is unprecedented, one might argue that it would be difficult to accommodate this uncertainty within the standard model of Bayesian decision-making under risk. As a result, decision-makers might favor policies that perform well for a variety of models in the proximity of the benchmark (approximate) model. Robust control is one rigorous way to model this problem; see, for example, \cite{hansen2008robustness}, \cite{hansen2010wanting}, \cite{anderson2014robust}, \cite{li2016robust}, \cite{hansen2022structured}, and \cite{barnett2022climate}. In this section we introduce model uncertainty regarding the damages resulting from climate change through the variables $\gamma_{1}$, $\gamma _{2}$. Our goal is to demonstrate that the ITM can be applied to this setup. Here, we will only develop explicitly the global planner's (GP) problem (we relegate the proofs in Appendix \ref{supapp:robustglobal}). The restricted planner's case is studied in Subsection \ref{supapp:robustrestricted}), while the non-cooperative Nash outcomes are studied in Subsection \ref{supapp:robustNash}. 


\subsection{The global planner's (GP) robust control problem}
In the GP case, the concern about model-uncertainty is
represented by a two-person zero-sum dynamic game in which, after observing
the choice of a social planner, a \textquotedblleft malevolent
player\textquotedblright\ chooses the worst specification of the model
according to a metric we describe below.\footnote{Our modeling follows \cite{vandenBroek2003}. Our results generalize to the use of the relative entropy approach, which is often employed in economic applications; see, for example, \cite{hansen2008robustness}.} In particular, we assume that the malevolent nature's
deviation from the approximating distribution is penalized by
adding 
\begin{equation}
\int_{0}^{\infty }e^{-\rho t}\left[\alpha_1 \left\vert \gamma _{1}(t)-%
\widehat{\gamma }_{1}\right\vert^2 +\alpha_2\left\vert \gamma _{2}(t)-\widehat{\gamma }%
_{2}\right\vert^2 \right] dt
\end{equation}%
to the planner's objective function, where $\alpha_i$ represent the
magnitude of the deviation \textquotedblleft punishment\textquotedblright for the two countries, making it more costly for the malevolent player to deviate from the approximate model. In what follows, we will concentrate on the benchmark case where $\alpha_1 = \alpha_2=\alpha$. We interpret $\widehat{%
\gamma }_{1}$ and $\widehat{\gamma }_{2}$ as the (strictly positive) benchmark values of the respective parameters. Since a larger $\alpha$ implies a higher penalty for the malevolent player resulting from their deviation from the approximating distribution, it makes such a
deviation less likely, which is equivalent to a lower concern about model uncertainty (robustness). Incorporating the malevolent player's decision, the robust GP problem can be written as follows:

\begin{eqnarray}
\label{eq:functionalrobust-planner}
&&\max_{C(\cdot),K(\cdot),B(\cdot),R_{a}(\cdot),R_{b}(\cdot)} \qquad 
\min_{\gamma _{1}(\cdot),\gamma _{2}(\cdot)}U^{R}((\alpha_1,\alpha_2) ) = U_1^R(\alpha_1)+U_2^R(\alpha_2)=
\notag
\\
&&\int_{0}^{\infty }e^{-\rho t}\left[ u_{1}\left( C_{1}(t)\right)
-\gamma _{1}(t)\left( S(t)-\overline{S}\right) +\alpha_1 \left\vert \gamma
_{1}(t)-\widehat{\gamma }_{1}\right\vert^2 \right] dt  \notag 
\\
&&+\int_{0}^{\infty }
e^{-\rho t}\left[ u_{2}\left( C_{2}(t)\right)
-\gamma _{2}(t)\left( S(t)-\overline{S}\right) +\alpha_2 \left\vert \gamma
_{2}(t)-\widehat{\gamma }_{2}\right\vert^2 \right] dt
\end{eqnarray}
subject to a single resource constraint:
\begin{multline}
\label{eq:plannerresourceconstraint}
C_1(t)+ C_2(t)+ B_1(t)+ B_2(t)+ R_a(t)+ R_b(t)+ K_1(t)+ K_2(t) \\
= Y_1(t) + Y_2(t) =  \overline{A}(t)f_1(K_1(t)) + h(R_a(t)) \overline{A}(t)f_2(K_2(t)).
\end{multline}
As before, the control strategies of the GP are assumed to belong to the set $\mathcal{U}^{GP}$ given in
\eqref{eq:admstratplanner1}, while those of the malevolent player
$(\gamma _{1}(\cdot),\gamma _{2}(\cdot))$
belong to $L^2_\rho(\R_+;\R^2)$. Following Theorem Theorem \ref{th:staticreduction}, the above functional can be rewritten, allowing us to obtain the following minimax theorem.

\begin{Theorem}
\label{th:maxmingeneral}
Suppose that Hypothesis \ref{hp:easycase} holds.\footnote{The statement holds more generally under the joint concavity of $G$.} Then:
$$
\max_{C(\cdot),K(\cdot),B(\cdot),R_{a}(\cdot),R_{b}(\cdot)} \qquad 
\min_{\gamma _{1}(\cdot),\gamma _{2}(\cdot)}U^{R}((\alpha_1,\alpha_2) ) 
=
$$
$$
\min_{\gamma _{1}(\cdot),\gamma _{2}(\cdot)} \qquad
\max_{C(\cdot),K(\cdot),B(\cdot),R_{a}(\cdot),R_{b}(\cdot)} 
U^{R}((\alpha_1,\alpha_2) ) 
$$
\end{Theorem}
\begin{proof}
See Appendix \ref{app:robusteness}.
\end{proof}

The above clearly holds if we restrict attention to constant strategies $(\gamma_1,\gamma_2)$ for the malevolent player. For expository purposes, we will restrict attention to this case under logarithmic utility in what follows.
\subsubsection{Logarithmic payoffs}
\label{subsub:special-robust-GP}

We now investigate the logarithmic case under constant $\gamma_1$ and $\gamma_2$. Using Theorem \ref{th:maxmingeneral}, we solve the following: 
\begin{multline}
\label{eq:problem-robust-planner1}
\min_{\gamma _{1},\gamma _{2} \in \mathbb{R}} \; \max_{C(\cdot),B(\cdot), K(\cdot), R(\cdot) \in \mathcal{U}^{p1}} U^{R}(\alpha ) = \min_{\gamma _{1},\gamma _{2} \in \mathbb{R}} \; \max_{C(\cdot),B(\cdot), K(\cdot), R(\cdot) \in \mathcal{U}^{p1}} U_{1}+U_{2}=\\
\min_{\gamma _{1},\gamma _{2} \in \mathbb{R}} \; \max_{C(\cdot),B(\cdot), K(\cdot), R(\cdot) \in \mathcal{U}^{p1}} \int_{0}^{\infty }e^{-\rho t}\Big[ \ln\left( C_{1}(t)\right) + \ln\left( C_{2}(t) \right )
-(\gamma_{1}+\gamma_{2})\left( S(t)-\overline{S}\right) \\ +\alpha \left\vert \gamma
_{1}-\widehat{\gamma }_{1}\right\vert^2 +\alpha \left\vert \gamma
_{2}-\widehat{\gamma }_{2}\right\vert^2 \Big] dt \end{multline}

Both the planner's choices reflect the fact that they internalize the damages resulting from the climate externality and their decisions only depend on the value of $\gamma_1+\gamma_2$. For the planner's payoff to be decreasing in $S(t)$ we will require that $(\gamma_1+\gamma_2)$ is positive. We then have the following.


\begin{Proposition}
\label{pr:robust-global-planner}
Suppose Assumption \ref{hp:easycase} holds and assume logarithmic payoffs. Given $\hat{\gamma}_1, \hat{\gamma}_2 >0$, the values of $(\gamma_1, \gamma_2)$ that solve (\ref{eq:problem-robust-planner1}) are unique and given by:
\[
\gamma_1 = \frac{1}{4\alpha} \left [ - \alpha(- 3\widehat{\gamma}_1 + \widehat{\gamma}_2) - \rho \Gamma_1 + \sqrt{\left [\rho \Gamma_1 -\alpha(\widehat{\gamma}_1 + \widehat{\gamma}_2) \right ]^2 + 8 \alpha} \right ],
\]
\[
\gamma_2 = \frac{1}{4\alpha} \left [ - \alpha(- 3\widehat{\gamma}_2 + \widehat{\gamma}_1) - \rho \Gamma_1 + \sqrt{\left [\rho \Gamma_1 -\alpha(\widehat{\gamma}_1 + \widehat{\gamma}_2) \right ]^2 + 8 \alpha} \right ],
\]
where
\[
\Gamma_1 := \left[
\frac{\overline{S}}{\rho} -
\frac{P(0)}{\rho}-\frac{T(0)}{\rho+\phi}
\right] - \frac{\Phi \eta_K}{\rho} (R_b + B_1 + B_2) + \frac{\Phi \eta_B}{\rho}  \Big (B_1^{\theta_1} + g(R_b) B_2^{\theta_2} \Big ).
\]
The values of $R, C, B$ and $K$ are given by the corresponding resource constraint in Proposition \ref{pr:optimum-planner-easycase-log}. 
\end{Proposition}
\begin{proof}
See  Appendix \ref{app:robusteness}.
\end{proof}



As mentioned earlier, the reader is refereed to Appendix \ref{app:robusteness} for the corresponding results for the RP and Nash cases. The last case is of independent interest, as it extends robust control from decision-theoretic to a differential game setup. As an illustration, Figure \ref{fig:TemperatureRobustPoint} extends the numerical example studied earlier to the case where there is concern about robustness, assuming logarithmic utility. The three figures below summarize the paths for global temperatures corresponding to the solutions of: (A) the global planner, (B) the restricted planner, and (C) the Nash  case, respectively. Each graph is given for different values of $\alpha$, with $\alpha = \infty$ corresponding to the approximate (rational expectations) model. 

Clearly, lower values of $\alpha$ give rise to ``more cautious" behavior, as they correspond to a higher concern about model uncertainty. In Appendix \ref{app:random} we discuss related simulations and corresponding confidence intervals when nature draws independently over time from an exponential distribution with mean $\gamma$. The graphs indicate a sharp difference between the planners' solutions (graphs A,B) and those resulting from the non-cooperative equilibria (graphs C). 

\begin{figure}[!htpb]
\begin{subfigure}{0.49\textwidth}  
\includegraphics[width=\textwidth]{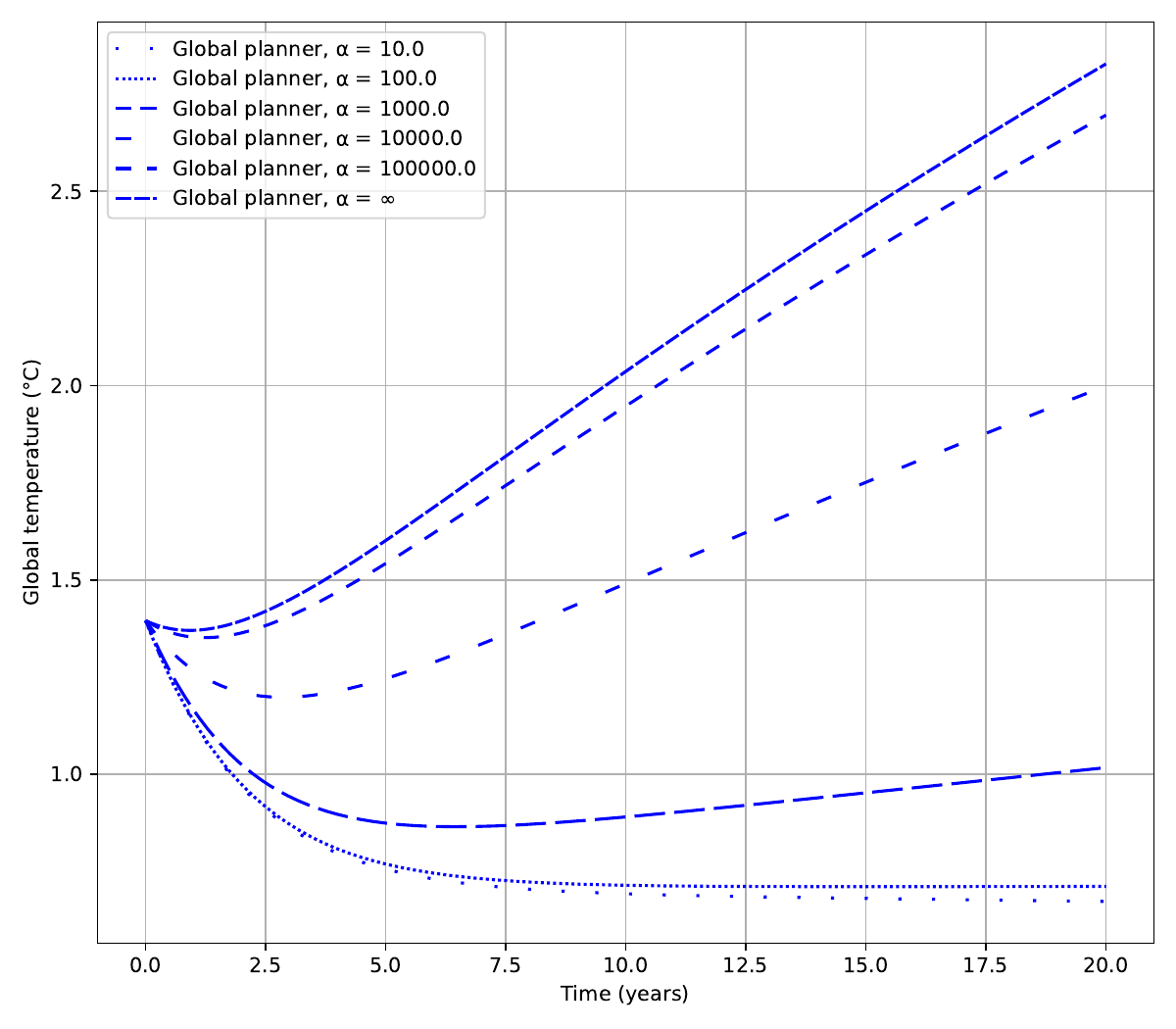}
\caption{Global Planner}
\label{fig:TemperatureRobustPointPlanner}
\end{subfigure}
\begin{subfigure}{0.49\textwidth}  
\includegraphics[width=\textwidth]{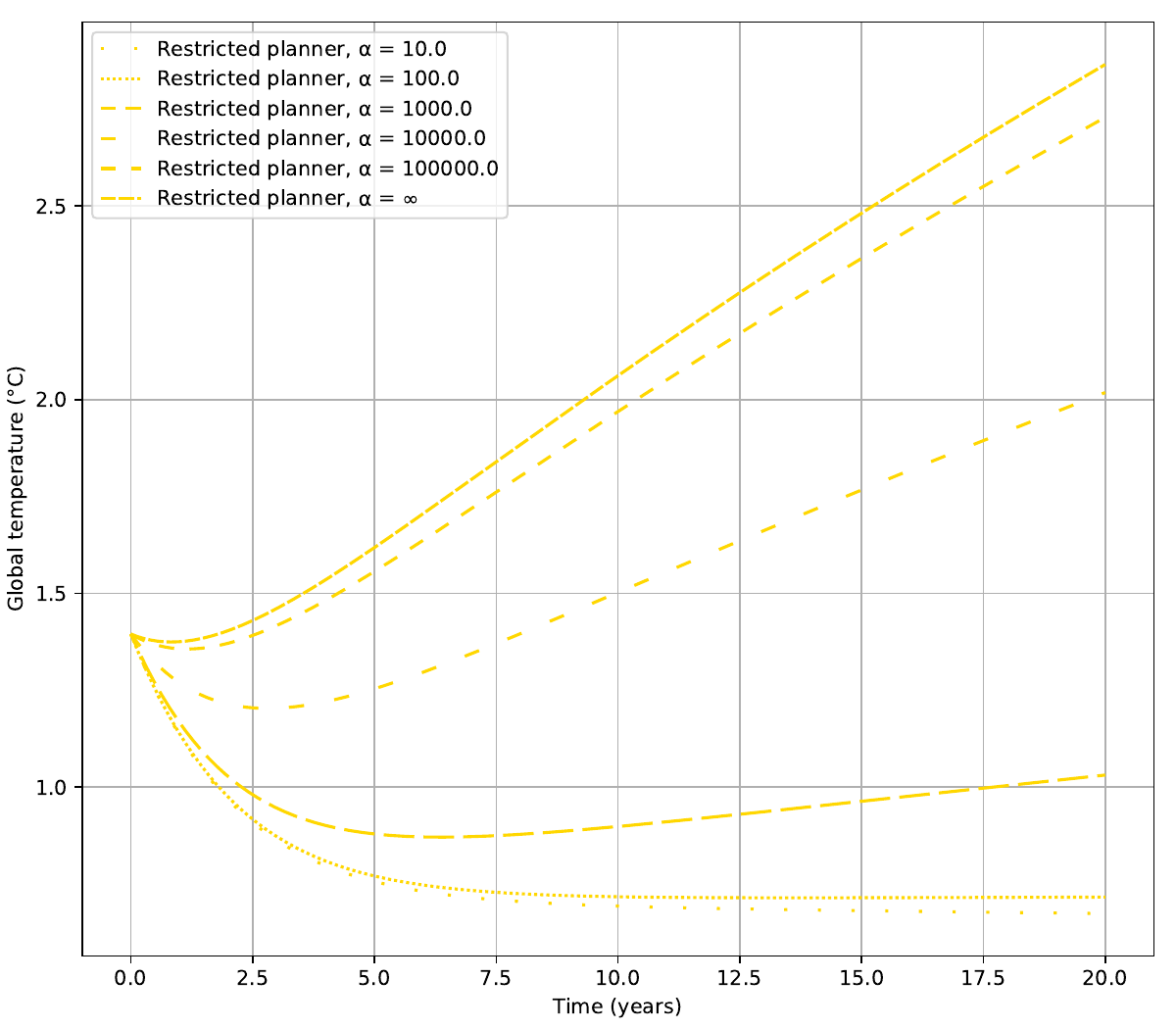}
\caption{Restricted Planner}
\label{fig:TemperatureRobustPointPlannerNoResources}
\end{subfigure}
\begin{subfigure}{0.49\textwidth}  
\includegraphics[width=\textwidth]{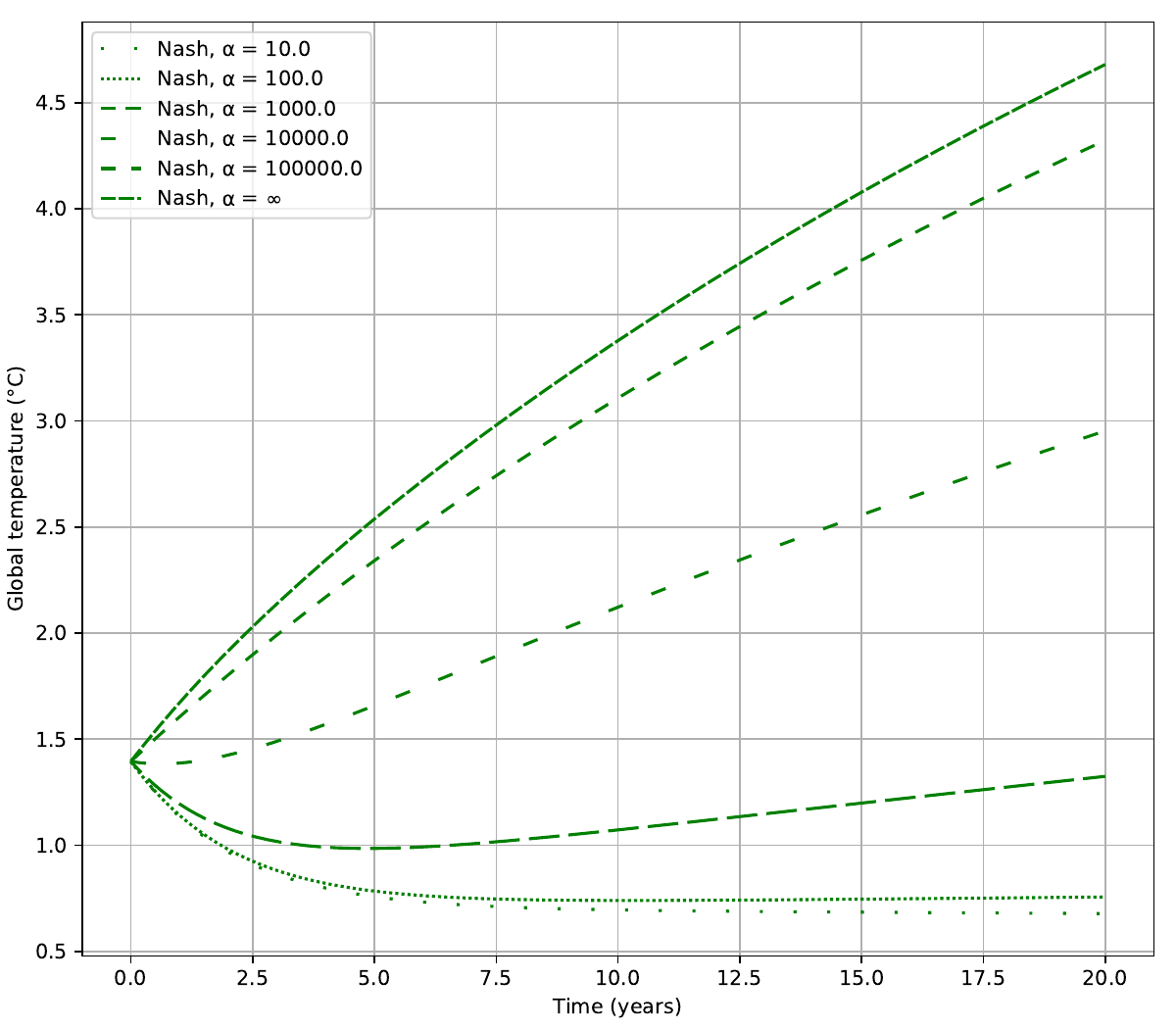}
\caption{Nash}
\label{fig:TemperatureRobustPointNash}
\end{subfigure}
\caption{$\log$, asymmetric damages, global temperature as function of $\alpha$.}
\label{fig:TemperatureRobustPoint}
\end{figure}

\section{Conclusions} Differential games are subject to restrictive linearity assumptions which are often considered necessary for tractability, but which can also have far-reaching economic implications. We proposed an Integral Transformation
Method (ITM) for the study of suitable optimal control and differential game models. When applicable, the ITM allows for a solution to such dynamic problems to be found through the
analysis of a family of temporary optimization problems parametrized by time, thus allowing for the rewriting of the objective functional in a way that permits a pointwise optimization of the
integrand. The method is quite flexible, and it can be applied in several economic applications where the state equation and 
the objective functional are linear in a (suitably defined) state variable. This includes, for example, cases involving time-dependent systems, control and state constraints, differential games, and robust control optimization.

We illustrated the ITM in the context of climate economics, by developing a two-country extension of the integrated assessment model in \cite{Golosov2014}. We characterized emissions, consumption, transfers, and welfare by computing the Nash equilibria of the associated dynamic game and comparing them to efficient benchmarks. Furthermore, we illustrated how the ITM can be applied in a robust control setup in order to investigate how (deep) uncertainty affects climate outcomes. The ITM might be particularly well-suited for future applications in the study of mean-field games; see \cite{lasry2007mean} and, for a more recent application, \cite{alvarez2023price}, including mean-field games involving impulse control, as in \cite{bertucci2020fokker}.

\begin{small}

\bigskip

\bigskip

\bibliographystyle{apalike} 
\bibliography{biblioclimateNEW}     

\bigskip

\bigskip

\appendix

\section{The ITM method: proofs}
\label{app:reduction}


\begin{proof}[Proof of Theorem \ref{th:staticreductionNash}]
\begin{itemize}
\item[(i)]
Let $i\in\{1,...N\}$. 
By Proposition
\ref{pr:rewriting-generalNash}, we have that, for any
$\mathbf{u}_i(\cdot) \in
\mathcal{U}_i(\mathbf{x}_0; \widehat{\mathbf{u}}_{-i}(\cdot))$,
\begin{align*}
&\mathcal{J}_{i}(\mathbf{x}_0,\widehat{\mathbf{u}}_{-i}(\cdot);
\mathbf{u}_i(\cdot))\nonumber\\&=
\left\langle
\mathbf{b}_i(0),\mathbf{x}_0
\right\rangle
+
\int_{0}^{\infty }e^{-\rho_i t}\left[
\left\langle
\mathbf{b}_i(t),
\mathbf{f}(t,\widehat{\mathbf{u}}_{-i}(t),\mathbf{u}_{i}(t))
\right\rangle +  h_i(t,\widehat{\mathbf{u}}_{-i}(t),\mathbf{u}_i(t))\right]dt.
\end{align*}
On the other hand, when ${\mathbf{u}}_{-i,t}=\widehat{\mathbf{u}}_{-i}(t)$, for a.e. $t\ge 0$, the argmax in \eqref{eq:staticNash} is non-empty and contains $\widehat{\mathbf{u}}_{i}(t)$.
Moreover, the map $t \mapsto \widehat{\mathbf{u}}_i(t)$
belongs to 
$\mathcal{U}_{i}(\mathbf{x}_0, \widehat{\mathbf{u}}_{-i}(\cdot))$.
Hence, 
\begin{align*}
&
\left\langle
\mathbf{b}_i(0),\mathbf{x}_0
\right\rangle
+
\int_{0}^{\infty }e^{-\rho_i t}\left[
\left\langle
\mathbf{b}_i(t),
\mathbf{f}(t,\widehat{\mathbf{u}}_{-i}(t),\mathbf{u}_i(t))
\right\rangle +  h_i(t,\widehat{\mathbf{u}}_{-i}(t),\mathbf{u}_i(t))\right]dt\nonumber\\&
\le
\left\langle
\mathbf{b}_i(0),\mathbf{x}_0
\right\rangle
+
\int_{0}^{\infty }e^{-\rho_i t}
\sup_{\mathbf{u_i}\in \mathbf{U}_{i,{\widehat{\mathbf{u}}}_{-i,t}}}
\left[
\left\langle
\mathbf{b}_i(t),
\mathbf{f}(t,{\widehat{\mathbf{u}}}_{-i,t},\mathbf{u}_i)
\right\rangle +  h_i(t,{\widehat{\mathbf{u}}}_{-i,t},\mathbf{u}_i)\right]dt\nonumber\\
&=
\left\langle
\mathbf{b}_i(0),\mathbf{x}_0
\right\rangle
+
\int_{0}^{\infty }e^{-\rho_i t}\left[
\left\langle
\mathbf{b}_i(t),
\mathbf{f}(t,\widehat{\mathbf{u}}_{-i}(t),\widehat{\mathbf{u}}_i(t))
\right\rangle +  h_i(t,\widehat{\mathbf{u}}_{-i}(t),\widehat{\mathbf{u}}_i(t))\right]dt\nonumber\\
&
=
\mathcal{J}_{i}(\mathbf{x}_0,\widehat{\mathbf{u}}_{-i}(\cdot);
\widehat{\mathbf{u}}_i(\cdot)).\nonumber
\end{align*}
Since this is true for all $i=1,\dots,N$, it follows that $\widehat{\mathbf{u}}(\cdot)$
is a Nash equilibrium for the dynamic game.

The conclusion that, in the absence of state constraints,  $\widehat{\bu}(\cdot)$ is a Nash equilibrium for every initial condition is due to the fact that $\mathcal{U}_{G}$ does not depend on $\mathbf{x}_{0}$ and in the argument above $\mathbf{x}_{0}$ only plays a role in the the definition of the set of admissible strategies.
\smallskip

\item[(ii)] In the absence of state constraints, 
assume that $\overline{\mathbf{u}}(\cdot)\in\mathcal{U}_{G}$
is an open-loop Nash equilibrium for the dynamic game.
By Proposition \ref{pr:rewriting-generalNash}, for all $\mathbf{u}^i(\cdot)\in \mathcal{U}_i(\overline{\mathbf{u}}_{-i}(\cdot))$ we have
\begin{multline}\label{asl}
\mathcal{J}_{i}(\mathbf{x}_0;
\overline{\mathbf{u}}_{-i}(\cdot);
\mathbf{u}_i(\cdot))\\
=
\left\langle
\mathbf{b}_i(0),\mathbf{x}_0
\right\rangle
+
\int_{0}^{\infty }e^{-\rho_i t}
\left[
\left\langle
\mathbf{b}_i(t),\mathbf{f}
(t,\overline{\mathbf{u}}_{-i}(t),\mathbf{u}_i(t))
\right\rangle +
h_i(t,\overline{\mathbf{u}}_{-i}(t),\mathbf{u}_i(t))
\right]dt.
\end{multline}
Our assumptions imply \cite[Prop.\,7.33, p.\,153]{BertsekasShreveBOOK} the existence of a Borel measurable map $\widehat{\mathbf{u}}_i(\cdot)\in \mathcal{U}(\overline{\mathbf{u}}_{-i}(\cdot))$ such that 
\begin{multline}\label{als2}
\sup_{\mathbf{u}_{i}\in  \mathbf{U}_{i,\overline{\mathbf{u}}_{-i}(t)}}\Big\{\left\langle
\mathbf{b}_i(t),\mathbf{f}
(t,\overline{\mathbf{u}}_{-i}(t),\mathbf{u}_i)
\right\rangle +
h_i(t,\overline{\mathbf{u}}_{-i}(t),\mathbf{u}_i)\Big\} \\
= \left\langle
\mathbf{b}_i(t),\mathbf{f}
(t,\overline{\mathbf{u}}_{-i}(t),\widehat{\mathbf{u}}_i(t))
\right\rangle +
h_i(t,\overline{\mathbf{u}}_{-i}(t),\widehat{\mathbf{u}}_i(t)), \ \ \ \forall t\in\R_{+}.
\end{multline}
By the definition of Nash equilibrium, and using \eqref{asl} and \eqref{als2}, it follows that, for every $i=1,...,N$,
\begin{align*}
&\mathcal{J}_{i}(\mathbf{x}_0,
\overline{\mathbf{u}}_{-i}(\cdot);
\overline{\mathbf{u}}_i(\cdot))
\geq 
\mathcal{J}_{i}(\mathbf{x}_0;
\overline{\mathbf{u}}_{-i}(\cdot);
\widehat{\mathbf{u}}_i(\cdot))\\
&=  \left\langle
\mathbf{b}_i(0),\mathbf{x}_0
\right\rangle
+
\int_{0}^{\infty }e^{-\rho_i t}
\left[
\left\langle
\mathbf{b}_i(t),\mathbf{f}
(t,\overline{\mathbf{u}}_{-i}(t),\widehat{\mathbf{u}}_i(t))
\right\rangle +
h_i(t,\overline{\mathbf{u}}_{-i}(t),\widehat{\mathbf{u}}_i(t))
\right]dt\\
&= \left\langle
\mathbf{b}_i(0),\mathbf{x}_0
\right\rangle
+
\int_{0}^{\infty }e^{-\rho_i t}\sup_{\mathbf{u}_{i}\in  \mathbf{U}_{i,\overline{\mathbf{u}}_{-i}(t)}}
\left[
\left\langle
\mathbf{b}_i(t),\mathbf{f}
(t,\overline{\mathbf{u}}_{-i}(t),\mathbf{u}_i)
\right\rangle +
h_i(t,\overline{\mathbf{u}}_{-i}(t),\mathbf{u}_i)
\right]dt\\
&\geq  \left\langle
\mathbf{b}_i(0),\mathbf{x}_0
\right\rangle
+
\int_{0}^{\infty }e^{-\rho_i t}
\left[
\left\langle
\mathbf{b}_i(t),\mathbf{f}
(t,\overline{\mathbf{u}}_{-i}(t),\overline{\mathbf{u}}_i(t))
\right\rangle +
h_i(t,\overline{\mathbf{u}}_{-i}(t),\overline{\mathbf{u}}_i(t))
\right]dt\\
&= \mathcal{J}_{i}(\mathbf{x}_0;
\overline{\mathbf{u}}_{-i}(\cdot);
\overline{\mathbf{u}}_i(\cdot)).
\end{align*}
Thus, the above inequalities are in fact equalities and, thus, we have
\begin{align*}
&\sup_{\mathbf{u}_{i}\in  \mathbf{U}_{i,\overline{\mathbf{u}}_{-i}(t)}}
\left[
\left\langle
\mathbf{b}_i(t),\mathbf{f}
(t,\overline{\mathbf{u}}_{-i}(t),\mathbf{u}_i)
\right\rangle +
h_i(t,\overline{\mathbf{u}}_{-i}(t),\mathbf{u}_i)
\right]\\
&= \left\langle
\mathbf{b}_i(t),\mathbf{f}
(t,\overline{\mathbf{u}}_{-i}(t),\overline{\mathbf{u}}_i(t))
\right\rangle +
h_i(t,\overline{\mathbf{u}}_{-i}(t),\overline{\mathbf{u}}_i(t)), \ \ \mbox{for a.e.} \ t\in\R_{+}.
\end{align*}
Thus, the first claim follows.
\medskip 

The claim that $\overline{\bu}(\cdot)$ is a Nash equilibrium for every initial condition then follows combining the above with (i).
\medskip

\item[(iii)] This is immediate from point (ii).

\medskip

\item[(iv)] In the absence of state constraints, any single-valued measurable selection $\widehat\bu(\cdot)$ as in the claim belongs to $\mathcal{U}_{G}$ and, by construction, satisfies the requirements of point (i). Hence, the claim follows.
\end{itemize}
\end{proof}

\begin{proof}[Proof of Proposition \ref{uniqueaffine}]
For simplicity of notation and clearness of exposition, we will illustrate the proof for the case $N=2$. Notice that, given the assumptions on $h_{1}$ and $h_{2}$, Theorem \ref{th:staticreductionNash} implies that there exists a unique open-loop Nash equilibrium for the game, which also constitutes a (degenerate) MPE.
Let the initial time be $t_{0}\in\R_{+}$ and consider an MPE  $\widehat\bvarphi=(\widehat \bvarphi_{1}, \widehat\bvarphi_{2})\in\mathcal{M}_{G}^{L}$. We identify $\widehat\bvarphi$ with
$\big((\widehat{\mathbf{L}}_{1},\widehat{\mathbf{w}}_{1}), (\widehat{\mathbf{L}}_{2},\widehat{\mathbf{w}}_{2})\big)
$, and let $\widehat{\bx}^{t_0,\bx_{0}}(\cdot)$ stand for the solution to the closed-loop equation associated with $\widehat\bvarphi$, starting at $(t_{0},\bx_{0})$. Remark \ref{rem:on} implies that $\widehat{{\bu}}_{1}(\cdot):= \widehat \bvarphi_{1}(\cdot, \widehat{\bx}^{t_{0},\bx_{0}}(\cdot))$ is an optimal control for the problem:
$$
\sup_{\bu_{1}(\cdot)\in  \mathcal{U}_{G}^{1}(t_{0})} \int_{t_{0}}^{\infty} e^{-\rho_{1}t} \big[\langle \mathbf{a}_{1}(t), \mathbf{x}(t)\rangle + h_1(t,{\bu}_1(t))\big]dt,
$$
under the state equation 
$$
\begin{cases}
\bx'(t)
=\mathbf{A}(t)\bx(t) + \mathbf{P}(t) \left(\begin{array}{c}
\bu_{1}(t)\\
(\widehat{\mathbf{L}}_{2}(t)\bx(t)+
\widehat{\mathbf{w}}_{2}(t))
\end{array}
\right)+\mathbf{j}(t)\\ \\
{\bx}(0)=\mathbf{x}_0\in \mathbf{X},
\end{cases}
$$
where
$$ \mathcal{U}_{1}^{G}(t_{0})=\Big\{{\bu}_{1}(\cdot)\in L^{1}_{\rho_{1}}([t_{0},\infty);\mathbf{U}_{1}): \  \   t\mapsto h_{1}(t, \bu_{1}(t)) \in L^{1}_{\rho_{1}}([t_{0},\infty);\R)\Big\}.
$$
This optimal control problem has the same structure as the one investigated in Subsection \ref{app:reductioncontrol}. Namely, it involves a linear dependence of the state variable in the functional. This allows us to perform the same transformation of the functional, ending up with a family of temporary optimization problems parametrized by $t$. For all $t\in\R_{+}$, each of these temporary problems admits a unique solution. Denote these by $\overline\bu_{1,t}$. Since $\widehat{\bvarphi}$ is a MPE, we have that, for a.e. $t_{0}\in\R_{+}$,
$$\overline{\bu}_{1, t_{0}} =\widehat{\bu}_{1}(t_{0})= \widehat{\bvarphi}_{1}(t_{0},\bx_{0}),\ \ \ \ \forall \bx_{0}\in\mathbf{X}.$$
It follows that $\widehat{\bvarphi}_{1}$ cannot depend on $\mathbf{x}$. By symmetry, the same is true for $\widehat{\bvarphi}_{2}$. Thus, $\widehat\bvarphi$ is the unique open-loop Nash equilibrium of the game.
\end{proof}

\section{The GP, RP, and Nash equilibrium solutions under Assumption \ref{hp:easycase}}

Here we give the proofs for the GP, RP, and Nash equilibrium cases discussed in section \ref{sec:institutionalarrangements}. In each case we state and prove specific results under Assumption \ref{hp:easycase}.

\subsection{The Global Planner's case (GP)}
\label{subapp:generalplannerspecificcas}

\begin{proof}[Proof of Corollary \ref{cor:rewriting-planner-max}]
The proof follows by applying Theorem \ref{th:staticreduction} using the specific forms in \eqref{eq:defplannerxu}-\eqref{eq:defplannerl}.
\end{proof}
Next, we discuss the case under Assumption \ref{hp:easycase}.
\begin{Proposition}
\label{pr:optimum-planner-easycase}
Suppose Assumption \ref{hp:easycase} holds and assume $\sigma_{1},\sigma_{2}\neq 1$ (non-logarithmic utility). Then there exists a unique solution to the problem given in
\eqref{eq:staticglobalplannerSpecific}-\eqref{eq:StaticConstraintsPlannerSpecific}. Moreover, the GP's policy is characterized by the following:
\begin{equation}  
\label{eq:C1C2planner-easy}
C_1=\left (\frac{\overline{A}-1}{(\gamma_{1}+\gamma_2) \Phi \eta_K} \right )^{1/\sigma_1}, \qquad 
C_2=\left (\frac{\overline{A}-1}{(\gamma_{1}+\gamma_2) \Phi \eta_K} \right )^{1/\sigma_2}, 
\end{equation}
\begin{equation}
\label{eq:Raplanner-bis}
R_a= 0,
\end{equation}
\begin{equation}
\label{eq:RbB2plannernew-bis-ultra}
R_b \in \argmax \left [ -R_b +  g(R_b)^{\frac{1}{1-\theta_2}} \left ( \frac{\theta_2}{\frac{\eta_K}{(\overline{A} -1)\eta_B}} \right )^{\frac{1}{1-\theta_2}} \left ( \frac{1}{\theta_2} -1 \right )\right ],
\end{equation}
\begin{equation}
\label{eq:B1planner-bis}
B_1=\left ( \frac{{\eta_B}\theta_1}{\eta_K}
(\overline{A}-1) \right )^{\frac{1}{1-\theta_1}}, 
\qquad   B_2= \left ( \frac{\eta_B g(R_b) \theta_2}{\eta_K}(\overline{A}-1)   \right)^{\frac{1}{1-\theta_2}},
\end{equation}
and 
\begin{equation}
\label{eq:I1I2planner-bis}
\displaystyle   K_1 = \frac{1}{\overline{A}-1} (C_1 + C_2 + B_1 + B_2 + R_b),\qquad   K_2 = 0.
\end{equation}
\end{Proposition}
\begin{proof}
By Theorem \ref{cor:rewriting-planner-max}, under Assumption \ref{hp:easycase} the GP maximizes
\begin{equation}
\label{eq:staticglobalplanner-proof}
\max_{C,B,K,R_a,R_b} \; \Big  [ \frac{C_1^{1-\sigma_1}}{1-\sigma_1}  +  \frac{C_2^{1-\sigma_2}}{1-\sigma_2}  - (\gamma_{1}+\gamma_2) \Phi\left ( \eta_K (K_1+K_2)- \eta_B (B_1(t)^{\theta_1} 
+ g(R_b) B_2(t)^{\theta_2} ) \right )
\Big ]
\end{equation}
under the constraints
\begin{equation}
\label{eq:StaticConstraintsPlanner-proof}
\left \{
\begin{array}{l}
C, B, K, R_a, R_b \geq 0\\[8pt]
C_1 + C_2 + B_1 + B_2 + K_1 + K_2 + R_a + R_b \leq
\overline{A} \left [K_1 + h(R_a) K_2 \right ].
\end{array}
\right .
\end{equation}
Observe that for any plan satisfying the second constraint in  \eqref{eq:StaticConstraintsPlanner-proof} with strict inequality, we can construct a feasible plan which satisfies it with equality by increasing $C_1$ (or $C_2$). Hence, we can assume without loss of generality that \eqref{eq:StaticConstraintsPlanner-proof} holds as an equality.
After a simple substitution, our problem is equivalent to 
\begin{multline}
\label{eq:staticglobalplanner-proof-2}
\max_{C,B,K,R_a,R_b} \; \bigg  [ \frac{C_1^{1-\sigma_1}}{1-\sigma_1}  +  \frac{C_2^{1-\sigma_2}}{1-\sigma_2}  - (\gamma_{1}+\gamma_2) \Phi \left ( 
- \eta_B (B_1(t)^{\theta_1} 
+ g(R_b) B_2(t)^{\theta_2} ) \right )\\
- (\gamma_{1}+\gamma_2) \Phi \eta_K \left ( \frac{1}{\overline{A} - 1} \left ( C_1 + C_2 + B_1 + B_2  + \overline{A} (1 -  h(R_a)) K_2 + R_a + R_b \right )    \right )
\bigg ].
\end{multline}
Since $h<1$,  using a simple argument by contradiction we can reduce the problem to one where $R_a = K_2=0$. Indeed, if a policy involves $K_2 >0$, since its coefficient in \eqref{eq:staticglobalplanner-proof-2} is strictly positive, we can construct an alternative plan implying a strictly higher payoff by increasing $C_1$ or $C_2$,  and setting $K_2=0$. We can then use the same argument to establish that $R_a=0$. This, in turn, implies that the second constraint in \eqref{eq:StaticConstraintsPlanner-proof} becomes
\[
K_1 = \frac{1}{\overline{A} -1} (C_1 + C_2 + B_1 + B_2 + R_b ).
\]
Hence the problem is equivalent to
\begin{multline}
\label{eq:staticglobalplanner-proofbis}
\max_{C_1, C_2, B_1, B_2, R_b} \; \bigg  [ \frac{C_1^{1-\sigma_1}}{1-\sigma_1}  +  \frac{C_2^{1-\sigma_2}}{1-\sigma_2}  - (\gamma_{1}+\gamma_2) \Phi \left ( 
- \eta_B (B_1(t)^{\theta_1} 
+ g(R_b) B_2(t)^{\theta_2} ) \right )\\
- (\gamma_{1}+\gamma_2) \Phi \eta_I \left ( \frac{1}{\overline{A} - 1} \left ( C_1 + C_2 + B_1 + B_2  + R_b \right )    \right )
\bigg ]
\end{multline}
under the non-negativity constraints: $C_1, C_2, B_1, B_2, R_b \geq 0$.

The objective function in \eqref{eq:staticglobalplanner-proofbis} is well defined, continuous and negatively coercive in $\R^5_+$. This implies that a maximum exists. Moreover, by Assumption \ref{hp:easycase}-(v) it is strictly convex, so the maximum is unique.

Next, we find it convenient to separate the objective in \eqref{eq:staticglobalplanner-proofbis} into four parts:
\begin{multline}
F_1(C_1) + F_2(C_2) + F_3(B_1) + F_4(B_2, R_b) := 
\left [ \frac{C_1^{1-\sigma_1}}{1-\sigma_1} -  \frac{(\gamma_{1}+\gamma_2) \Phi \eta_K}{\overline{A} -1} C_1 \right ] \\
+ \left [ \frac{C_2^{1-\sigma_2}}{1-\sigma_2} -  \frac{(\gamma_{1}+\gamma_2) \Phi \eta_K}{\overline{A} -1} C_2 \right ] +  
(\gamma_{1}+\gamma_2) \Phi \left [  - \frac{\eta_K}{\overline{A} -1} B_1 + \eta_B  B_1^{\theta_1} \right ] \\ + 
(\gamma_{1}+\gamma_2) \Phi \left [ - \frac{\eta_K}{\overline{A} -1} (B_2 + R_b) + \eta_B g(R_b) B_2^{\theta_2} \right ]. 
\end{multline}
We proceed by maximizing each part independently. By maximizing $F_1$, $F_2$ and $F_3$ we obtain (\ref{eq:C1C2planner-easy}) and the first part of (\ref{eq:B1planner-bis}). 
Maximizing $F_4$ w.r.t. $B_2$, we obtain the second part of (\ref{eq:B1planner-bis}). However, the value of $R_b$ is not yet explicit in this expression. We proceed by substituting the expression for the optimal $B_2$ in $F_4$. Then, maximizing the expression w.r.t. $R_b$ (skipping the factor $(\gamma_{1}+\gamma_2) \Phi$, as it multiplies the entire expression), we obtain:
\[
- \frac{\eta_K}{\overline{A} -1} \left [ \left ( \frac{\eta_K}{\eta_B g(R_b)  \theta_2}\frac{1}{\overline{A}-1}   \right)^{\frac{1}{\theta_2-1}} + R_b \right ]  + \eta_B g(R_b) \left ( \frac{\eta_K}{\eta_B g(R_b) \theta_2}\frac{1}{\overline{A}-1}   \right)^{\frac{\theta_2}{\theta_2-1}}.
\]
The last expression can be rewritten as
\begin{equation}
\label{eq:chesevealremark}
\frac{\eta_K}{\overline{A} -1} \left [ -R_b +  g(R_b)^{\frac{1}{1-\theta_2}} \left ( \frac{\theta_2}{\frac{\eta_K}{(\overline{A} -1)\eta_B}} \right )^{\frac{1}{1-\theta_2}} \left ( \frac{1}{\theta_2} -1 \right )\right ],
\end{equation}
and (\ref{eq:RbB2plannernew-bis-ultra}) follows. 
\end{proof}

\begin{Proposition}
\label{pr:optimum-planner-easycase-log}
Let Assumption \ref{hp:easycase} hold and assume logarithmic utility. Then there exists a unique solution to the problem
\eqref{eq:staticglobalplannerSpecific}-\eqref{eq:StaticConstraintsPlannerSpecific}. Moreover, the GP's policy is characterized by the following:
\begin{equation*}  
C_1=\left (\frac{\overline{A}-1}{(\gamma_{1}+\gamma_2) \Phi \eta_K} \right ), \quad C_2=\left (\frac{\overline{A}-1}{(\gamma_{1}+\gamma_2) \Phi \eta_K} \right ), 
\end{equation*}
\begin{equation*}
R_a = 0,
\end{equation*}
\begin{equation}
\label{eq:RbB2plannernew-bis-ultra-log}
R_b \in \argmax \left [ -\frac{\eta_K}{(\overline{A} -1)}R_b  +  (1-\theta_2)\theta_2^\frac{\theta_2}{1-\theta_2}(\eta_B g(R_b))^{\frac{1}{1-\theta_2}} \left ( \frac{(\overline{A} -1)}{\eta_K} \right )^{\frac{\theta_2}{1-\theta_2}} \right ],
\end{equation}
\begin{equation*}
, \qquad B_1=\left ( \frac{{\eta_B}\theta_1}{\eta_K}
(\overline{A}-1) \right )^{\frac{1}{1-\theta_1}}, \quad B_2=\left ( \frac{{\eta_B}g(R_b^{GP})\theta_2}{\eta_K}
(\overline{A}-1) \right )^{\frac{1}{1-\theta_2}},
\end{equation*}
and
\begin{equation*}
K_1 = \frac{1}{\overline{A}-1} (C_1 + C_2 + B_1 + B_2 + R_b),\qquad   K_2 = 0.
\end{equation*}
\end{Proposition}

\begin{proof}
The proof follows the same lines as in the proof of Proposition \ref{pr:optimum-planner-easycase} and will be omitted.
\end{proof}

\subsection{The restricted planner's case (RP)}
\label{subapp:restrictedplannerspecificcas}

\begin{proof}[Proof of Corollary \ref{cor:rewriting-planner-notransfer}]
As in the proof of Corollary \ref{cor:rewriting-planner-max} we define the terms of the optimal control problem as in \eqref{eq:defplannerxu}-\eqref{eq:defplannerg} and define the constraint on the control variables as 
\begin{multline}
\mathbf{l}(t,\mathbf{u}(t)) := \big(-C(t),-B(t),-K(t),-R_a(t),-R_b(t),\\ C_1(t)+B_1(t)+K_1(t)+R_a(t)+R_b(t) - \overline{A}(t)f_1(K_1(t)), \\ C_2(t)+B_2(t)+K_2(t) - h(R_a(t))\overline{A}(t)f_2(K_2(t))\big)^{T}.
\end{multline}
The claim then follows by applying Theorem \ref{th:staticreduction}.
\end{proof}
We will again discuss the optimum under Assumption \ref{hp:easycase}.

\begin{Proposition}
\label{pr:planner-notrasnfer-easy}
Suppose Assumption \ref{hp:easycase} holds and $\sigma\neq 1$ (non-logarithmic utility case). Then the RP’s policies are characterized by the following:
\begin{equation}\label{eq:C1C2-RP-specific}
C_1=\left ( (\gamma_1 + \gamma_2)\Phi \eta_K\frac{1}{\overline{A}-1} \right ) ^{-1/\sigma_1}, \quad C_2 = \left (\frac{\overline{A} h(R_a) -1}{(\gamma_1 + \gamma_2)\Phi \eta_K} \right )^{1/\sigma_2}
\end{equation}
\begin{multline}
\label{eq:RaRb-RP-specific}
(R_a ,R_b) \in\argmax \bigg \{
\frac{\sigma_2}{1-\sigma_2}\left (\frac{1}{(\gamma_1 + \gamma_2)\Phi \eta_K} \right )^{\frac{1}{\sigma_2}} (\overline{A}h(R_a) -1)^{\frac{1-\sigma_2}{\sigma_2}}
\\ +  
{(1-\theta_2) \eta_B g(R_b)  \left ( \frac{\eta_K}{\eta_Bg(R_b) \theta_2}\cdot\frac{1}{\overline{A}h(R_a) -1} \right )^{\frac{\theta_2}{\theta_2-1}}
- \frac{\eta_K}{\overline{A} -1} (R_a + R_b) \bigg \}}
\end{multline}
\begin{equation}\label{eq:B1B2-RP-specific} 
B_1=\left ( \frac{\eta_B\theta_1}{\eta_K}(\overline{A}-1) \right )^{\frac{1}{1-\theta_1}},    \quad B_2=\left ( \frac{\eta_K}{\eta_Bg(R_b)\theta_2}\frac{1}{\overline{A}h(R_a)-1} \right )^{\frac{1}{\theta_2-1}},   
\end{equation}
and
\begin{equation*}
K_1 = \frac{1}{\overline{A}-1} (C_1  + B_1 +R_a + R_b),\qquad   K_2 = \frac{C_2 + B_2}{\overline{A} h(0)-1}.
\end{equation*}
\end{Proposition}

\begin{proof} We find if convenient to use Corollary \ref{cor:rewriting-planner-notransfer} together with the budget constraint, in order to rewrite the optimization problem into separate parts. Thus, we need to maximize the following:
\begin{multline}
F_1(C_1) + F_2(C_2,R_a) + F_3(B_1) + F_4(B_2,R_a) + F_5(R_a, R_b) := 
\left [ \frac{C_1^{1-\sigma_1}}{1-\sigma_1} -  \frac{(\gamma_{1}+\gamma_2) \Phi \eta_K}{\overline{A} -1} C_1 \right ] \\
+ \left [ \frac{C_2^{1-\sigma_2}}{1-\sigma_2} -  \frac{(\gamma_{1}+\gamma_2) \Phi \eta_K}{\overline{A} h(R_a) -1} C_2 \right ] +  
(\gamma_{1}+\gamma_2) \Phi \left [  - \frac{\eta_K}{\overline{A} -1} B_1 + \eta_B  B_1^{\theta_1} \right ] \\ + 
(\gamma_{1}+\gamma_2) \Phi \left [ - \frac{\eta_K}{h(R_a)\overline{A} -1} B_2 + \eta_B g(R_b) B_2^{\theta_2} \right ]
+
(\gamma_{1}+\gamma_2) \Phi \left [ - \frac{\eta_K}{\overline{A} -1} (R_a + R_b) \right ].
\end{multline}
Maximizing $F_1$ and $F_3$ gives the two equations in \eqref{eq:C1C2-RP-specific}, while maximizing $F_2$ and $F_4$ results in the two expressions in \eqref{eq:B1B2-RP-specific}. Using the expressions in \eqref{eq:B1B2-RP-specific}, to find $R_a$ and $R_b$ we need to maximize (ignoring the common factor $(\gamma_{1}+\gamma_2) \Phi$):
\begin{multline}
\frac{\sigma_2}{1-\sigma_2}\left (\frac{1}{(\gamma_1 + \gamma_2)\Phi \eta_K} \right )^{\frac{1}{\sigma_2}} (h(R_a) \overline{A}-1)^{\frac{1-\sigma_2}{\sigma_2}}\\ +  
{(1-\theta) \eta_B g(R_b)  \left ( \frac{\eta_K}{\eta_B g(R_b) \theta_2}\frac{1}{h(R_a) \overline{A}-1} \right )^{\frac{\theta_2}{\theta_2-1}}
- \frac{\eta_K}{\overline{A} -1} (R_a + R_b).}
\end{multline}
If $(R_a, R_b)$ is a maximum point for this expression, the control is optimal.
\end{proof}

\begin{Proposition}
\label{pr:planner-notrasnfer-easy-log}
Suppose Assumption \ref{hp:easycase} holds and assume logarithmic utility. Then the RP’s policies are characterized by the following:
\begin{equation}
C_1^{RP}=\left (\frac{\overline{A}-1}{(\gamma_{1}+\gamma_2) \Phi \eta_K} \right ), \qquad C_2^{RP}=
\left (\frac{\overline{A} h(R_a^{p_2}) -1}{(\gamma_1 + \gamma_2)\Phi \eta_K} \right ),
\end{equation}
\begin{multline}
\label{eq:Ra-Rb-RP-log}
(R_a^{RP} ,R_b^{RP}) \in\argmax 
\bigg \{ { \frac{1}{(\gamma_1+\gamma_2)\Phi}\log\left(\dfrac{\overline{A}h(R_a)-1}{(\gamma_1+\gamma_2)\Phi\eta_K} \right)} \\
+(1-\theta_2)\theta_2^\frac{\theta_2}{1-\theta_2} \left ( \eta_B  g(R_b) \right)^{\frac{1}{1 - \theta_2}} \left ( \frac{\overline{A}h(R_a) -1}{\eta_K} \right )^{\frac{\theta_2}{1-\theta_2}}
- \frac{\eta_K}{\overline{A} -1} (R_a + R_b) \bigg \}
\end{multline}
\begin{equation}
B_1^{RP}=\left ( \frac{\eta_B\theta_1}{\eta_K}(\overline{A}-1) \right )^{\frac{1}{1-\theta_1}},  \quad B_2^{RP}=\left ( \frac{\eta_Bg(R_b^{RP})\theta_2}{\eta_K}(\overline{A}h(R_a)-1) \right )^{\frac{1}{1-\theta_2}},
\end{equation}
and
\begin{equation*}
K_1 = \frac{1}{\overline{A}-1} (C_1  + B_1 +R_a + R_b),\qquad   K_2 = \frac{C_2 + B_2}{\overline{A} h(0)-1}.
\end{equation*}
\end{Proposition}
\begin{proof}
The proof follows the lines of the proof of Proposition \ref{pr:planner-notrasnfer-easy} and will be omitted.
\end{proof}

\subsection{Nash equilibrium (N)}
\label{subapp:Nashspecificcas}

\begin{Proposition}
\label{pr:rewriting-Nash} 
Let 
$(C(\cdot), K(\cdot), B(\cdot),  R(\cdot)) \in\mathcal{U}^{RP}$ and assume that 
$t\mapsto G(t, C(t), B(t), R_b(t)) \in L_{\rho}^1(\mathbb{R}_+)$.
Then the objective functional of country $i$ in expression \eqref{eq5:defUPbisgameNashtransfer} can be written as
\begin{equation}
\label{eq:defUPbisnew-Nash}
U_i=
\gamma_{i}\left[
\frac{\overline{S}}{\rho} -
\frac{P(0)}{\rho}-\frac{T(0)}{\rho+\phi}
\right] + \int_{0}^{\infty }e^{-\rho t}
\left[u_{i}\left(C_i(t)\right)
-\gamma_{i}\Phi G\left(t,K(t),B(t),R_b(t)\right)\right]dt,
\end{equation}
where $\Phi$ is defined in (\ref{eq:defPhi}). 
\end{Proposition}
\begin{proof}
We will apply Proposition \ref{pr:rewriting-generalNash}. We begin by defining  
\begin{equation}
\label{eq:stateNash}
\mathbf{x}(t) := \begin{pmatrix} P(t) \\ T(t) \\ \end{pmatrix},
\end{equation}
\begin{equation}
\mathbf{u}^{1}(t) := \left(C_1(t),B_1(t),K_1(t),R_a(t),R_b(t)\right)^{T}, \quad \mathbf{u}^{2}(t) := \left(C_2(t),B_2(t),K_2(t)\right)^{T}, 
\end{equation}
\begin{equation}
\mathbf{A}(t) := \begin{pmatrix} 0 & 0 \\ 0 & -\phi \\ \end{pmatrix}, \quad \mathbf{f}(t,\mathbf{u}(t)) := \begin{pmatrix} \phi_L G\left(t,K(t),B(t),R_b(t)\right) \\ (1-\phi_L)\phi_0 G\left(t,K(t),B(t),R_b(s)\right)  \end{pmatrix},
\end{equation}
\begin{equation}
\mathbf{a}^{1}(t) := -\gamma_1\begin{pmatrix} 1 \\ 1 \\ \end{pmatrix},\quad  h^1(t,\mathbf{u}^{1}(t),\mathbf{u}^{-1}(t)) := u_1(C_1(t)) + \gamma_1\overline{S},
\end{equation}
\begin{equation}
\mathbf{a}^{2}(t) := -\gamma_2\begin{pmatrix} 1 \\ 1 \\ \end{pmatrix},\quad  h^2(t,\mathbf{u}^{2}(t),\mathbf{u}^{-2}(t)) := u_2(C_2(t)) + \gamma_2\overline{S},
\end{equation}
\begin{equation}
\mathbf{g}^{1}(t,\mathbf{x}(t)) := 0,\quad \mathbf{g}^{2}(t,\mathbf{x}(t)) := 0,
\end{equation}
and
\begin{multline}
\label{eq:lNash-1}
\mathbf{l}^1(t,\mathbf{u}(t)) := \big(-C_1(t),-B_1(t),-K_1(t),-R_a(t),-R_b(t),\\ C_1(t)+B_1(t)+K_1(t)+R_a(t)+R_b(t) - \overline{A}(t)f_1(K_1(t))\big)^{T},
\end{multline}
\begin{multline}
\label{eq:lNash-2}
\mathbf{l}^2(t,\mathbf{u}(t)) := \big(-C_2(t),-B_2(t),-K_2(t), C_2(t)+B_2(t)+K_2(t) - h(R_a(t))\overline{A}(t)f_2(K_2(t))\big)^{T}.
\end{multline}
Under these choices, the general formulation in Section \ref{app:reductionNash} reduces to the problem in \eqref{eq5:defUPbisgameNashtransfer}. Moreover, we have
\begin{equation*}
\Phi_A^{\ast}(t+\tau,t) = \exp\left(\tau A\right) = \begin{pmatrix} 1 & 0 \\ 0 & e^{-\phi \tau} \\ \end{pmatrix},
\end{equation*}
and
\begin{equation*}
\mathbf{b}^{i}(t) = -\gamma_i \begin{pmatrix} {1}/{\rho} \\ {1}/{(\rho+\phi)} \\ \end{pmatrix},\quad i \in \{1,2\}.
\end{equation*}
Thus, Assumption \ref{hp:reductionNash} is readily satisfied, due to the assumptions in Section \ref{subsect:theclimatemodel}. Applying Proposition  \ref{pr:rewriting-generalNash} completes the proof.
\end{proof}

\begin{proof}[Proof of Corollary \ref{cor:rewriting+conditions-Nash}]

As noted earlier, in order to rule out outcomes when (\ref{eq:vincolo1gameplayer1}) or (\ref{eq:vincolo2gameplayer2}) do not hold, we assume that the payoff to either player is $-\infty$ if their budget constraint is violated. It the follows that Player 1 will choose $R_a=0$. For, if this is not the case, their payoff could increase by reducing $R_a$ and increasing $C_1$ (keeping $B_1$, $R_b$ and $K_1$ constant). The result then follows by applying Theorem \ref{th:staticreductionNash} using the definitions in \eqref{eq:stateNash}-\eqref{eq:lNash-1}-\eqref{eq:lNash-2}.
\end{proof}

Next, we once again turn to the special case in the main text.

\begin{Proposition}
\label{pr:Nash-easy}
Suppose that Assumption \ref{hp:easycase} holds and assume $\sigma\neq 1$ (non-logarithmic utility case). Then Nash equilibrium is characterized by the following equations:
\begin{equation}  \label{eq:C1B1easy-game1}
C_1=\left ( \gamma_{1}\Phi \eta_K\frac{1}{\overline{A}-1} \right ) ^{-1/\sigma_1}, \qquad 
C_2= \left(\gamma_{2}\Phi \eta_K\frac{1}{\overline{A}h(0) -1}\right)^{-1/\sigma_2},
\end{equation}
\begin{equation}
\label{eq:RaRbN}
R_a^{N}= 0, \qquad R_b^{N} 
\in \argmax\bigg \{  
(1-\theta_2)\theta_2^\frac{\theta_2}{1-\theta_2} \left ( \eta_B  g(R_b) \right)^{\frac{1}{1 - \theta_2}} \left ( \frac{\overline{A}h(0) -1}{\eta_K} \right )^{\frac{\theta_2}{1-\theta_2}}
- \frac{\eta_K}{\overline{A} -1}R_b \bigg \}
\end{equation}
\begin{equation*}
B_1=\left ( \frac{\eta_K}{\eta_B\theta_1}\frac{1}{\overline{A}-1} \right )^{\frac{1}{\theta_1-1}}, B_2 = \left ( \frac{\eta_B g(R_b) \theta_2 (\overline{A} h(0) -1)}{\eta_K} \right )^{\frac{1}{1 - \theta_2}},
\end{equation*}
and
\begin{equation*}
K_1 = \frac{1}{\overline{A}-1} (C_1  + B_1 + R_b),\qquad   K_2 = \frac{C_2 + B_2}{\overline{A} h(0)-1}.
\end{equation*}
\end{Proposition}
\begin{proof}
We look for a solution of the form described in Proposition \ref{cor:rewriting+conditions-Nash}. As demonstrated in the beginning of the proof of Corollary \ref{cor:rewriting+conditions-Nash}, we necessarily have $R_a\equiv 0$. This implies that $R_a = 0$ and the expressions for $K_1$ and $K_2$ follow.

Using the equation for $K_1$ and $K_2$, we can rewrite the maximization problem of Player 1 as
\begin{multline}
\max_{C_1, B_1, R_b} u_{1}\left(C_1\right) -\gamma_{1}\Phi 
\eta_K \left (\frac{1}{\overline{A}-1}
(C_1  + B_1  + R_b) + \frac{1}{h(0)\overline{A}-1} (C_2 + B_2) \right ) \\ 
+ \gamma_{1}\Phi \eta_B \left (B_1^{\theta_1} +g(R_b) B_2^{\theta_2} \right) = F_1 (C_1) - F_2 (B_1) - F_3 (R_b, B_2) - F_4(B_2,C_2)=
\\
= \left [ u_{1}\left(C_1\right) -  \frac{\gamma_{1}\Phi \eta_K}{\overline{A}-1}
C_1 \right ] - 
\left [ 
\frac{\gamma_{1}\Phi \eta_K}{\overline{A}-1} B_1 -  \gamma_{1}\Phi \eta_B B_1^{\theta_1}
\right ] - 
\left [
\frac{\gamma_{1}\Phi \eta_K}{\overline{A}-1} R_b - \gamma_{1}\Phi \eta_{B} g(R_b) B_2^{\theta_2} 
\right ]\\
- \gamma_{1}\Phi \eta_K \left [\frac{1}{h(0)\overline{A}-1} (C_2 + B_2)\right ].
\end{multline}
Maximizing the term containing $F_1$ and $F_2$ we obtain (\ref{eq:C1B1easy-game1}). The term containing $F_4$ does not depend of the decisions of Player 1, so it is not taken into account in his decision. The maximization of the term containing $F_3$ (recall that $g$ is concave) gives
\begin{equation}
\label{eq:maxR_b-Player1}    
\left \{
\begin{array}{ll}
R_b = 0,      & \text{if } g'(0) \leq \frac{1}{\overline{A}-1} \frac{\eta_K}{\eta_B} B_2^{-\theta_2} \\[8pt]
R_b = (g')^{-1} \left ( \frac{1}{\overline{A}-1} \frac{\eta_K}{\eta_B} B_2^{-\theta_2} \right ),      & \text{if } g'(0) > \frac{1}{\overline{A}-1} \frac{\eta_K}{\eta_B} B_2^{-\theta_2}.
\end{array}
\right .
\end{equation}

The maximization problem of Player 2 reads as 
\begin{multline}
\max_{C_2, B_2} u_{2}\left(C_2\right) -\gamma_{2}\Phi
\eta_K \left (\frac{1}{\overline{A}-1}
(C_1  + B_1  + R_b) + \frac{1}{h(0)\overline{A}-1} (C_2 + B_2) \right ) 
+ \gamma_{2}\Phi \eta_B \left (B_1^{\theta_1} +g(R_b) B_2^{\theta_2} \right)\\
= J_1 (C_2) - J_2(B_2, R_b) + J_3(C_1, B_1) \\
\\
= \left [ u_{1}\left(C_2\right) -  \frac{1}{h(0)\overline{A}-1} C_2 \right ] - 
\left [ 
\frac{\gamma_{2}\Phi \eta_K}{h(0)\overline{A}-1} B_2 - \gamma_{2}\Phi \eta_B g(R_b) B_2^{\theta_2} 
\right ] \\
+ \left [
-\gamma_{2}\Phi 
\eta_K \left (\frac{1}{\overline{A}-1}
(C_1  + B_1  + R_b) \right ) + \gamma_{2}\Phi \eta_B B_1^{\theta_1}
\right ].
\end{multline}

Maximizing the term $J_1$, we obtain the expression for $C_2$. The term $J_3$ does not depend of the decisions of Player 2. The maximization of the $J_2$ gives 
\begin{equation}
\label{eq:feedback-RP}
B_2 := \left ( \frac{\eta_B g(R_b) \theta_2 (\overline{A} h(0) -1)}{\eta_K} \right ) ^{\frac{1}{1 - \theta_2}}.
\end{equation}
Thus, a Nash equilibrium with $R_b=0$ exists if and only if this expression evaluated at $R_b=0$ satisfies the condition in the first line of (\ref{eq:maxR_b-Player1}); i.e. if and only if 
\[
g'(0) \leq \frac{1}{\overline{A}-1} \frac{\eta_K}{\eta_B} \left ( \frac{\eta_B g(0)  \theta_2 (\overline{A} h(0) -1)}{\eta_K} \right ) ^{\frac{- \theta_2}{1 - \theta_2}}.
\]
Rearranging this expression, we obtain the condition (i) in the text.

To establish the existence of a a Nash equilibrium with $R_b>0$, we need to find a pair $B_2$, $R_b$ satisfying the second line of (\ref{eq:maxR_b-Player1}) and (\ref{eq:feedback-RP}). Thus, we need to find $R_b^h$ such that 
\[
\frac{\eta_K}{\eta_B} = g'(R_b^h)^{1-\theta_2} {\theta_2^{\theta_2}}g(R_b^h)^{\theta_2} (\overline{A}-1)^{1-\theta_2} (\overline{A} h(0) -1)^{\theta_2}.
\]
Since $g'(R_b)^{1-\theta_2} g(R_b)^{\theta_2}\to0$ as $R_b \to +\infty$, the previous equation has a solution if and only if (ii) holds. The associated value of $B_2^h$ can be found using (\ref{eq:feedback-RP}). We remark that the condition in the second line of (\ref{eq:maxR_b-Player1}) holds, since $B_2^h$ can be found using the expression in (\ref{eq:maxR_b-Player1}). Thus, since $g$ is concave,
\[
g'(0) > g'(R_b^h) = \left ( \frac{1}{\overline{A}-1} \frac{\eta_K}{\eta_B} {B_2^h}^{-\theta_2} \right ).
\]
Implication (iii) then follows from (i) and (ii).
\end{proof}

\begin{Proposition}
\label{pr:Nash-easy-log}
Suppose that Assumption \ref{hp:easycase} holds and assume logarithmic utility. Then, Nash equilibrium is characterized by the following:
\begin{equation}
\label{eq:C1C2N-log}
C_1^{N}=\left ( \frac{\overline{A}-1}{\gamma_{1}\Phi \eta_K} \right ), \qquad C_2^{N}= \left(\frac{\overline{A}h(0) -1}{\gamma_{2}\Phi \eta_K}\right), 
\end{equation}
\begin{equation}
\label{eq:RaRbN-log}
R_a^{N}= 0, \qquad R_b^{N} 
\in \argmax\bigg \{  
(1-\theta_2)\theta_2^\frac{\theta_2}{1-\theta_2} \left ( \eta_B  g(R_b) \right)^{\frac{1}{1 - \theta_2}} \left ( \frac{\overline{A}h(0) -1}{\eta_K} \right )^{\frac{\theta_2}{1-\theta_2}}
- \frac{\eta_K}{\overline{A} -1}R_b \bigg \}
\end{equation}
\begin{equation}
\label{eq:B1B2N-log}
B_1^{N}=\left ( \frac{\eta_B\theta_1}{\eta_K}(\overline{A}-1) \right )^{\frac{1}{1-\theta_1}},\quad
B_2^{N} := \left ( \frac{\eta_B g(R_b^{N}) \theta_2}{\eta_K} (\overline{A} h(0) -1) \right )^{\frac{1}{1 - \theta_2}},
\end{equation}
and
\begin{equation*}
K_1 = \frac{1}{\overline{A}-1} (C_1  + B_1 + R_b),\qquad   K_2 = \frac{C_2 + B_2}{\overline{A} h(0)-1}.
\end{equation*}
\end{Proposition}
\begin{proof}
The proof follows along the lines of the proof of Proposition \ref{pr:Nash-easy} and will be omitted.
\end{proof}

\section{Comparisons under Assumption \ref{hp:easycase} and logarithmic payoffs}
\label{app:proofs}

Here we concentrate on the specific case 
described in Assumption \ref{hp:easycase} under logarithmic utility. This allows us to analytically compare the outcomes under the regimes studied in the previous sections, namely the global (GP) and restricted planner's (RP) solutions, as well as the Nash equilibrium (N) outcomes. These findings are later confirmed using numerical methods. The detailed results of the logarithmic utility case in the three arrangements are given in Propositions \ref{pr:optimum-planner-easycase-log}, \ref{pr:planner-notrasnfer-easy-log}, and \ref{pr:Nash-easy-log}. Here we report on the related comparisons.

\subsection{Transfers} 
In Propositions \ref{pr:optimum-planner-easycase-log}, \ref{pr:planner-notrasnfer-easy-log}, and \ref{pr:Nash-easy-log}, we derive the values for $R_a$ and $R_b$ in the three regimes. 
Regarding the order of the transfers in the different regimes, we have the following.

\begin{Proposition}
\label{pr:log-odinediR}
Assume that Assumption \ref{hp:easycase} holds, and $u_i(C) = \ln(C)$ $i=1,2$. Then
\[
0=R_a^{GP}=R_a^{N} \le R_a^{RP}
\]
and
\[
R_b^{N} \le R_b^{RP} \leq R_b^{GP}.\footnote{Using the expression in Proposition \ref{pr:planner-notrasnfer-easy-log}, we can conclude that $R_a^{RP}$ is  strictly positive provided that $\frac{\eta_K}{\eta_B}$ is sufficiently small.}
\]
\end{Proposition}
\begin{proof}
We compare the expressions for $R_a$ and $R_b$ under the three arrangements described in Section \ref{sec:institutionalarrangements}. Their respective expressions are given in Propositions \ref{pr:optimum-planner-easycase-log}, \ref{pr:planner-notrasnfer-easy-log}, and \ref{pr:Nash-easy-log}. The first part of the claim follows since we have already demonstrated that $0=R_a^{GP}=R_a^{N}$ (see the proofs of Proposition \ref{pr:optimum-planner-easycase} and of Corollary \ref{cor:rewriting+conditions-Nash}). 

For the second inequality, observe that the three expressions for $R_b$ in (\ref{eq:RbB2plannernew-bis-ultra}), (\ref{eq:RaRb-RP-specific}), and (\ref{eq:RaRbN}), can be rewritten respectively as follows: $R_b \in \argmax{-R_b + \mu(R_b)(\bar A -1)^{\frac{\theta_2}{1-\theta_2}}}$, $R_b \in \argmax{-R_b + \mu(R_b)(\bar Ah(R_a^{RP}) -1)^{\frac{\theta_2}{1-\theta_2}}}$, and $R_b \in \argmax{-R_b + \mu(R_b)(\bar A h(0) -1)^{\frac{\theta_2}{1-\theta_2}}}$, where $\mu(R_b)$ is a concave and increasing function. The claim then follows since Since $(\bar Ah(0) -1) \leq (\bar Ah(R_a^{RP}) -1) \leq (\bar A -1)$.
\end{proof}

\subsection{Consumption and abatement}
In Propositions \ref{pr:optimum-planner-easycase-log}, \ref{pr:planner-notrasnfer-easy-log}, and \ref{pr:Nash-easy-log},  we derive the values of $C_1$, $C_2$, $B_1$ and $B_2$ in the three regimes.  We have the following.
\begin{Proposition}
\label{pr:log-ordine2}
Suppose that Assumption \ref{hp:easycase} holds and $u_i(C) = \ln(C)$, $i=1,2$. Then

\[
C_1^{GP} =  C_1^{RP} < C_1^N, \qquad B_1^{GP} =  B_1^{RP} = B_1^N
\]
and
\[
C_2^{RP} < C_2^{GP},  \qquad  B_2^{N} \leq B_2^{RP}  < B_2^{GP}.
\]
Moreover, provided that the technological differences (before transfers) between North and South are not too large; i.e.,
\begin{equation}
\label{eq:perordine}
(1- h(0)) < \frac{\bar A - 1}{\bar A} \frac{\gamma_1}{\gamma_1 + \gamma_2}
\end{equation}
we have 
\begin{equation}
\label{eq:consumpopiualtoNas}
C_2^{RP} < C_2^{N}.
\end{equation}
\end{Proposition}
\begin{proof}
We compare the expressions of $C_1$ and $C_2$ for the three arrangements described in Section \ref{sec:institutionalarrangements}. Their expression are given in Propositions \ref{pr:optimum-planner-easycase-log}, \ref{pr:planner-notrasnfer-easy-log} and \ref{pr:Nash-easy-log}.

The results for Country $1$ are immediate. The first claim for Country $2$ is straightforward given the properties of $h$. The second follows from the fact that $R_b^{N} \leq  R_b^{RP} \leq R_b^{GP}$ and $0 = R_a^{N} = R_a^{GP} \leq R_a^{RP}$ (see Proposition \ref{pr:log-odinediR}), together with the fact that $h$ and $g$ are increasing functions.
Finally, (\ref{eq:consumpopiualtoNas}) follows directly from the respective expressions for consumption in the two cases, since (\ref{eq:perordine}) implies
\[
\left(\frac{\gamma_1 + \gamma_2}{\gamma_{2}}\right) > \left (\frac{\overline{A} -1}{\overline{A}h(0) -1} \right ).
\]
\end{proof}

\subsection{Welfare}

Next, we turn to a study of welfare in the two countries under the different regimes. We denote by $U_i$ the utility of country $i$ and we let $U:= U_1 + U_2$.  We have the following.

\begin{Proposition}
\label{pr:log-ordineutility}
Suppose that Assumption \ref{hp:easycase} holds, and $u_i(C) = \ln(C)$ $i=1,2$. Then
\begin{itemize}
\item[(i)] $U^{GP} > U^{RP} > U^{N}$,
\item[(ii)] $U_1^{GP} > U_1^{RP}$, 
\item[(iii)] $U_2^{GP} > U_2^{RP}$. 
\end{itemize}
\end{Proposition}
\begin{proof}

Directly from the definitions, it follows that $U^{GP} > U^{RP} > U^{N}$. In addition, we have that $U_1^{GP} > U_1^{RP}$, and $U_2^{GP} > U_2^{RP}$. To see this, decompose the payoff of country 1 and abstract the part in $B_1$ (which is the same in all the cases), to obtain
\begin{multline}
\label{eq:UGPVSURP-Utility1}
U_1^{GP} = \ln(C_1^{GP}) - \gamma_1\Phi \eta_K \frac{1}{\overline{A} - 1} C_1^{GP} - \gamma_1\Phi \eta_K \frac{1}{\overline{A} - 1} C_2^{GP} \\ + \gamma_1 \Phi \eta_B g(R_b^{GP}) \left ( B_2^{GP} \right )^{\theta_2} - \gamma_1\Phi \eta_K \frac{1}{\overline{A} - 1} B_2^{GP} - \gamma_1\Phi \eta_K \frac{1}{\overline{A} - 1} R_b^{GP} \\
=\ln \left ( \frac{\overline{A} -1}{(\gamma_1 + \gamma_2)\Phi \eta_K} \right ) - \frac{\gamma_1}{\gamma_1 + \gamma_2} - \frac{\gamma_1}{\gamma_1 + \gamma_2} \\+ \gamma_1 \Phi \eta_B g(R_b^{GP}) \left ( B_2^{GP} \right )^{\theta_2} - \gamma_1\Phi \eta_K \frac{1}{\overline{A} - 1} B_2^{GP} - \gamma_1\Phi \eta_K \frac{1}{\overline{A} - 1} R_b^{GP},
\end{multline}
and
\begin{multline}
\label{eq:UGPVSURP-Utility2}
U_1^{RP} = \ln(C_1^{RP}) - \gamma_1\Phi \eta_K \frac{1}{\overline{A} - 1} C_1^{RP} - \gamma_1\Phi \eta_K \frac{1}{h(R_a)\overline{A} - 1} C_2^{RP} \\ - \gamma_1\Phi \eta_K \frac{1}{\overline{A} - 1} R_a^{RP} + \gamma_1 \Phi \eta_B g(R_b^{RP}) \left ( B_2^{RP} \right )^{\theta_2} - \gamma_1\Phi \eta_K \frac{1}{\overline{A} - 1} B_2^{RP} - \gamma_1\Phi \eta_K \frac{1}{\overline{A} - 1} R_b^{RP} \\
=\ln \left ( \frac{\overline{A} -1}{(\gamma_1 + \gamma_2)\Phi \eta_K} \right ) - \frac{\gamma_1}{\gamma_1 + \gamma_2} - \frac{\gamma_1}{\gamma_1 + \gamma_2}  \\- \gamma_1\Phi \eta_K \frac{1}{\overline{A} - 1} R_a^{RP} + \gamma_1 \Phi \eta_B g(R_b^{RP}) \left ( B_2^{RP} \right )^{\theta_2} - \gamma_1\Phi \eta_K \frac{1}{h(R_a)\overline{A} - 1} B_2^{RP} - \gamma_1\Phi \eta_K \frac{1}{\overline{A} - 1} R_b^{RP}.
\end{multline}
The values of $B_2$ and of $R_b$ (respectively of $B_2$, $R_b$, and $R_a$) are chosen by the global planner (respectively the restricted planner) to maximize $ (\gamma_1 +  \gamma_2) \Phi$ times
\begin{equation}
\label{eq:UGPVSURP-1}
\eta_B g(R_b^{GP}) \left ( B_2^{GP} \right )^{\theta_2} - \eta_K \frac{1}{\overline{A} - 1} B_2^{GP} - \eta_K \frac{1}{\overline{A} - 1} R_b^{GP},
\end{equation}
respectively,
\begin{equation}
\label{eq:UGPVSURP-2}
- \eta_K \frac{1}{\overline{A} - 1} R_a^{RP} + \eta_B g(R_b^{GP}) \left ( B_2^{GP} \right )^{\theta_2} - \eta_K \frac{1}{\overline{A} - 1} B_2^{GP} -  \eta_K \frac{1}{\overline{A} - 1} R_b^{GP}.
\end{equation}
Since, for any choice of $B_2$, $R_b$, and $R_a$, the expression in (\ref{eq:UGPVSURP-2}) is smaller than (\ref{eq:UGPVSURP-1}), the last line of (\ref{eq:UGPVSURP-Utility1}) is smaller than the last line of (\ref{eq:UGPVSURP-Utility2}). Thus, the payoff to country 1 under the GP is higher than its payoff under the RP. The same argument holds for country 2.
\end{proof}

\subsection{Emissions}

Under Assumption \ref{hp:easycase}, we can separate the GHG emissions due to each player. Denote by $G_i$ the instantaneous net emissions of player $i$, given by 
\[
G_1 := \Lambda_I I_1 - \Lambda_D B_1^{\theta_1} \qquad and \qquad 
G_2:= \Lambda_I I_2 - \Lambda_D g(R_b) B_2^{\theta_2}.
\]

For the aggregate GHG emission flows, $G$, we have the following.
\begin{Proposition}
\label{pr:logcase-emissionranking}
Suppose that Assumption \ref{hp:easycase} holds and $u_i(C) = \ln(C)$ $i=1,2$. Then
\[
G^{GP} \leq G^{RP}.
\]
If, in addition, (\ref{eq:perordine}) holds, then
\[
G^{GP} \leq G^{RP} \leq G^{N}.
\]
\end{Proposition}
\begin{proof}
Proposition \ref{pr:log-ordine2} established that $C_1^{GP} = C_1^{RP}$ and Proposition \ref{pr:log-ordineutility} established that $U_1^{GP} > U_1^{RP}$. Since the payoff of country 1 is given by $\ln(C_1)$ minus a disutility part which is linear in total emissions, we obtain
\[
G^{GP} < G^{RP}.
\]
We prove now the second part. From Proposition \ref{pr:log-ordine2} we know that $C_1^{RP} < C_1^N$ and, under hypothesis (\ref{eq:perordine}) we also have (see (\ref{eq:consumpopiualtoNas}) $C_2^{RP} < C_2^N$. So the utility coming from consumption is higher in both countries in the Nash case than in the RP case. On the other hand we also know, from Proposition \ref{pr:log-ordineutility} that $U^{RP} > U^{N}$ so the only possibility is that the disutility coming from emissions is also lower i.e. that emissions are lower. This concludes the proof.
\end{proof}

\section{Robust control proofs}
\label{app:robusteness}

\subsection{The GP robust control problem}
\label{supapp:robustglobal}

\begin{proof}[Proof of Theorem \ref{th:maxmingeneral}]
We need to demonstrate that $U^R$, as a function of the control variables, is concave and, as a function of $(\gamma_1(\cdot), \gamma_2(\cdot))$ is convex. The proof then follows from Proposition 2.2, p.173, in \cite{ekeland1999convex}. We briefly describe how to check these conditions

\begin{itemize}
\item $(\gamma_1(\cdot), \gamma_2(\cdot))$ vary in the cone of non-negative functions of $L^2_\rho(\mathbb{R}^+, \mathbb{R}^2)$ which is convex, closed and non-empty.

\item $(C(\cdot), B(\cdot), R_a(\cdot), R_b(\cdot))$ vary in the subset of $L^6_\rho(\mathbb{R}^+, \mathbb{R}^2)$ given by non-negative functions satisfying (\ref{eq:plannerresourceconstraint}). Since the right-hand side of (\ref{eq:plannerresourceconstraint}) is a quasi-concave function of $K_1$, $R_a$ and $K_2$, the subset is convex. It is also clearly closed and non-empty.

\item The functional (\ref{eq:functionalrobust-planner}) is convex with respect to $(\gamma_1(\cdot), \gamma_2(\cdot))$.

\item To check the concavity of the functional in(\ref{eq:functionalrobust-planner}) with respect to $(C(\cdot), B(\cdot), R_a(\cdot), R_b(\cdot))$, we first remark that the concavity of $(g(R_b))^{\frac{1}{1-\theta_2}}$ implies that the function
\begin{equation}
\label{eq:jointly}
(R_b, B_2) \mapsto g(R_b) B_2^{\theta_2}
\end{equation}
is jointly concave\footnote{To check this property it is enough to check the signature of the Hessian matrix of (\ref{eq:jointly}) and to use the sign of the second derivative of  $(g(R_b))^{\frac{1}{1-\theta_2}}$.} and then the function $G$ is convex. The convexity of $G$ then implies the concavity of the functional. Using Corollary \ref{cor:rewriting-planner-max}, this can be rewritten as
\begin{multline}
(\gamma_{1}+\gamma_2)\left[
\frac{\overline{S}}{\rho} -
\frac{P(0)}{\rho}-\frac{T(0)}{\rho+\phi}
\right]+
\\[2mm]
\notag
\int_{0}^{\infty }e^{-\rho t}
\left[u_{1}\left(C_1(t)\right)+ u_{2}\left(C_2(t)\right)
-(\gamma_{1}+\gamma_2) \Phi G\left(t,K(t),B(t),R_b(t)\right)\right]dt.
\end{multline}
This concludes the proof.
\end{itemize}
\end{proof}

\begin{proof}[Proof of Proposition \ref{pr:robust-global-planner}]
We demonstrated that the initial max-min problem is equivalent to the min-max problem given by
\begin{multline}
\min_{\gamma _{1},\gamma _{2}} \max_{C_{i}(t),B_{i}(t), R_b(t)} U^{R}(\alpha) =U_{1}+U_{2}=
\notag \\
\int_{0}^{\infty }e^{-\rho t}\left[ u_{1}\left( C_{1}(t)\right)
-\gamma _{1}\left( S(t)-\overline{S}\right) +\alpha \left\vert \gamma
_{1}-\widehat{\gamma }_{1}\right\vert^2 \right] dt  \notag \\
+\int_{0}^{\infty }
e^{-\rho t}\left[ u_{2}\left( C_{2}(t)\right)
-\gamma _{2}\left( S(t)-\overline{S}\right) +\alpha \left\vert \gamma
_{2}-\widehat{\gamma }_{2}\right\vert^2 \right] dt\\
=
\min_{\gamma _{1},\gamma _{2}} \max_{C_{i}(t),B_{i}(t), R_b(t)}
\bigg \{
(\gamma_{1}+\gamma_2)\left[
\frac{\overline{S}}{\rho} -
\frac{P(0)}{\rho}-\frac{T(0)}{\rho+\phi}
\right]+ \frac{\alpha}{\rho} (\left\vert \gamma
_{1}-\widehat{\gamma }_{1}\right\vert^2 + \left\vert \gamma
_{2}-\widehat{\gamma }_{2}\right\vert^2) +
\\
\int_{0}^{\infty }e^{-\rho t}
\left[u_{1}\left(C_1(t)\right)+ u_{2}\left(C_2(t)\right)
-(\gamma_{1}+\gamma_2) \Phi G\left(t,K(t),B(t),R_b(t)\right)\right]dt
\bigg \}.
\end{multline}
Without loss of generality, we can restrict attention to the case $\gamma_1+\gamma_2>0$. In the logarithmic case, the previous expression can be written as:
\begin{multline}
\min_{\gamma _{1},\gamma _{2}} 
\bigg \{
(\gamma_{1}+\gamma_2)\left[
\frac{\overline{S}}{\rho} -
\frac{P(0)}{\rho}-\frac{T(0)}{\rho+\phi}
\right]+ \frac{\alpha}{\rho} (\left\vert \gamma
_{1}-\widehat{\gamma }_{1}\right\vert^2 + \left\vert \gamma
_{2}-\widehat{\gamma }_{2}\right\vert^2)+
\\
\int_{0}^{\infty }e^{-\rho t}
\bigg[ \ln \left (\frac{\overline{A}-1}{(\gamma_{1}+\gamma_2) \Phi \eta_K} \right ) + \ln \left (\frac{\overline{A}-1}{(\gamma_{1}+\gamma_2) \Phi \eta_K} \right ) \\
- (\gamma_{1}+\gamma_2) \Phi \eta_K \left ( R_b + B_1 + B_2 + \left (\frac{\overline{A}-1}{(\gamma_{1}+\gamma_2) \Phi \eta_K} \right ) + \left (\frac{\overline{A}-1}{(\gamma_{1}+\gamma_2) \Phi \eta_K} \right ) \right ) \\ 
+ (\gamma_{1}+\gamma_2) \Phi\eta_B \Big (B_1(t)^{\theta_1} + g(R_b) B_2(t)^{\theta_2} \Big ) dt \bigg \}.
\end{multline}
Evaluating the time integral, simplifying, and eliminating the terms that do not depend on $\gamma_i$, we have that $(\gamma_1, \gamma_2)$ is a solution of the previous minimization problem if and only if it is a solution to: 
\begin{multline}
\min_{\gamma _{1},\gamma _{2}} \Bigg \{ 
\frac{2}{\rho}\ln \left (\frac{\overline{A}-1}{(\gamma_{1}+\gamma_2) \Phi \eta_K} \right )
+(\gamma_{1}+\gamma_2) \Gamma_1
+ \frac{\alpha}{\rho} \Big(\left\vert \gamma
_{1}-\widehat{\gamma }_{1}\right\vert^2 + \left\vert \gamma
_{2}-\widehat{\gamma }_{2}\right\vert^2 \Big )\Bigg \}.
\end{multline}
The previous equation is coercive, convex, goes to $+\infty$ when $(\gamma_1+\gamma_2) \to 0$, and has exactly one minimum for $\gamma_1 + \gamma_2 > 0$.  At the minimum point the first-order conditions are necessary and sufficient:
\[
-\frac{2}{\rho}\frac{1}{\gamma_1 + \gamma_2} + \Gamma_1 + \frac{2\alpha}{\rho} (\gamma_1-\widehat{\gamma}_1) = 0 =-\frac{2}{\rho}\frac{1}{\gamma_1 + \gamma_2} + \Gamma_1 + \frac{2\alpha}{\rho} (\gamma_2-\widehat{\gamma}_2),
\]

Since $\gamma_2 = \gamma_1 + \widehat{\gamma}_2 - \widehat{\gamma}_1$, the previous equation becomes:
\[
0 = 2\alpha \gamma_1^2 + \gamma_1 \left [ \alpha(- 3\widehat{\gamma}_1 + \widehat{\gamma}_2) + \rho \Gamma_1 \right ] + (\widehat{\gamma}_2 -\widehat{\gamma}_1) \left [ \frac{\rho\Gamma_1}{2} - \alpha \widehat{\gamma}_1 \right ] -1. 
\]
Thus,
\begin{multline}
\gamma_1 = \frac{1}{4\alpha} \left [ - \alpha(- 3\widehat{\gamma}_1 + \widehat{\gamma}_2) - \rho \Gamma_1 \pm \sqrt{\left [ \alpha(- 3\widehat{\gamma}_1 + \widehat{\gamma}_2) + \rho \Gamma_1 \right ]^2 - 8 \alpha \left ((\widehat{\gamma}_2 -\widehat{\gamma}_1) \left [ \frac{\rho\Gamma_1}{2} - \alpha \widehat{\gamma}_1 \right ] -1  \right )}\right ] \\
= \frac{1}{4\alpha} \left [ - \alpha(- 3\widehat{\gamma}_1 + \widehat{\gamma}_2) - \rho \Gamma_1 \pm \sqrt{\left [\rho \Gamma_1 -\alpha(\widehat{\gamma}_1 + \widehat{\gamma}_2) \right ]^2 + 8 \alpha} \right ],
\end{multline}
and
\[
\gamma_2 = \frac{1}{4\alpha} \left [ - \alpha(- 3\widehat{\gamma}_2 + \widehat{\gamma}_1) - \rho \Gamma_1 \pm \sqrt{\left [\rho \Gamma_1 -\alpha(\widehat{\gamma}_1 + \widehat{\gamma}_2) \right ]^2 + 8 \alpha} \right ].
\]
The unique positive solution is then given by: 
\[
\gamma_1 = \frac{1}{4\alpha} \left [ - \alpha(- 3\widehat{\gamma}_1 + \widehat{\gamma}_2) - \rho \Gamma_1 + \sqrt{\left [\rho \Gamma_1 -\alpha(\widehat{\gamma}_1 + \widehat{\gamma}_2) \right ]^2 + 8 \alpha} \right ],
\]
and
\[
\gamma_2 = \frac{1}{4\alpha} \left [ - \alpha(- 3\widehat{\gamma}_2 + \widehat{\gamma}_1) - \rho \Gamma_1 + \sqrt{\left [\rho \Gamma_1 -\alpha(\widehat{\gamma}_1 + \widehat{\gamma}_2) \right ]^2 + 8 \alpha} \right ].
\]
The choices for the other variables follow as in Subsection \ref{sub:planner}. In particular,
\[
\gamma_1 + \gamma_2 = \frac{1}{2\alpha} \left [ \alpha( \widehat{\gamma}_2 + \widehat{\gamma}_1) - \rho \Gamma_1 + \sqrt{\left [\rho \Gamma_1 -\alpha(\widehat{\gamma}_1 + \widehat{\gamma}_2) \right ]^2 + 8 \alpha} \right ]>0.
\]

\end{proof}

\subsection{The RP robust control problem}
\label{supapp:robustrestricted}

The restricted robust social planner's problem can be written as follows: 
\begin{eqnarray}
&&\max_{C(\cdot),K(\cdot),B(\cdot),R_{a}(\cdot),R_{b}(\cdot)} \qquad 
\min_{\gamma _{1}(\cdot),\gamma _{2}(\cdot)}U^{R}(\alpha ) = U_1^R(\alpha)+U_2^R(\alpha)=
\notag
\\
&&\int_{0}^{\infty }e^{-\rho t}\left[ u_{1}\left( C_{1}(t)\right)
-\gamma _{1}(t)\left( S(t)-\overline{S}\right) +\alpha \left\vert \gamma
_{1}(t)-\widehat{\gamma }_{1}\right\vert^2 \right] dt  \notag 
\\
&&+\int_{0}^{\infty }
e^{-\rho t}\left[ u_{2}\left( C_{2}(t)\right)
-\gamma _{2}(t)\left( S(t)-\overline{S}\right) +\alpha \left\vert \gamma
_{2}(t)-\widehat{\gamma }_{2}\right\vert^2 \right] dt,
\end{eqnarray}
subject to the resource constraints:
\begin{equation}
\label{eq:planner-notrasnfer-resourceconstraint}
\left\{
\begin{array}{l}
C_1(t)+ B_1(t)+ K_1(t)+ R_a(t)+ R_b(t)= Y_1(t)= \overline{A}(t)f_1(K_1(t))\\ 
C_2(t)+ B_2(t)+ K_2(t) = Y_2(t) =
\overline{A}(t)h(R_a(t))f_2(K_2(t)).
\end{array}
\right.
\end{equation}
The control variables belong to the set $\mathcal{U}^{RP}$ defined below 
\eqref{eq:planner-notrasnfer-resourceconstraintmain}, while the strategies of the malevolent player
$(\gamma _{1}(\cdot),\gamma _{2}(\cdot))$
are assumed to belong to $L^2_\rho(\R_+;\R^2)$.

Similarly to Theorem \ref{th:maxmingeneral} in the GP set-up, we have here the following minimax theorem.

\begin{Theorem}
\label{th:maxmingeneralrestricted}
Under the above assumptions we have:
$$
\max_{C(\cdot),K(\cdot),B(\cdot),R_{a}(\cdot),R_{b}(\cdot)} \qquad 
\min_{\gamma _{1}(\cdot),\gamma _{2}(\cdot)}U^{R}(\alpha) 
=
$$
$$
\min_{\gamma _{1}(\cdot),\gamma _{2}(\cdot)} \qquad
\max_{C(\cdot),K(\cdot),B(\cdot),R_{a}(\cdot),R_{b}(\cdot)} 
U^{R}(\alpha) 
$$
\end{Theorem}

\subsubsection{Logarithmic payoffs}
\label{subsub:special-robust-RP}
As in the case of the GP we study the special case where Assumption \ref{hp:easycase} holds, the payoffs are logarithmic, and $\gamma_1$ and $\gamma_2$ are real constants. Using Theorem \ref{th:maxmingeneralrestricted}, and again restrict attention to the case where $\gamma_1 + \gamma_2 >0$, we need to solve the following: 
\begin{multline}
\label{eq:problem-robust-planner2}
\min_{\gamma _{1},\gamma _{2} \in \mathbb{R}} \; \max_{C(\cdot),B(\cdot), K(\cdot), R(\cdot) \in \mathcal{U}^{RP}} U^{R}(\alpha ) = \min_{\gamma _{1},\gamma _{2} \in \mathbb{R}} \; \max_{C(\cdot),B(\cdot), K(\cdot), R(\cdot) \in \mathcal{U}^{p1}} U_{1}+U_{2}=\\
\min_{\gamma _{1},\gamma _{2} \in \mathbb{R}} \; \max_{C(\cdot),B(\cdot), K(\cdot), R(\cdot) \in \mathcal{U}^{RP}} \int_{0}^{\infty }e^{-\rho t}\Big[ \ln\left( C_{1}(t)\right) + \ln\left( C_{2}(t) \right )
-(\gamma_{1}+\gamma_{2})\left( S(t)-\overline{S}\right) \\ +\alpha \left\vert \gamma
_{1}-\widehat{\gamma }_{1}\right\vert^2 +\alpha \left\vert \gamma
_{2}-\widehat{\gamma }_{2}\right\vert^2 \Big] dt. 
\end{multline}
We have the following result.

\begin{Proposition}
\label{pr:robust-restricted-planner}
Suppose Hypothesis \ref{hp:easycase} holds and assume logarithmic payoffs. Given $\hat{\gamma}_1, \hat{\gamma}_2 >0$ the values of $(\gamma_1, \gamma_2)$ that solve (\ref{eq:problem-robust-planner2}) are given by
\begin{multline}
(\gamma _{1},\gamma _{2}) \in \arg\min \Bigg \{ 
\frac{1}{\rho}\ln \left (\frac{\overline{A}-1}{(\gamma_{1}+\gamma_2) \Phi \eta_K} \frac{h(R_a(\gamma_1, \gamma_2))\overline{A}-1}{(\gamma_{1}+\gamma_2) \Phi \eta_K} \right )
+(\gamma_{1}+\gamma_2) \Lambda_1(\gamma_1, \gamma_2)\\
+  h(R_a(\gamma_1, \gamma_2))\overline{A}
+ \frac{\alpha}{\rho} \Big(\left\vert \gamma
_{1}-\widehat{\gamma }_{1}\right\vert^2 + \left\vert \gamma
_{2}-\widehat{\gamma }_{2}\right\vert^2 \Big )\Bigg \},
\end{multline}
where
\begin{multline}
\label{eq:defGamma1-robust-RP}
\Lambda_1(\gamma_1, \gamma_2) := \left[ \frac{\overline{S}}{\rho} - \frac{P(0)}{\rho} - \frac{T(0)}{\rho+\phi} \right] - \frac{\Phi \eta_K}{\rho} (R_a(\gamma_1, \gamma_2) + R_b(\gamma_1, \gamma_2) + B_1 + B_2) \\
+ \frac{\Phi \eta_B}{\rho}  \Big (B_1(t)^{\theta_1} + g(R_b(\gamma_1, \gamma_2)) B_2(t)^{\theta_2} \Big ),
\end{multline}
and $\Big ( R_a(\gamma_1, \gamma_2), R_b(\gamma_1, \gamma_2) \Big)$ is any pair $(R_a, R_b)$ solving the maximization problem (\ref{eq:Ra-Rb-RP-log}). The values of $R, C, B,$ and $K$ are given by the corresponding resource constraint. 
\end{Proposition}
\begin{proof}
Using the equivalence between max-min and min-max problems given in Theorem \ref{th:maxmingeneralrestricted} and Corollary \ref{cor:rewriting-planner-notransfer}, we obtain:
\begin{multline}
\min_{\gamma _{1},\gamma _{2}} \max_{C_{i}(t),B_{i}(t), R_b(t)} U^{R}(\alpha ) = U_{1} + U_{2} = \\
\int_{0}^{\infty }e^{-\rho t}\left[ u_{1}\left( C_{1}(t)\right) -\gamma _{1}\left( S(t)-\overline{S}\right) +\alpha \left\vert \gamma_{1}-\widehat{\gamma }_{1}\right\vert^2 \right] dt  \\
+\int_{0}^{\infty}e^{-\rho t}\left[ u_{2}\left( C_{2}(t)\right) -\gamma _{2}\left( S(t)-\overline{S}\right) +\alpha \left\vert \gamma_{2}-\widehat{\gamma }_{2}\right\vert^2 \right] dt = \\
\min_{\gamma _{1},\gamma _{2}} \max_{C_{i}(t),B_{i}(t), R_b(t)} \bigg \{ (\gamma_{1}+\gamma_2)\left[ \frac{\overline{S}}{\rho} - \frac{P(0)}{\rho} - \frac{T(0)}{\rho+\phi} \right] \\
+ \frac{\alpha}{\rho} (\left\vert \gamma_{1}-\widehat{\gamma }_{1}\right\vert^2 + \left\vert \gamma_{2}-\widehat{\gamma }_{2}\right\vert^2) + \int_{0}^{\infty }e^{-\rho t} \left[u_{1}\left(C_1(t)\right) + u_{2}\left(C_2(t)\right) \right. \\
\left. -(\gamma_{1}+\gamma_2) \Phi G\left(t,K(t),B(t),R_b(t)\right)\right]dt \bigg \}.
\end{multline}
Using the expression for $\Phi$ defined in (\ref{eq:defPhi}) and the expressions for $C_i$ derived earlier, we obtain:
\begin{multline}
\min_{\gamma _{1},\gamma _{2}} 
\bigg \{
(\gamma_{1}+\gamma_2)\left[
\frac{\overline{S}}{\rho} -
\frac{P(0)}{\rho}-\frac{T(0)}{\rho+\phi}
\right]+ \frac{\alpha}{\rho} (\left\vert \gamma
_{1}-\widehat{\gamma }_{1}\right\vert^2 + \left\vert \gamma
_{2}-\widehat{\gamma }_{2}\right\vert^2)+
\\
\int_{0}^{\infty }e^{-\rho t}
\bigg[ \ln \left (\frac{\overline{A}-1}{(\gamma_{1}+\gamma_2) \Phi \eta_K} \right ) + \ln \left (\frac{h(R_a(\gamma_1, \gamma_2))\overline{A}-1}{(\gamma_{1}+\gamma_2) \Phi \eta_K} \right ) \\
- (\gamma_{1}+\gamma_2) \Phi \eta_K \left ( R_a(\gamma_1, \gamma_2) + R_b(\gamma_1, \gamma_2) + B_1 + B_2 + \left (\frac{\overline{A}-1}{(\gamma_{1}+\gamma_2) \Phi \eta_K} \right ) + \left (\frac{h(R_a(\gamma_1, \gamma_2))\overline{A}-1}{(\gamma_{1}+\gamma_2) \Phi \eta_K} \right ) \right ) \\ 
+ (\gamma_{1}+\gamma_2) \Phi\eta_B \Big (B_1(t)^{\theta_1} + g(R_b(\gamma_1, \gamma_2)) B_2(t)^{\theta_2} \Big ) dt \bigg \}.
\end{multline}
Evaluating the time integral, simplifying, and eliminating the terms which do not depends of $\gamma_i$, we obtain that $(\gamma_1, \gamma_2)$ is a solution of the previous minimization if and only if it is a solution to: 
\begin{multline}
\min_{\gamma _{1},\gamma _{2}} \Bigg \{ 
\frac{1}{\rho}\ln \left (\frac{\overline{A}-1}{(\gamma_{1}+\gamma_2) \Phi \eta_K} \frac{h(R_a(\gamma_1, \gamma_2))\overline{A}-1}{(\gamma_{1}+\gamma_2) \Phi \eta_K} \right )
+(\gamma_{1}+\gamma_2) \Lambda_1(\gamma_1, \gamma_2)\\
+  h(R_a(\gamma_1, \gamma_2))\overline{A}
+ \frac{\alpha}{\rho} \Big(\left\vert \gamma
_{1}-\widehat{\gamma }_{1}\right\vert^2 + \left\vert \gamma
_{2}-\widehat{\gamma }_{2}\right\vert^2 \Big )\Bigg \},
\end{multline}
where $\Lambda_1(\gamma_1, \gamma_2)$ is defined in (\ref{eq:defGamma1-robust-RP}) and $\Big ( R_a(\gamma_1, \gamma_2), R_b(\gamma_1, \gamma_2) \Big)$ is determined by solving the maximization problem (\ref{eq:Ra-Rb-RP-log}). The previous equation is corcive and goes to infinity when $(\gamma_1 + \gamma_2)$ goes to zero. So it has a minimum, and for this minimum point we have $\gamma_1+ \gamma_2 >0$ (observe that $R_a$ and $R_b$ are well defined for $\gamma_1+ \gamma_2 >0$). 
\end{proof}

As in the previous case, these values are generically unique. Due to the fact that $R_a$ and $R_b$ depend here on $\gamma_i$, we cannot obtain a closed expression as in Subsection \ref{subsub:special-robust-GP}. Still, the previous expression allows us to reduce the problem to a finite-dimensional minimization problem that one can treat numerically.

\subsection{The Nash robust control problem}
\label{supapp:robustNash}

In the Nash problem, each country takes as given the (robust) strategy of the
other country when choosing their best response. More precisely, we will consider one malevolent player for
each country.\footnote{This follows the ``soft constraint" formulation in \cite{vandenBroek2003}.} Formally, given $C_{2}(t)$, $B_{2}(t)$, $%
t\in \lbrack 0,\infty )$, country $1$ solves:

\begin{equation}
\max_{C_{1}(t),B_{1}(t), R_a(t), R_b(t)}\min_{\gamma _{1}(t)}U_{1}^{R}(\alpha)=\int_{0}^{\infty }e^{-\rho t}\left[ u_{1}\left( C_{1}(t)\right)
-\gamma _{1}(t)\left( S(t)-\overline{S}\right) +\alpha \left\vert \gamma
_{1}(t)-\widehat{\gamma }_{1}\right\vert^2 \right] dt
\end{equation}
subject to its feasibility constraint. 

Similarly, given $C_{1}(t)$, $B_{1}(t)$, $R_a(t)$, $R_b(t)$; $t\in \lbrack 0,\infty )$, country $2$ solves:
\begin{equation}
\max_{C_{2}(t),B_{2}(t)}\min_{\gamma _{2}(t)}U_{2}^{R}(\alpha)=\int_{0}^{\infty }e^{-\rho t}\left[ u_{2}\left( C_{2}(t)\right)
-\gamma _{2}(t)\left( S(t)-\overline{S}\right) +\alpha \left\vert \gamma
_{2}(t)-\widehat{\gamma }_{2}\right\vert^2 \right] dt
\end{equation}%
subject to its feasibility constraint.

\label{rm:gi>0}

Once again, in what follows we will restrict attention to the case of logarithmic payoffs and a constant $\gamma$. It can be shown that the minimax Theorem we used in the analysis of the planners' problems also holds in this case. Proceeding as in Theorem  \ref{th:maxmingeneral}, we can again exchange the max-min with the min-max in the previous problem for Country 1 to obtain:
\begin{multline}
\label{eq:functional-Nash-Rubust-Country1}
\min_{\gamma _{1}} \max_{C_{1}(t),B_{1}(t), R_a(t), R_b(t)} U^{R}_1(\alpha) =
\int_{0}^{\infty }e^{-\rho t}\left[ \ln\left( C_{1}(t)\right) -\gamma _{1}\left( S(t)-\overline{S}\right) +\alpha \left\vert \gamma_{1}-\widehat{\gamma }_{1}\right\vert^2 \right] dt  \\
= \min_{\gamma _{1}} \max_{C_{i}(t),B_{i}(t), R_a(t), R_b(t)} \bigg \{ \gamma_{1}\left[ \frac{\overline{S}}{\rho} - \frac{P(0)}{\rho} - \frac{T(0)}{\rho+\phi} \right] \\
+ \frac{\alpha}{\rho} (\left\vert \gamma_{1}-\widehat{\gamma }_{1}\right\vert^2 ) + \int_{0}^{\infty }e^{-\rho t} \left[\ln\left(C_1(t)\right) -\gamma_{1} G\left(t,K(t),B(t),R_b(t)\right)\right]dt \bigg \}.
\end{multline}
Similarly, Country 2 solves: 
\begin{multline}
\label{eq:functional-Nash-Rubust-Country2}
\min_{\gamma _{2}} \max_{C_{2}(t),B_{2}(t)} \bigg \{ \gamma_{2}\left[ \frac{\overline{S}}{\rho} - \frac{P(0)}{\rho} - \frac{T(0)}{\rho+\phi} \right] \\
+ \frac{\alpha}{\rho} (\left\vert \gamma_{2}-\widehat{\gamma }_{2}\right\vert^2 ) + \int_{0}^{\infty }e^{-\rho t} \left[\ln\left(C_2(t)\right) -\gamma_{2} G\left(t,K(t),B(t),R_b(t)\right)\right]dt \bigg \}.
\end{multline}

We then have the following.

\begin{Proposition}
\label{pr:robust-Nash}
Suppose Assumption \ref{hp:easycase} holds, payoffs are logarithmic, and $\hat{\gamma}_1, \hat{\gamma}_2 >0$.
Suppose that $(\gamma_1, \gamma_2)>0$ solve the following system:
\[
\left \{
\begin{array}{l}
G_1(\gamma_2)+\frac{2}{\rho}\alpha( \gamma_{1}-\widehat{\gamma }_{1})- \frac{1}{\rho} \frac{1}{\gamma_{1}}=0\\[8pt]
G_2(\gamma_1)+\frac{2}{\rho}\alpha( \gamma_{2}-\widehat{\gamma }_{2})- \frac{1}{\rho} \frac{1}{\gamma_{2}}=0
\end{array}
\right.
\]
where the expressions for $G_1(\gamma_2)$ and $G_2(\gamma_1)$ are given in Appendix B. Then $(\gamma_1, \gamma_2)$ form part of a solution to the problem in (\ref{eq:functional-Nash-Rubust-Country1})-(\ref{eq:functional-Nash-Rubust-Country2}).


\end{Proposition}
\begin{proof}
Observe that the choices of $B_i$, $R_a=0$ and $R_b$ do not depend on $\gamma_1$ and on $\gamma_2$, while the values of $C_2$ and $K_2$ are independent of $\gamma_1$.  In the logarithmic case, the expression (\ref{eq:functional-Nash-Rubust-Country1}) simplifies and the problem of Country 1 reduces to:
\begin{multline}
\label{eq:min-in-gamma1}
\min_{\gamma _{1}} \bigg \{ \gamma_{1}\left[ \frac{\overline{S}}{\rho} - \frac{P(0)}{\rho} - \frac{T(0)}{\rho+\phi} - \frac{\eta_K\Phi}{\rho}\left ( \frac{1}{\overline{A}-1}(B_1  + R_b)+K_2(\gamma_2) \right) + \frac{\eta_B\Phi}{\rho} (B_1(t)^{\theta_1} + g(R_b) B_2(t)^{\theta_2}) \right] \\
+ \frac{\alpha}{\rho} (\left\vert \gamma_{1}-\widehat{\gamma }_{1}\right\vert^2 ) + \frac{1}{\rho}\ln \left ( \frac{\overline{A}-1}{\gamma_{1}\Phi \eta_K} \right ) - \frac{1}{\rho} \bigg \}.
\end{multline}
Country 2, in turn, solves:
\begin{multline}
\label{eq:min-in-gamma2}
\min_{\gamma _{2}} \bigg \{ \gamma_{2}\left[ \frac{\overline{S}}{\rho} - \frac{P(0)}{\rho} - \frac{T(0)}{\rho+\phi} - \frac{\eta_K\Phi}{\rho}\left (K_1(\gamma_1) +  \frac{1}{\overline{A}h(0)-1}B_2 \right) + \frac{\eta_B\Phi}{\rho} (B_1(t)^{\theta_1} + g(R_b) B_2(t)^{\theta_2}) \right] \\
+ \frac{\alpha}{\rho} (\left\vert \gamma_{2}-\widehat{\gamma }_{2}\right\vert^2 ) + \frac{1}{\rho}\ln \left ( \frac{\overline{A}h(0)-1}{\gamma_{2}\Phi \eta_K} \right ) - \frac{1}{\rho} \bigg \}.
\end{multline}
The expression to be minimized in (\ref{eq:min-in-gamma1}) is convex in $\gamma_1\in (0, +\infty)$, it is coercive, and it goes to $+\infty$ when $\gamma_1\to 0^+$. Thus, for any fixed $\gamma_2$, the point of minimum in $\gamma_1$ is unique, and similarly for the the expression (\ref{eq:min-in-gamma2}). The two first order conditions are necessary and sufficient and imply:
\[
G_1(\gamma_2)+\frac{2}{\rho}\alpha( \gamma_{1}-\widehat{\gamma }_{1})- \frac{1}{\rho} \frac{1}{\gamma_{1}}=0,
\]
\[
G_2(\gamma_1)+\frac{2}{\rho}\alpha( \gamma_{2}-\widehat{\gamma }_{2})- \frac{1}{\rho} \frac{1}{\gamma_{2}}=0,
\]
where we denoted
\[
G_1(\gamma_2):= \left[ \frac{\overline{S}}{\rho} - \frac{P(0)}{\rho} - \frac{T(0)}{\rho+\phi} - \frac{\eta_K\Phi}{\rho}\left ( \frac{1}{\overline{A}-1}(B_1  + R_b)+K_2(\gamma_2) \right) + \frac{\eta_B\Phi}{\rho} (B_1(t)^{\theta_1} + g(R_b) B_2(t)^{\theta_2}) \right],
\]
and
\[
G_2(\gamma_1):=\left[ \frac{\overline{S}}{\rho} - \frac{P(0)}{\rho} - \frac{T(0)}{\rho+\phi} - \frac{\eta_K\Phi}{\rho}\left (K_1(\gamma_1) +  \frac{1}{\overline{A}h(0)-1}B_2 \right) + \frac{\eta_B\Phi}{\rho} (B_1(t)^{\theta_1} + g(R_b) B_2(t)^{\theta_2}) \right],
\]
where
\[
K_1(\gamma_1) = \frac{1}{\overline{A} -1}(C_1^{N} + B_1^{N} + R_b^N) 
=
\frac{1}{\overline{A} -1} \left ( \left(\frac{\overline{A} -1}{\gamma_{1}\Phi \eta_K}\right)  + B_1^{N} + R_b^N \right ),
\]
where $B_1^N$ and $R_b^N$ are given in (\ref{eq:B1B2N-log}) and (\ref{eq:RaRbN-log}) 
and 
\[
K_2(\gamma_2) = \frac{1}{\overline{A}h(0) -1}(C_2^{N} + B_2^{N})
=
\frac{1}{\overline{A}h(0) -1} \left ( \left(\frac{\overline{A}h(0) -1}{\gamma_{2}\Phi \eta_K}\right) + B_2^{N} \right ),
\]
where $B_2^N$ is are given in (\ref{eq:B1B2N-log}).
This concludes the proof.

\end{proof}

\newpage

\section{The numerical example}
\label{app:Tables}

For our numerical example, we used the parameter values in the following:
\begin{table}[!htpb]
\begin{tabular}{rll}
\textbf{Parameter} & \textbf{Used in}                    & \textbf{Value}                            \\ \hline
$\rho$             & Discounting factor, Eq. (1)         & $-\log\left((0.96)^{10}\right)$           \\
$\alpha$           & Production, Eq. (16) + $\alpha$     & 1                                         \\
$\overline{A}$     & Technology, Eq. (16)                & 10                                        \\
$\sigma_1$         & Utility rich, Assn 3 (iv)            & $1$ (log) or $1.2$                        \\
$\sigma_2$         & Utility poor, Assn 3 (iv)            & $1$ (log) or $1.2$                        \\
$\gamma_1$         & damages from emission rich, Eq. (1) & $0.0125$ or $0.0075$                      \\
$\gamma_2$         & damages from emission poor, Eq. (1) & $0.0125$ or $0.0075$                      \\
$\phi$             & climate, Eq. (6)                    & 0.5                                       \\
$\phi_L$           & climate, Eq. (5)                    & 0.2                                       \\
$\phi_0$           & climate, Eq. (6)                    & 0.393                                     \\
$\theta_1$         & abatement technology, Assn 3 (i)     & 0.5                                       \\
$\theta_2$         & abatement technology, Assn 3 (i)     & 0.5                                       \\
$\Lambda_I$        & abatement technology, Assn 3 (i)     & 1                                         \\
$\Lambda_D$        & abatement technology, Assn 3 (i)     & 1                                         \\
$g(x)$             & abatement technology, Assn 3 (i),(v) & $g(x) = \frac{g(\infty) x + g(0)}{x+1}$   \\
$g(0)$             & abatement technology, Assn 3 (i),(v) & 0.2                                       \\
$g(\infty)$        & abatement technology, Assn 3 (i),(v) & 0.5                                       \\
$h(x)$             & abatement technology, Assn 3 (i),(v) & $h(x) = \frac{h(\infty) x + h(0)}{x+1}$   \\
$h(0)$             & abatement technology, Assn 3 (i),(v) & 0.5                                       \\
$h(\infty)$        & abatement technology, Assn 3 (i),(v) & 0.9                                       \\
\end{tabular}
\caption{Parameters used in the static optimization problems.}
\end{table}

\section{Robust Randomization}
\label{app:random}

Finally, we consider the case where $\hat{\gamma}_1$ and $\hat{\gamma}_2$ in the approximate model are drawn from anexponential distribution with parameters $1/\hat{\gamma}_1$ and $1/\hat{\gamma}$, respectively. The exponential distribution is “fat tailed,” implying a higher probability of extreme values than a normal distribution. We consider a fixed value for the parameter $\alpha$. We then proceed according to the following sequence for each $n = 1,\dots, N$, where $N$ is the number of randomizations:

\begin{enumerate}
\item Pick $\hat{\gamma}_1(\omega)$ and $\hat{\gamma}_2(\omega)$ from an exponential distribution with parameter $1/\hat{\gamma}_1$ and $1/\hat{\gamma}_2$, respectively.
\item The malevolent players choose: $\gamma_1(\omega)$ and $\gamma_2(\omega)$ 
\item For each regime (GP, RP, Nash) the decision-makers choose their choice variables depending on $\omega$. 
\item For each $n$ we compute the resulting global temperatures depending on  $\gamma_1(\omega)$ and $\gamma_2(\omega)$.
\item Once this is done, for each regime and for every temperature profile, we select the top $97.5\%$ and the bottom $2.5\%$ of temperature generated trajectories, as $\omega$ varies. This generates the upper and the lower confidence intervals seen in the Figure below, indicating a stark contrast between the non-cooperative solution and the two efficiency benchmarks. 
\end{enumerate}

\begin{figure}[!htpb]
\includegraphics[width=0.49\textwidth]{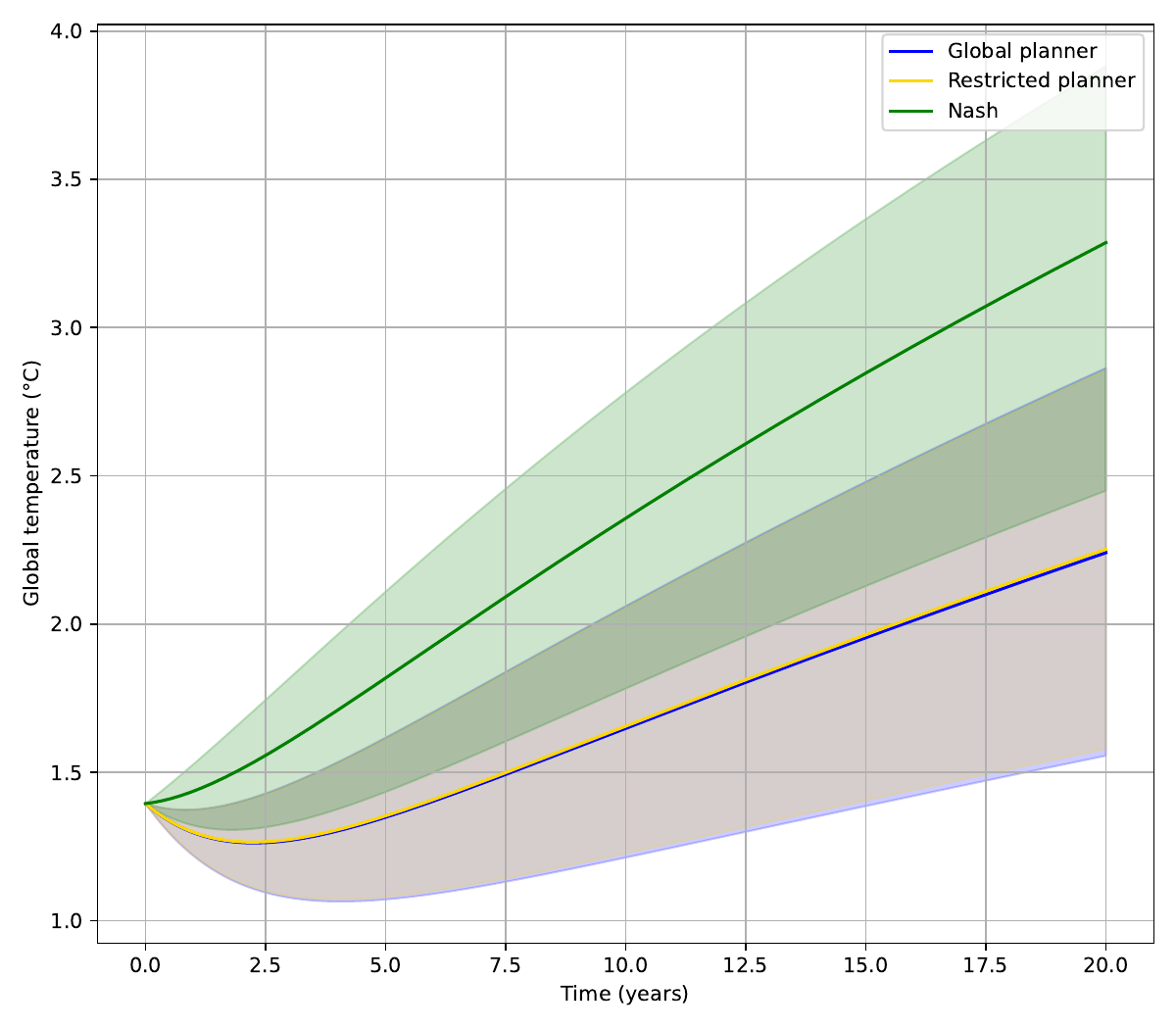}
\caption{Randomization of $\hat{\gamma}$ }
\label{fig:temperatureRandomizationGamma}
\end{figure}

\end{small}

\end{document}